\documentclass[11pt]{article}

\usepackage[usenames,dvipsnames]{xcolor}
\definecolor{Gred}{RGB}{219, 50, 54}
\definecolor{ToCgreen}{RGB}{0, 128, 0}

\usepackage[margin=1in]{geometry}
\usepackage{palatino}

\usepackage{makecell}
\usepackage{dsfont}

\usepackage{multirow}
\usepackage{amssymb,amsmath,amsthm,amsfonts}
\usepackage{bm}
\usepackage{bbm}
\usepackage{textgreek}
\usepackage{mathtools}
\usepackage{enumitem}
\usepackage[numbers,comma,sort&compress]{natbib}
\usepackage{authblk}
\usepackage{graphicx}
\usepackage[font=small]{caption}
\usepackage[labelformat=simple]{subcaption}

\usepackage{float}
\usepackage[ruled,vlined,linesnumbered]{algorithm2e}
\usepackage{algorithmic}

\usepackage{physics}
\usepackage{footnote}
\usepackage{xcolor}
\usepackage{mathrsfs}
\usepackage{bbm}

\usepackage[colorlinks]{hyperref}
\usepackage{cleveref}
\hypersetup{
      colorlinks=true,
  citecolor=ToCgreen,
  linkcolor=Sepia,
  filecolor=Gred,
  urlcolor=Gred
  }

\newtheorem{theorem}{Theorem}
\newtheorem{definition}{Definition}
\newtheorem{lemma}{Lemma}

\newtheorem{example}{Example}

\newtheorem{remark}{Remark}
\newtheorem{problem}{Problem}

\newcommand{\thm}[1]{\hyperref[thm:#1]{Theorem~\ref*{thm:#1}}}
\newcommand{\cor}[1]{\hyperref[cor:#1]{Corollary~\ref*{cor:#1}}}
\newcommand{\defn}[1]{\hyperref[defn:#1]{Definition~\ref*{defn:#1}}}
\newcommand{\lem}[1]{\hyperref[lem:#1]{Lemma~\ref*{lem:#1}}}
\newcommand{\prop}[1]{\hyperref[prop:#1]{Proposition~\ref*{prop:#1}}}
\newcommand{\assum}[1]{\hyperref[assum:#1]{Assumption~\ref*{assum:#1}}}
\newcommand{\fig}[1]{\hyperref[fig:#1]{Figure~\ref*{fig:#1}}}
\newcommand{\tab}[1]{\hyperref[tab:#1]{Table~\ref*{tab:#1}}}
\newcommand{\algo}[1]{\hyperref[algo:#1]{Algorithm~\ref*{algo:#1}}}
\renewcommand{\sec}[1]{\hyperref[sec:#1]{Section~\ref*{sec:#1}}}
\newcommand{\append}[1]{\hyperref[append:#1]{Appendix~\ref*{append:#1}}}
\newcommand{\fac}[1]{\hyperref[fac:#1]{Fact~\ref*{fac:#1}}}
\newcommand{\lin}[1]{\hyperref[lin:#1]{Line~\ref*{lin:#1}}}

\def\>{\rangle}
\def\<{\langle}

\newcommand{\vect}[1]{\ensuremath{\mathbf{#1}}}

\newcommand{\Z}{\mathbb{Z}}

\newcommand{\C}{\mathbb{C}}

\newcommand{\E}{\mathbb{E}}
\newcommand{\cM}{\mathcal{M}}

\newcommand{\cT}{\mathcal{T}}
\newcommand{\cC}{\mathcal{C}}

\DeclareMathOperator{\poly}{poly}

\newcommand{\z}{\vect{z}}

\def\Tr{\operatorname{Tr}}\def\:{\hbox{\bf:}}

\newcommand{\blambda}{{\bm{\lambda}}}

\newcommand{\Lambdadep}{{\Lambda_{\text{dep}}}}

\let\oldnl\nl
\newcommand{\nonl}{\renewcommand{\nl}{\let\nl\oldnl}}

\renewcommand{\epsilon}{\varepsilon}

\begin{document}


\title{Efficient Pauli channel estimation with logarithmic quantum memory}

\author{Sitan Chen
\thanks{SEAS, Harvard University. Email: \href{mailto:sitan@seas.harvard.edu}{sitan@seas.harvard.edu}. This work was supported in part by NSF Award 2430375.}
\qquad
Weiyuan Gong
\thanks{SEAS, Harvard University; IIIS, Tsinghua University. Email: \href{mailto:wgong@g.harvard.edu}{wgong@g.harvard.edu}.}
}

\date{}
\maketitle

\begin{abstract}
Here we revisit one of the prototypical tasks for characterizing the structure of noise in quantum devices: estimating every eigenvalue of an $n$-qubit Pauli noise channel to error $\epsilon$. Prior work~\cite{chen2022quantum} proved no-go theorems for this task in the practical regime where one has a limited amount of quantum memory, e.g. any protocol with $\le 0.99n$ ancilla qubits of quantum memory must make exponentially many measurements, provided it is \emph{non-concatenating}. Such protocols can only interact with the channel by repeatedly preparing a state, passing it through the channel, and measuring immediately afterward. 

This left open a natural question: does the lower bound hold even for general protocols, i.e. ones which chain together many queries to the channel, interleaved with arbitrary data-processing channels, before measuring? Surprisingly, in this work we show the opposite: there is a protocol that can estimate the eigenvalues of a Pauli channel to error $\epsilon$ using only $O(\log n/\epsilon^2)$ ancilla qubits  and $\tilde{O}(n^2/\epsilon^2)$ measurements. In contrast, we show that any protocol with zero ancilla \--- even a concatenating one\--- must make $\Omega(2^n/\epsilon^2)$ measurements, which is tight.

Our results imply, to our knowledge, the first quantum learning task where logarithmically many qubits of quantum memory suffice for an exponential statistical advantage.
\end{abstract}

\clearpage

\tableofcontents

\clearpage

\section{Introduction}

One of the key challenges in demonstrating quantum advantage on near-term devices is that these devices are inherently noisy, from state preparation to gate application to measurement. To mitigate the effects of noise, it is essential to \emph{characterize} it. In this paper, we study this question in the context of \emph{Pauli channels}~\cite{watrous2018theory}, which constitute a widely studied family of noise models and encompass a range of natural channels like dephasing, depolarizing, and bit flips. Importantly, essentially any physically relevant noise model can be reduced to Pauli noise via twirling, making Pauli channel estimation a crucial sub-routine in numerous protocols for quantum benchmarking~\cite{wallman2016noise,erhard2019characterizing,magesan2011scalable,harper2020efficient,liu2021benchmarking,eisert2020quantum}. The central role it plays in both theory and practice has led to a long line of work on the complexity of this task~\cite{fujiwara2003quantum,hayashi2010quantum,chiuri2011experimental,ruppert2012optimal,collins2013mixed,flammia2020efficient,harper2021fast,flammia2021pauli,mohseni2008quantum,fawzi2023lower,chen2023learnability}.

Recently, there has been a surge of interest in understanding this and related quantum learning problems in settings where the quantum algorithms one is allowed to run are limited by the constraints faced by contemporary quantum devices~~\cite{bubeck2020entanglement,chen2022toward,chen2022tight,chen2022tight2,chen2022exponential,chen2021hierarchy,aharonov2022quantum,fawzi2023quantum,huang2021information,huang2022quantum,fawzi2023lower,liu2023memory}. In such settings, a number of works have shown nontrivial separations between the statistical rates one can achieve with and without such constraints. For instance, a series of works~\cite{huang2021information,chen2022exponential,huang2022quantum,chen2021hierarchy} showed that for certain natural state learning tasks, $n$ qubits of quantum memory are sufficient to achieve exponential speedups in the amount of quantum data needed for these tasks. Interestingly, it was also shown in~\cite{chen2022exponential} that $n$ qubits of quantum memory are \emph{necessary} for such speedups, in the sense that even with $0.99n$ qubits, the sample complexity for these tasks is exponential.

Following in the wake of these results, we study Pauli channel estimation in the same setting where one has bounded quantum memory~\cite{chen2022quantum,fawzi2023lower,flammia2020efficient,flammia2021pauli}. As we will see, while prior work has established a qualitatively similar picture to what is known in the state learning setting, the results in the present work will demonstrate a crucial difference between the two settings.

The starting point for our results is a recent paper by Chen et al.~\cite{chen2022quantum} that proved a number of upper and lower bounds in this regime. They showed that any protocol for estimating the eigenvalues of a Pauli channel using only $k$ ancillary qubits of quantum memory must make at least $\Omega(2^{(n-k)/3})$ measurements, provided the protocol is ``non-concatenating.'' Roughly speaking, this means that the algorithm is forced to perform a measurement immediately after every channel query (see \Cref{fig:model}). They also proved a qualitatively matching upper bound by giving a simple non-concatenating protocol that only makes $O(2^{n-k})$ measurements, suggesting that with a moderate amount of quantum memory, i.e. $(1 - o(1))n$ qubits, significant speedups are possible. Conversely, with a limited amount of quantum memory, e.g. $0.99$ qubits, it seems that exponentially many measurements are necessary, in close analogy to the aforementioned no-go results for state learning.

Intuitively however, the non-concatenating setting seems unnecessarily restrictive. One might wonder whether more general ``concatenating'' protocols could perform much better by querying the channel multiple times, with quantum post-processing between each query, before making a measurement (see \Cref{sec:problem} for formal definitions for this model). When $k = 0$, under this more general setting, Chen et al.~\cite{chen2022quantum} showed that $2^{n/3}$ measurements are still necessary.

Taken together, these two impossibility results imply that neither the ability to concatenate, nor the presence of a limited but nonzero amount of quantum memory (e.g. $k\le 0.99n$ qubits), is sufficient \emph{on its own} to achieve sample-efficient Pauli channel estimation. Indeed, there are various reasons to suspect that even protocols with both features suffer from exponential scaling:

\begin{enumerate}[leftmargin=*,itemsep=0pt,label=(\Alph*)]
\item \textbf{Existing separations for learning with quantum memory.} The only known separations (to our knowledge) between protocols with a small amount of quantum memory and protocols with zero quantum memory for any natural quantum learning task require enough qubits of memory to perform a Bell measurement or a swap test that is entangled with essentially all of the $n$ system qubits. Indeed, the upper bound in~\cite{chen2022quantum} falls under this umbrella, as do the aforementioned recent works on state learning~\cite{huang2022quantum, huang2021information}, as well as results on quantum process learning~\cite{caro2022learning} and purity testing~\cite{chen2022exponential,aharonov2022quantum}. Unfortunately, separations of this flavor are necessarily bottlenecked at requiring $k \ge (1 - o(1))n$ qubits of quantum memory.
    
\item \textbf{Concatenation and noise accumulation.} It is unclear whether being able to repeatedly query a Pauli channel before measurement buys us much. For example, if the channel has a spectral gap, then after a small number of repeated queries, the overall noise that the system experiences in the $n$ qubits on which the channel acts is close to completely depolarizing. This suggests that any efficient protocol for Pauli channel estimation should restrict to strategies that use limited or no concatenation.
\end{enumerate}
    
\noindent We therefore ask:
\begin{center}
    {\em Are exponentially many measurements necessary to estimate Pauli channels even using protocols that are concatenating and use at most $k = 0.99n$ qubits of quantum memory?}
\end{center}

\noindent The main result in our work is, surprisingly, an answer in the negative! In fact, we show that not only do concatenation and $k = 0.99n$ qubits of memory afford an \emph{exponential speedup}, but there is even a concatenating protocol that achieves $\poly(n)$ measurement complexity with just \emph{logarithmically} many qubits of memory:

\begin{theorem}[Informal, see Theorem~\ref{thm:ChanEstCon}]\label{thm:ChanEstCon_informal}
Given an arbitrary Pauli channel and error parameter $\epsilon$, there exists a $k=O(\log n/\epsilon^2)$-ancilla, non-adaptive, concatenating protocol that makes $\tilde{O}(n^2/\epsilon^2)$ measurements\footnote{Here, $\tilde{O}$ hides extra terms depending polylogarithmically on $n$} and learns all the eigenvalues of a given Pauli channel to within error $\epsilon$ with high probability.
\end{theorem}

\noindent Our protocol provides intuitive counterpoints to (A) and (B) above. As we explain in Section~\ref{sec:overview}, the algorithm circumvents (A) by utilizing the $O(\log n/\epsilon^2)$ qubits of quantum memory not to perform any kind of swap or Bell measurement on a subsystem, but instead to perform a certain \emph{purification} procedure that amplifies the difference between how the system behaves under the unknown noise channel versus under a fixed hypothesis channel. As for (B), the algorithm crucially exploits the small but nonzero amount of quantum memory in order to ``pump out entropy.'' The point is that while noise might accumulate in the $n$ system qubits from consecutive queries to the channel, the ancillary qubits do not experience noise and can thus be used to safely record information that is learned over the course of the protocol.


An important caveat of our result is that although the protocol makes $\poly(n)$ measurements, each measurement is performed after a sequence of exponentially many queries to the unknown channel. Naturally, one might wonder whether there exist better protocols which are efficient not just in terms of measurement complexity, but also in terms of query complexity. Unfortunately, we prove that for Pauli channel estimation, this is not possible:

\begin{theorem}[Informal, see \Cref{thm:AdaptConKLow}]\label{thm:AdaptConKLow_informal}
	Any (possibly adaptive and possibly concatenating) protocol for estimating the eigenvalues of an unknown Pauli channel to within constant error requires $\Omega(2^{(n-k)/3})$ queries.
\end{theorem}

\noindent We remark that this is a strict improvement upon the aforementioned $\Omega(2^{(n-k)/3})$ lower bound from~\cite{chen2022quantum}, as their lower bound applied only to non-concatenating protocols, for which there is no distinction between query complexity and measurement complexity. As a bonus, we also tighten the $2^{n/3}$ lower bound of~\cite{chen2022quantum} for concatenating protocols with $k = 0$ ancillas to get an optimal lower bound that works for \emph{any} $\epsilon$:

\begin{theorem}[Informal, see \Cref{thm:AdaptConLow}]\label{thm:AdaptConLow_informal}
    For any $0\le \epsilon < 1$, any (possibly adaptive and possibly concatenating) protocol for estimating the eigenvalues of an unknown Pauli channel to within error $\epsilon$ using $k = 0$ ancillas requires $\Omega(2^n/\epsilon^2)$ queries.
\end{theorem}

\subsection{Related works} 
Pauli channel estimation has been studied in a variety of settings. For learning the eigenvalues of a Pauli channel, \cite{chen2022quantum} studied the number of measurements needed by protocols under different assumptions, such as whether adaptive strategies, concatenating measurements, or quantum memory are allowed (see Section~\ref{sec:problem} for a formal definition of these notions). We defer a formal comparison between our results and theirs to \Cref{tab:Rev}. 

\begin{table}[t!]
\centering
\begin{tabular}{c|c|c|c|l} 
\textbf{Quantum memory}  & \textbf{Adaptive} & \textbf{Concatenating} & \textbf{Accuracy} &  \textbf{Lower bound}\\
\hline\hline
No & No & No & $1/2$ & $T=\Omega(n2^{n})$ \\
\hline
No & Yes & Yes & $1/2$ & $T=\Omega(2^{n/3})$ \\
\hline
No & Yes & Yes & $\epsilon\in(0,1]$ & $T=\Omega(2^{n}/\epsilon^2)$ [\textbf{*}] \\
\hline
$k$-qubit & No & No & $1/2$ & $T=\Omega(n2^{n-k})$ \\
\hline
$k$-qubit & Yes & No & $1/2$ & $T=\Omega(2^{(n-k)/3})$ \\
\hline
$k$-qubit & Yes & Yes & $1/2$ & $N=\Omega(2^{(n-k)/3})$ [\textbf{*}]\\
\hline
$n$-qubit & No & No & $1/2$ & $T=\Omega(n)$
\end{tabular}
\caption{Lower bounds for Pauli channel eigenvalue estimation in~\cite{chen2022quantum} and this work. Our bounds are denoted by asterisks. Here, $T$ is the number of measurements required and $N$ is the total number of queries to the unknown Pauli channel. See Section~\ref{sec:problem} for formal definitions of the terms in the table. The entries in the ``Accuracy'' column refer to whether the corresponding result is a lower bound for the task of estimating the eigenvalues to within constant additive error, or even to within $\epsilon$ additive error for arbitrary $\epsilon$ (note that lower bounds of the latter type are strictly more general).}
\label{tab:Rev}
\end{table}

For learning Pauli error rates $\bm{p}$, there have also been several works providing algorithms~\cite{flammia2021pauli,flammia2020efficient} and lower bounds~\cite{fawzi2023lower} for adaptive and non-adaptive protocols. In this setting, whether quantum memory can yield advantages for this task, either in the number of measurements made or the number of queries to the unknown channel, remains open. We note that the task of Pauli error rate estimation is orthogonal to the thrust of our work.

Finally, we note that our work is part of a larger body of recent results exploring how near-term constraints on quantum algorithms for various quantum learning tasks affect the underlying statistical rates, see e.g.~\cite{bubeck2020entanglement,chen2022toward,chen2022tight,chen2022tight2,chen2022exponential,chen2021hierarchy,aharonov2022quantum,fawzi2023quantum,huang2021information,huang2022quantum,fawzi2023lower,liu2023memory}. A review of this literature is beyond the scope of this work, and we refer the reader to the survey~\cite{anshu2023survey} for a more thorough overview.

\paragraph{Concurrent work.} 
During the preparation of an earlier version of this manuscript~\cite{chen2023futility} which is subsumed by the present work, we were made aware of the independent and concurrent work of~\cite{chen2023tight}, which also studied the complexity of Pauli channel eigenvalue estimation. Their main result is closely related to our Theorem~\ref{thm:AdaptConLow_informal}: they show an $\Omega(2^n/\epsilon^2)$ lower bound for algorithms without quantum memory. This result and our Theorem~\ref{thm:AdaptConLow_informal} are incomparable and offer complementary strengths. Our result holds for the full range of $\epsilon \in (0,1]$ whereas theirs holds for $\epsilon \in (0,1/6]$. On the other hand, their bound has a more favorable constant factor and applies to a more general family of \emph{classical-memory assisted} protocols. Lastly, we remark that our main result, the upper bound in Theorem~\ref{thm:ChanEstCon_informal}, is unique to our work, as is Theorem~\ref{thm:AdaptConKLow_informal}.

\subsection{Outlook}

The theorem established in this work shows that exponential statistical advantage is possible for Pauli channel estimation by using concatenation and only logarithmically many qubits of quantum memory. Looking ahead, several open problems remain to be answered. A direct open question is whether this task can be solved using polynomially many measurements and even fewer ancilla qubits, say, $k=O(1)$. If not, what is the tight quantum memory requirement for arbitrary polynomial measurement protocols? It is also interesting to resolve whether the $\epsilon$-dependence for $k$ in this work is necessary and explore the tradeoff between quantum memory and measurement complexity for concatenated protocols. Another important question is whether we can obtain similar exponential statistical speedups using $o(n)$ quantum memory for other quantum learning tasks. For instance, quantum state learning requires $(1-o(1))n$ ancilla qubits to achieve an exponential reduction in the measurement complexity as we can not utilize the fresh copies of the state in a ``concatenated" fashion. However, if we have some channel that can prepare the unknown state, i.e. a state-preparation oracle as in~\cite{van2023quantum}, we can query these oracles using concatenated protocols. Can we obtain a similar quantum speedup in this setting?

\subsection{Roadmap} 
In \Cref{sec:overview}, we give a high-level overview of the techniques we use to prove the results in this work. In \Cref{sec:Prelim}, we give a formal description of the Pauli channel estimation problem we consider and provide some technical preliminaries. In Section~\ref{sec:ConMeasProtocol}, we provide a concatenated algorithm that solves a simpler version of the Pauli channel estimation problem, which we call \emph{Pauli spike detection}, using only $k=O(1)$ ancilla and a single measurement. In Section~\ref{sec:ChanEstCon}, we significantly extend this analysis to prove Theorem~\ref{thm:ChanEstCon_informal} and provide the corresponding algorithm. In the remaining sections we prove our lower bounds: in Section~\ref{sec:MeasComp}, we prove Theorem~\ref{thm:AdaptConLow_informal} and in Section~\ref{sec:SampComp}, we prove Theorem~\ref{thm:AdaptConKLow_informal}. In Appendix~\ref{app:Shadowlow}, we use the techniques for proving Theorem~\ref{thm:AdaptConLow_informal} to prove a new lower bound for shadow tomography~\cite{aaronson2018shadow,huang2020predicting}.

\section{Overview of techniques}
\label{sec:overview}
Here, we give a high-level discussion of the techniques for our main results. To give the reader intuition for our protocol for Theorem~\ref{thm:ChanEstCon_informal}, in Section~\ref{sec:OverviewLearnSingle} we consider an important special case. In Section~\ref{sec:OverviewReduction} we explain the series of reductions we perform to extend the ideas for this special case to the general case. Finally, in Section~\ref{sec:lowerbounds}, we overview the key ideas in the proofs of our lower bounds.

\subsection{Learning single component Pauli channels}\label{sec:OverviewLearnSingle}
\begin{figure}[ht]
    \centering
    \includegraphics[width=0.93\textwidth]{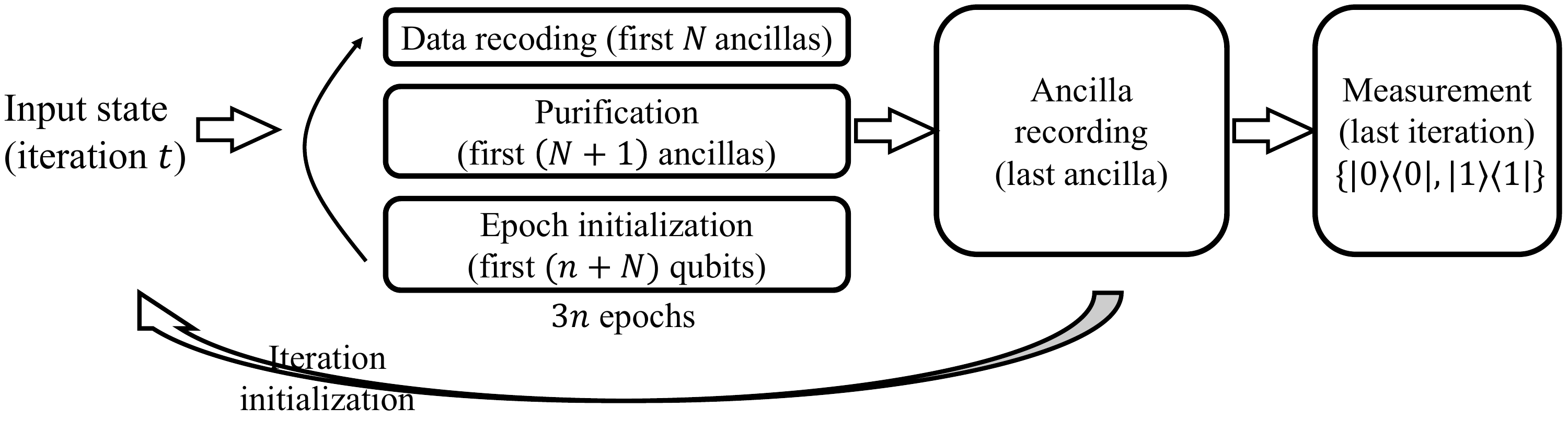}
    \caption{The overview of the algorithm for learning single component Pauli channels: the illustrative flow figure for the whole algorithm during iteration $t$. There are in total $4^n-1$ iterations as we have to enumerate all the possible Pauli strings $P_a$'s. A brief illustration of the algorithm is provided in \Cref{sec:OverviewLearnSingle} and the detailed algorithm is provided in \Cref{algo:LearnSingle}. }
    \label{fig:LearnSingle}
\end{figure}
\begin{table}[ht]
    \centering
    \begin{tabular}{|c|c|c|c|c|}  
    \hline
    \multirow{2}{*}{Operation} & \multirow{2}{*}{Output} & \multirow{2}{*}{Case 1} & \multicolumn{2}{|c|}{Case 2} \\ 
    \cline{4-5}
    & & & iter. $t\neq a$ & iter. $t=a$\\
    \hline
    \hline
    Input state & First $n+N+1$ & \multicolumn{3}{|c|}{$\frac{1}{2^n}(I+P_t)\otimes\ket{0}\bra{0}^{\otimes(N+1)}$}\\
    \hline
    \multirow{2}{*}{\makecell{Data recording \\($m$-th epoch)}} & \multirow{2}{*}{First $N$ ancilla} & \multirow{2}{*}{$I/2$} & \multirow{2}{*}{$I/2$} & \multirow{2}{*}{$(1-\lambda_a)\frac{I}{2}+\lambda_a\ket{0}\bra{0}$}\\
    &&&&\\
    \hline
    \multirow{2}{*}{\makecell{Purification \\($m$-th epoch)}} & \multirow{2}{*}{$(N+1)^\text{th}$ ancilla} & \multirow{2}{*}{$\text{diag}\left(\frac{1}{2^m},1-\frac{1}{2^m}\right)$} & \multirow{2}{*}{$\text{diag}\left(\frac{1}{2^m},1-\frac{1}{2^m}\right)$} & \multirow{2}{*}{$\left(1-\frac{1}{n}\right)^m\ket{0}\bra{0}+\rho'$}\\
    &&&&\\
    \hline
    \multirow{2}{*}{\makecell{Epoch initial \\($m$-th epoch)}} & \multirow{2}{*}{First $n+N$} & \multicolumn{3}{|c|}{\multirow{2}{*}{$\frac{1}{2^n}(I+P_t)\otimes\ket{0}\bra{0}^{\otimes N}$}}\\
    &&\multicolumn{3}{|c|}{}\\
    \hline
    \multirow{2}{*}{\makecell{Purification\\ ($3n$-th epoch)}}& \multirow{2}{*}{\makecell{Prob. $\ket{0}\bra{0}$\\$(N+1)^\text{th}$ ancilla}} & \multirow{2}{*}{$\frac{1}{8^n}$} & \multirow{2}{*}{$
    \frac{1}{8^n}$} & \multirow{2}{*}{$\geq e^{-3}$}\\
    &&&&\\
    \hline
    \multirow{3}{*}{Ancilla record} & \multirow{3}{*}{\makecell{Multi. factor\\ on $\ket{0}\bra{0}$ of\\$(N+2)^{\text{th}}$ ancilla}} & \multirow{3}{*}{$1-\frac{1}{8^n}$} & \multirow{3}{*}{$1-\frac{1}{8^n}$} & \multirow{3}{*}{$\leq\left(1-e^{-3}\right)$}\\
    &&&&\\
    &&&&\\
    \hline
    \multirow{2}{*}{\makecell{Prob. $\ket{1}\bra{1}$ \\($t=4^n-1$)}} & \multirow{2}{*}{$(N+2)^\text{th}$ ancilla} & \multirow{2}{*}{$\leq 2^{-n}$} & \multicolumn{2}{|c|}{\multirow{2}{*}{$\geq e^{-3}$}}\\
    &&&\multicolumn{2}{|c|}{}\\
    \hline
    \end{tabular}
    \caption{The overview of the algorithm for learning single component Pauli channels: the expected outcomes for the two cases after each step (channel) of the algorithm. In the algorithm, we choose $N=O(\log n/\epsilon^2)$ and use $k=N+2$ ancilla qubits. In the purification procedure of the $m$-th epoch, $\text{diag}(a,b)$ represents the density matrix of the state $a\ket{0}\bra{0}+b\ket{1}\bra{1}$ and $\rho'$ is the remaining part of the state which contains a non-negative coefficient for the entry $\ket{0}\bra{0}$. After the purification channel after the $3n$-th epoch, we consider the coefficient for the entry $\ket{0}\bra{0}$ (i.e., Probability of $\ket{0}\bra{0}$) at the $(N+1)$-th ancilla. We then record this information to the $(N+2)$-th ancilla qubit using the ancilla recording channel, which results in a multiplicative factor on the coefficient of $\ket{0}\bra{0}$. Finally, we consider the probability of obtaining result $\ket{1}\bra{1}$ on the last ancilla qubit after the last iteration.}
    \label{tab:LearnSingle}
\end{table}

As a warmup to give the reader a sense for the diverse ingredients that go into our final protocol for Pauli channel estimation, in this subsection we zoom in on a special case: suppose that the unknown channel only has a single nontrivial eigenvalue, that is, the channel is of the form
\begin{equation}
    \Lambda_a \triangleq \frac{1}{2^n}[I\Tr(\cdot)+ \lambda_aP_a\Tr(P_a(\cdot))]\,. \label{eq:LamAdef}
\end{equation}
While this might seem rather restrictive, it turns out many of the ingredients needed to handle the general case will already be present in the protocol we devise for this special case.

In this setting, the main difficulty is to identify $a$, after which estimating $\lambda_a$ to within sufficient accuracy is trivial. With a simple binary search argument one can further reduce the \emph{estimation} task of identifying $a$ to a \emph{distinguishing} task. We are promised that the unknown channel falls under one of two scenarios:
\begin{itemize}
    \item The unknown channel is the depolarization channel $\Lambdadep=\frac{1}{2^n}I\Tr(\cdot)$, or
    \item The unknown channel is sampled uniformly from a set of channels $\{\Lambda_a\}_{a\in\Z_2^{2n}\backslash \{0\}}$, where $\Lambda_a$ is given by Eq.~\eqref{eq:LamAdef} for some $|\lambda_a| > \epsilon$,
\end{itemize}
and the goal is simply to output which scenario we are in (see~\Cref{prob:SingleComp} for a formal definition). As we show in \Cref{lem:LearnSingle}, this task is solvable with $k=O(\log n/\epsilon^2)$ ancilla qubits and channel concatenation using only $O(n)$ measurements. We now overview the proof of~\Cref{lem:LearnSingle}. The flow chart of the algorithm is provided in \Cref{fig:LearnSingle} and the expected outcomes after each ingredient of the algorithm is given in \Cref{tab:LearnSingle}

\paragraph{Sequentially guessing Pauli strings.} One basic idea in our protocol is to sequentially guess Pauli strings $P_t$'s for $t=1,...,4^n-1$ over one long sequence of concatenations before measuring. Note that because the measurement is only performed at the end of this sequence, it isn't clear how to actually ascertain at any iteration whether one has correctly guessed $t = a$. Indeed, the entire challenge will be to somehow record the event that this happens in the state of the ancillas so that 1) when we reach $t = a$, the state of the ancillas changes sharply compared to the state they would be in if the unknown channel were simply the depolarization channel, and 2) this altered state persists for the rest of the concatenations.

At first blush, one might wonder whether the following observation suffices. Note that when we pass the state $(I+P_t)/2^n$ through the unknown channel and measure it along the subspace of all positive eigenstates of $P_t$ and all negative eigenstates of $P_t$, the probabilities of the two measurement outcomes are 
\begin{equation}\label{eq:ProbTA}
    \Bigl(\frac{1}{2},\frac{1}{2}\Bigr) \ \ \text{versus} \ \ \Bigl(\frac{1+\lambda_a\cdot\mathds{1}[a=t]}{2},\frac{1-\lambda_a\cdot\mathds{1}[a=t]}{2}\Bigr)
\end{equation}
when the unknown channel is $\Lambdadep$ and $\Lambda_t$ respectively.

It turns out that when $\lambda_a = 1$, which corresponds to what we call the \emph{Pauli spike detection task} (defined in \Cref{prob:SpikeDetect}), this is already essentially enough to give an efficient protocol, as we now explain.

We consider using the ancilla qubits to keep track of whether one has guessed the ``spike'' $a$ yet. Roughly speaking, because $\lambda_a$ is so large, as soon as one has correctly guessed $a$, it is not hard to create a large difference between the state of the ancillary qubits compared to the state they would have been in if the channel were the depolarization channel, and this difference can then be detected with a measurement at the end (see the proof for \Cref{thm:SpikeDetCon} and \Cref{thm:Charsingle} in \Cref{sec:ConMeasProtocol}). In particular, we use a family of \emph{control channels} (formally defined in \Cref{exp:ControlChan}) to record the information to ancilla qubits. This control channel either keeps the target ancilla unchanged or prepares it in state $\ket{1}\bra{1}$, and the probabilities for these two cases are given by the probabilities in Eq.~\eqref{eq:ProbTA}. Initializing the target ancilla in $\ket{0}\bra{0}$ at the beginning of the iteration, we repeatedly implement the control channel $O(n)$ times \---- we refer to each of these $O(n)$ rounds as an ``epoch." If the unknown channel is the depolarization channel or $t\neq a$, the final state of the target ancilla will have an exponentially small coefficient for $\ket{0}\bra{0}$. At $t=a$, however, the state of the target ancilla will be kept unchanged in $\ket{0}\bra{0}$. This creates the large difference between the state of the ancillary qubits in two cases as desired.

After creating the differences on one target ancilla qubit in the two cases when we reach $t=a$ for $\Lambda_a$ and when $t\neq a$ for $\Lambda_a$ or the unknown channel is $\Lambdadep$, we record this difference when we are enumerating all $t=1,...,4^n-1$. The idea here is to employ another control channel that acts on another ancilla qubit 
controlled by the target ancilla. If in some iteration $t$ we have $t=a$, we will record a sharp change of $O(1)$ distance on the final ancilla, otherwise, we will only perform an exponentially small perturbation. Finally, we measure the final ancilla qubit in the computational basis to distinguish the two cases.

\paragraph{Handling smaller $\lambda_a$.}

Unfortunately, when $\lambda_a$ is bounded away from $1$, the simple strategy outlined above for $\lambda_a = 1$ for creating a large difference between ancillary states under the two scenarios breaks down. Indeed, if we simply initialize the target ancilla to $\ket{0}\bra{0}$ at the beginning of the iteration and repeatedly implement the control channel for $O(n)$ epochs, we will always get an exponentially small $\ket{0}\bra{0}$ component in the final state of the target ancilla in both scenarios. Intuitively, without additional processing, the signal from the nontrivial eigenvalue gets exponentially damped.

To circumvent this, we make the following important modification to the above sequential guessing procedure. Instead of repeatedly applying the control channel to a single ancilla, in each epoch we apply the control channel $N=O(\log n/\epsilon^2)$ times and record the distribution in \eqref{eq:ProbTA} on $N$ \emph{different} ancilla qubits such that each ancilla is in
\begin{equation*}
    \frac{I}{2}\qquad\text{versus} \qquad \frac{1+\lambda_a\cdot\mathds{1}[a=t]}{2}\ket{0}\bra{0}+\frac{1-\lambda_a\cdot\mathds{1}[a=t]}{2}\ket{1}\bra{1}\,,
\end{equation*}
depending on whether the unknown channel is equal to $\Lambdadep$ or $\Lambda_a$ for $a\neq t$, or equal to $\Lambda_t$, respectively. 

Our key step is to then apply a special \emph{state purification} on these $N$ copies to amplify the signal introduced by the nontrivial eigenvalue.\footnote{The setting here is different from the standard state purification setting~\cite{cirac1999optimal,childs2023streaming}, where one would need linearly instead of logarithmically many copies. Roughly, the reason is that the rank-1 component that we are trying to purify towards is \emph{known} in our setting. See \Cref{rem:DiffPurification} for a discussion.} This amplification ensures that after $3n$ epochs, we can guarantee that if the unknown channel is the depolarization channel or $t\neq a$, the final state of the target ancilla will have an exponentially small coefficient for $\ket{0}\bra{0}$. However, if $t=a$, the state of the target ancilla will have a constant coefficient in $\ket{0}\bra{0}$. This again produces the large difference between the state of the ancilla qubits in the two cases needed to solve the distinguishing task.

\subsection{From learning a single component to learning arbitrary Pauli channels}\label{sec:OverviewReduction}

\begin{figure}[ht]
    \centering
    \includegraphics[width=0.9\textwidth]{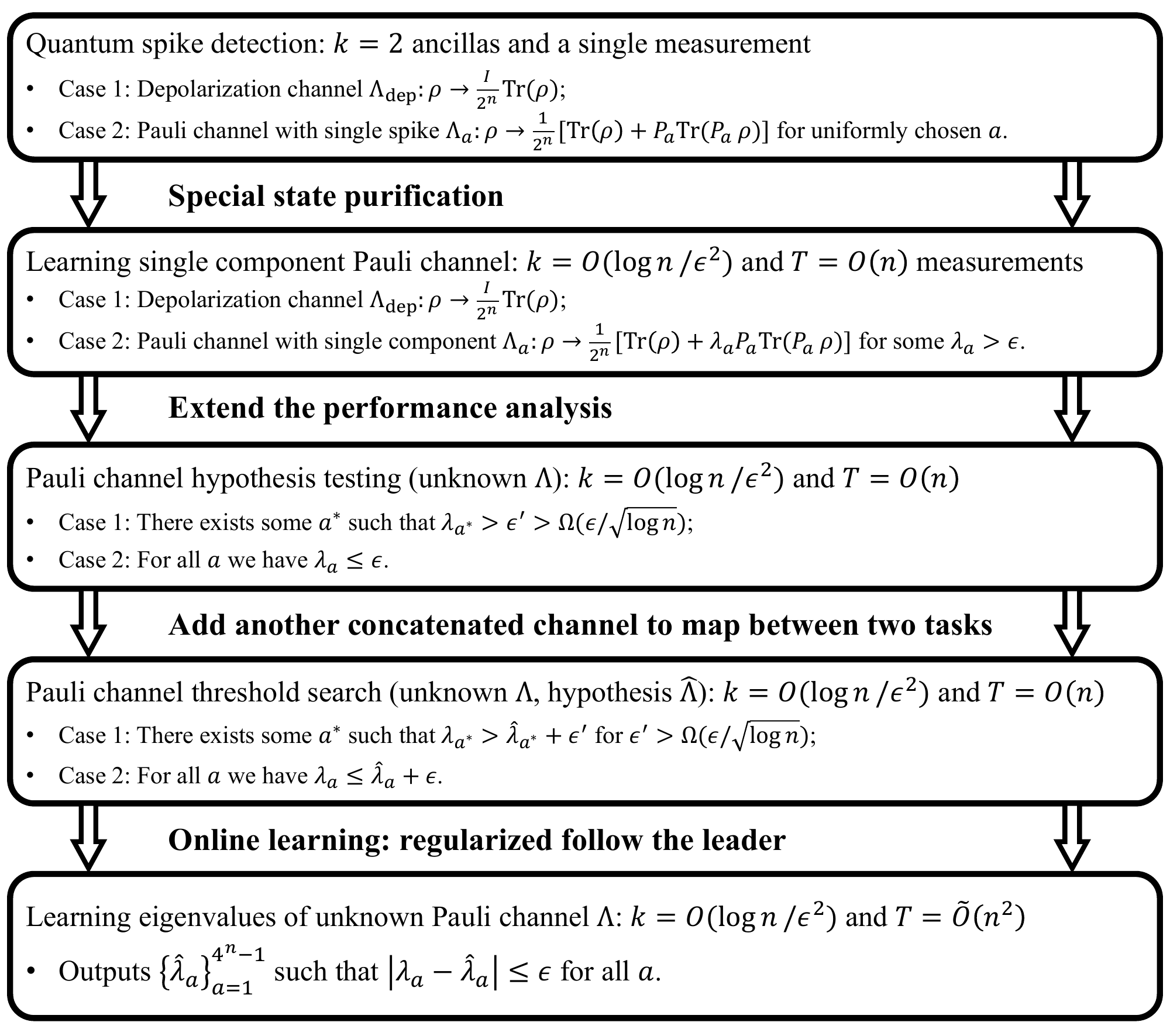}
    \caption{The illustrative figure for the reduction process from the algorithm for Pauli spike detection in \Cref{thm:SpikeDetCon} to the algorithm for estimating eigenvalues for a Pauli channel in \Cref{thm:ChanEstCon}. The technical overview for the first step reduction from Pauli spike detection to learning single component Pauli channel is provided in \Cref{sec:OverviewLearnSingle} and the remaining reductions are introduced in \Cref{sec:OverviewReduction}.}
    \label{fig:reduction}
\end{figure}
We further implement a series of reductions to solve the general Pauli channel eigenvalue estimation task starting from the single component Pauli channel learning protocol outlined in the previous section. We provide an illustration of the reductions in~\Cref{fig:reduction} about the reduction process from the algorithm for Pauli spike detection in \Cref{sec:ConMeasProtocol} (\Cref{thm:SpikeDetCon}) to the algorithm for estimating eigenvalues for a Pauli channel in \Cref{sec:ChanEstCon} (\Cref{thm:ChanEstCon}). 

We first show that the analysis of the protocol mentioned above can solve the more general question of \emph{Pauli channel hypothesis testing} (formally defined in \Cref{prob:PauliHypoTest}). The goal here is to determine whether the unknown channel deviates significantly from $\Lambdadep$ on some eigenvalue, but one is no longer promised that there is exactly one nontrivial eigenvalue as in the previous section. We prove that this task can be solved by adapting the analysis of the protocol for learning a single component Pauli channel (see~\Cref{lem:PauliHypoTest}). On the one hand, we can prove that when there are some large eigenvalues of value $\geq\epsilon$, we can always detect its existence using the protocol from the previous section. Conversely, by using anti-concentration for the binomial distribution, we prove that there are no ``false positives,'' i.e. when the protocol behaves as if it is in the second scenario of the single-component task, we are guaranteed that there is some eigenvalue of value at least $\epsilon'=\Theta(\epsilon/\sqrt{\log n})$.

Next, we go from Pauli channel hypothesis testing to \emph{Pauli channel threshold search} (defined in \Cref{prob:ThresholdSearch}). The goal for this task is as follows: given a description of a hypothesis Pauli channel $\hat{\Lambda}=\frac{1}{2^n}[I\Tr(\cdot)+\sum_{a=1}^{4^n-1}\hat{\lambda}_aP_a\Tr(P_a(\cdot))]$, output either
\begin{itemize}
    \item an $a^*\in\{1,...,4^n-1\}$ such that $\lambda_{a^*}>\hat{\lambda}_{a^*}+\epsilon'$ for some $\epsilon'=\Theta(\epsilon/\sqrt{\log n})$, or
    \item for all $a$'s we have $\lambda_a\leq\hat{\lambda}_a+\epsilon$,
\end{itemize}
with high success probability. When $\hat{\lambda}_a = 0$ for all $a$, this is essentially the ``search'' version of the Pauli channel hypothesis testing question. 

\begin{remark}
    This task is closely related to the quantum threshold search problem~\cite{buadescu2021improved} studied in the context of shadow tomography of quantum states. However, we note that our strategy for solving this problem is significantly different from existing approaches for shadow tomography: we are only allowed to use $O(\log n/\epsilon^2)$ qubits of quantum memory (as shown in \Cref{lem:ThresholdSearch}), while e.g. the protocol in~\cite{buadescu2021improved} for shadow tomography requires joint measurements on multiple copies of the unknown state, which requires $\omega(n)$ qubits of memory.
    That said, the settings are incomparable: in our case, we can perform channel concatenation to suppress the measurement complexity, whereas there is no analogous notion for learning states.
\end{remark}

\noindent We construct a reduction using one channel concatenated after each query to $\Lambda$ to reduce the Pauli channel threshold search problem to a Pauli channel hypothesis testing problem. Suppose $\hat{\lambda}_t\geq 0$. Specifically, given a quantum state $(I+\sigma_tP_t)/2^n$ for some $\sigma_t\in[-1,1]$, we construct a channel that maps $\sigma_t\to \sigma_t/(1+\hat{\lambda_t})-\hat{\lambda}_t/(1+\hat{\lambda_t})$. This channel is a linear map that maps $\sigma_t$ to $0$ if $\sigma_t=\hat{\lambda_t}$ and $\sigma_t$ to some value at least $\epsilon/(1+\hat{\lambda}_t)\geq\epsilon/2$ if $\sigma_t\geq\hat{\lambda}_t+\epsilon$. This reduces the Pauli channel threshold search problem with parameters $(\epsilon,\epsilon')$ to a Pauli channel hypothesis testing with parameters $(\epsilon,\epsilon'/2)$. We can also construct a channel that maps the two tasks for $\hat{\lambda}_t<0$.

Once we have an algorithm for threshold search, we can naturally feed it into an existing framework based on online learning~\cite{aaronson2018online,buadescu2021improved}. Roughly, this framework envisions the following interaction between a ``student'' who tries to refine her estimate of the unknown channel, and a ``teacher'' who identifies directions for improvement using threshold search. In each round of interaction, the student guesses some hypothesis channel $\hat{\Lambda}_t=\frac{1}{2^n}[I\Tr(\cdot)+\sum_{a=1}^{4^n-1}\hat{\lambda}_{a,t}P_a\Tr(P_a(\cdot))]$. Denote the true unknown channel by $\Lambda=\frac{1}{2^n}[I\Tr(\cdot)+\sum_{a=1}^{4^n-1}\lambda_aP_a\Tr(P_a(\cdot))]$. The student then receives some index $a_t\in\{1,...,4^n\}$. After receiving the hypothesis channel $\hat{\Lambda}_t$ and the index $a_t$, the teacher must either ``pass", or else declare a ``mistake" and supply a value $b_t$ such that $\abs{b_t-\lambda_{a_t}}\leq\epsilon_2$ for some $\epsilon_2\lesssim\epsilon'/2$. By appealing to standard regret minimization algorithms, we give an algorithm in~\Cref{lem:OnlineIter} for the student which leads to at most $\tilde{O}(n/\epsilon^2)$ ``mistakes." We implement the teacher via our threshold search procedure. The overall measurement complexity is thus $\tilde{O}(n^2/\epsilon^2)$.

\subsection{Lower bounds}
\label{sec:lowerbounds}

Our lower bounds are proved using the ``learning tree'' framework introduced in~\cite{bubeck2020entanglement,aharonov1997fault,huang2021information,chen2022exponential} for reasoning about adaptive protocols for quantum learning. As this technique is by now rather standard, we only emphasize here the key technical points of departure from existing approaches.

Previously, \cite{chen2022quantum} used this framework to prove an $\Omega(2^{n/3})$ measurement complexity lower bound for ancilla-free protocols for Pauli channel estimation. Their approach was based on extending the truncation technique of~\cite{huang2021information}, originally used to show a lower bound for shadow tomography, to channel learning. Like the bound in~\cite{huang2021information} however, the bound in~\cite{chen2022quantum} was off from the optimal bound by a factor of $3$ in the exponent. One might wonder whether the subsequent improved lower bound of~\cite{chen2022exponential} for shadow tomography could be similarly adapted to Pauli channel estimation. Unfortunately, that technique crucially exploited a certain convexity argument that breaks down in the presence of channel concatenations. To address this, we instead leverage a probabilistic argument based on likelihood ratio martingales, introduced in~\cite{chen2022complexity}. To further demonstrate the utility of this technique, we also show in~\Cref{app:Shadowlow} how to prove a refined version of the lower bound in~\cite{chen2022exponential} for shadow tomography.

Finally, our proof of~\Cref{thm:AdaptConKLow_informal} is based on carefully adapting another set of techniques from~\cite{chen2022exponential} for showing a lower bound for shadow tomography with bounded memory. As the details for this are rather technical, we defer them to Section~\ref{sec:SampComp}.

\section{Preliminaries}\label{sec:Prelim}

\subsection{Problem description}
\label{sec:problem}

Here we formalize the problems we consider throughout this paper, following the notation of~\cite{chen2022quantum}. We begin with the definition of a Pauli channel. In $n$-qubit Hilbert space, define $\sf{P}^n$ to be the \emph{Pauli group}. This is the abelian group consisting of $n$-fold tensor products of the single-qubit Pauli operators $I, X, Y, Z$. We will often denote $\sf{P}^n$ by $\{I,X,Y,Z\}^{\otimes n}$. Note that $\sf{P}^n$ is isomorphic to $\Z_2^{2n}$, as we can view every $a\in\Z_2^{2n}$ as a $2n$-bit classical string $1=a_{x,1}a_{z,1}\cdots a_{x,n}a_{z,n}$ corresponding to the Pauli operator $P_a=\otimes_{k=1}^ni^{a_{x,k}a_{z,k}}X^{a_{x,k}}Z^{a_{z,k}}$, where the phase is chosen to ensure Hermiticity.  We can further define a symplectic inner product $\langle\cdot,\cdot\rangle$ via $\langle a,b\rangle \coloneqq \sum_{k=1}^n(a_{x,k}b_{z,k}+a_{z,k}b_{x,k}) \mod{2}$.

An $n$-qubit Pauli channel $\Lambda$ is a quantum channel which, given density matrix $\rho$, outputs
\begin{align*}
\Lambda(\rho) \coloneqq \sum_{a\in\Z_a^{2n}}p_a\cdot P_a\rho P_a\,,
\end{align*}
where the parameters $\bm{p}\coloneqq\{p_a\}_a$ are the \emph{Pauli error rates}. Alternatively, $\Lambda$ can be expressed as
\begin{align}\label{eq:PauliEig}
\Lambda(\rho)=\frac{1}{2^n}\sum_{b\in\Z_2^{2n}}\lambda_b\Tr[P_b\rho]P_b\,,
\end{align}
where the parameters $\bm{\lambda}\coloneqq\{\lambda_b\}_b$ are the \emph{Pauli eigenvalues}~\cite{shor1996fault,aharonov1997fault,harper2020efficient,liu2021benchmarking,flammia2020efficient,harper2021fast,flammia2021pauli}. These two sets of parameters are related via the Walsh-Hadamard transformation:
\begin{align*}
    \lambda_b=\sum_{a}p_a(-1)^{\langle a,b\rangle} \qquad\text{and}\qquad p_a=\frac{1}{4^n}\sum_b\lambda_b(-1)^{\langle a,b\rangle}\,.
\end{align*}
The parameters $\bm{p}$ and $\bm{\lambda}$ are of complementary importance in applications like benchmarking and error mitigation~\cite{endo2018practical,endo2021hybrid}, and in this work we will focus on the latter.

We are now ready to formally state the recovery goal of \emph{Pauli Channel Eigenvalue Estimation} that we are studying:
\begin{problem}\label{prob:def}
Given query access to a Pauli channel with eigenvalues $\blambda=\{\lambda_b\}_b$ defined in~\eqref{eq:PauliEig}, our goal is to output a set of estimates $\hat{\bm{\lambda}}=\{\hat{\lambda}_{b}\}_b$ such that $\max_b|\hat{\lambda}_b-\lambda_b|<\epsilon$.
\end{problem}

\begin{figure*}[t]
    \centering
    \includegraphics[width = 0.9\textwidth]{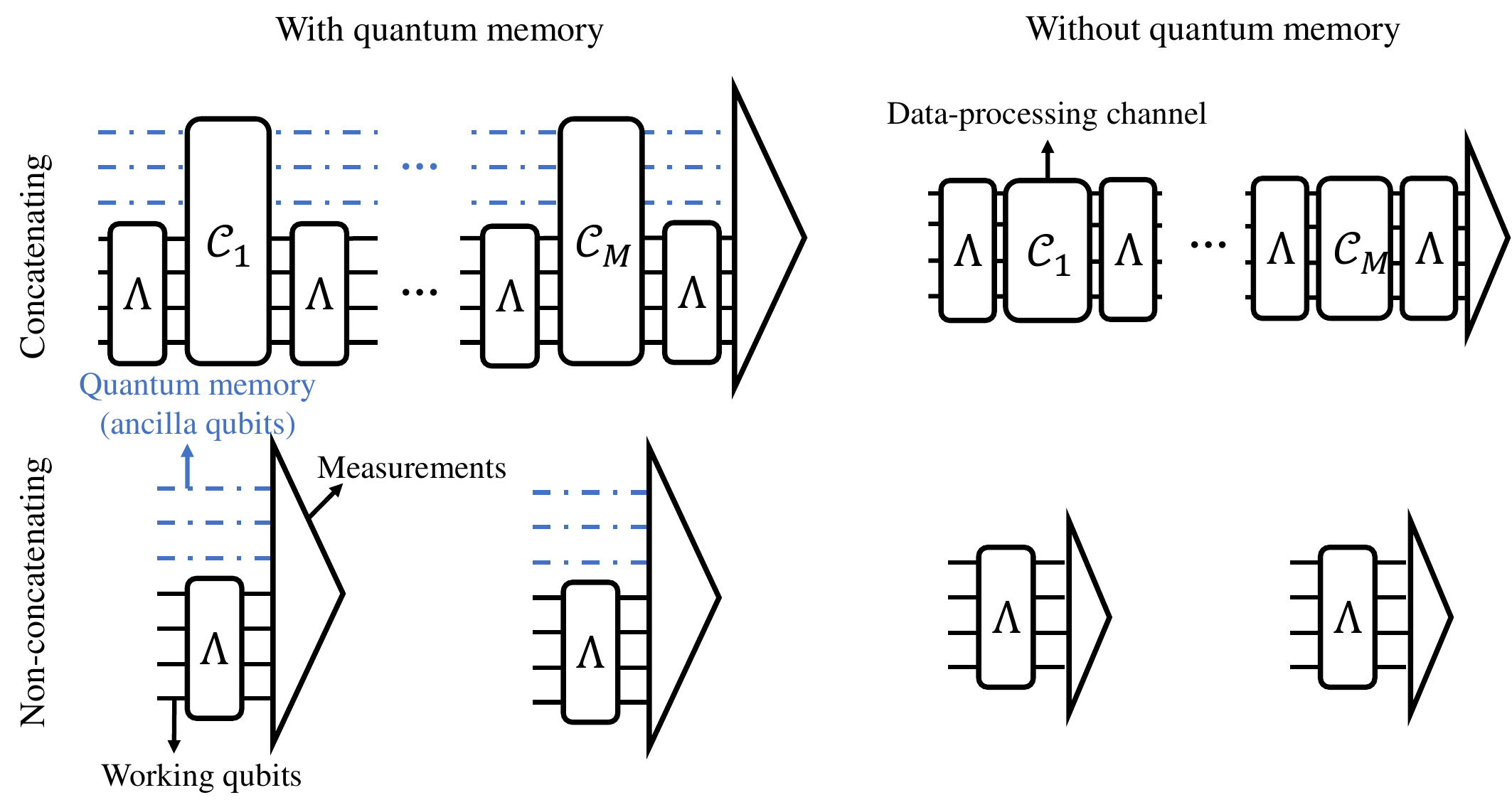}
    \caption{ Different measurement models for estimating an unknown channel $\Lambda$. The triangles denote state preparation and measurement, the boxes denote data-processing channels, black solid lines denote the working qubits, and blue dotted lines denote the quantum memory. }
    \label{fig:model}
\end{figure*}

It remains to specify how a learning protocol makes use of query access to $\Lambda$. we will consider a number of different models for this (see \Cref{fig:model}). 

The first axis along which these protocols can be classified is whether they are \emph{concatenating}. In each round of a concatenating protocol, one initializes the system to some state and then sequentially performs queries to the unknown channel $\Lambda$, interleaved with data-processing channels, before applying a measurement at the end. On the other hand, in each round of a non-concatenating protocol, the channel is queried once on some state, and this is immediately followed by a measurement.

The second axis is whether the protocol has \emph{quantum memory}. For our purposes, this refers to whether the system over which the protocol operates has some number of ancilla qubits in addition to the $n$ qubits over which $\Lambda$ acts.

The final axis is whether the protocol is \emph{adaptive}. For non-concatenating protocols, this means that in each round, the choice of state on which $\Lambda$ is queried and the measurement can depend on the previous history of measurement outcomes. For concatenating protocols, in addition to these one can also choose the intermediate data-processing channels adaptively. All of the lower bounds that we will prove apply even to adaptive protocols.

Finally, we must define the metric(s) one can use to assess the complexity, i.e. the resource demands, of a protocol for Pauli channel estimation. In this work we will discuss two metrics: (1) \emph{measurement complexity}, i.e. the total number of measurements made by the protocol, and (2) \emph{query complexity}, i.e. the total number of queries to the unknown channel $\Lambda$. For non-concatenating protocols, note that these are equivalent.

\subsection{Basic results in quantum information}\label{sec:BasicQ}

Here we record some standard definitions and calculations in quantum information.

\paragraph{Quantum measurements.} An $n$-qubit POVM is given by a set of positive-semidefinite matrices $\{F_s\}_s$, each corresponding to a classical outcome $s$, that satisfy $\sum_s F_s=I$. Here, $F_s$ is known as the POVM element. 

We consider the post-measurement state for POVM $\{F_s\}_s$ where each POVM element has some Cholesky decomposition $F_s = M_s^\dagger M_s$. When one measures a mixed state $\rho$ using this POVM, the probability of observing outcome $s$ is $\Tr(F_s\rho)$. If we consider implementing POVM $\{F_s\}_s$ on a given quantum state $\rho$, the post-measurement quantum state upon measuring $\rho$ with this POVM and observing outcome $s$ is given by
\begin{align*}
\rho\to\frac{M_s\rho M_s^\dagger}{\Tr(M_s^\dagger M_s\rho)}.
\end{align*}

When $F_s$ is rank-1 for all $s$, then we say that $\{F_s\}_s$ is a \emph{rank-1 POVM}. It is a standard fact, e.g. \cite[Lemma 4.8]{chen2022exponential}, that from an information-theoretic perspective, any POVM can be simulated by a rank-$1$ POVM and classical post-processing. In the sequel, we will thus assume without loss of generality, we assume all measurements in the learning protocols we consider to be rank-$1$ POVMs. We note that a useful parametrization for rank-1 POVMs is via
\begin{align*}
    \{w_s2^n\ket{\psi_s}\bra{\psi_s}\}\,,
\end{align*}
for pure states $\{\ket{\psi_s}\}$ and nonnegative weights $w_s$ satisfying $\sum_s w_s = 1$.

\paragraph{Identities on Pauli matrices.} 

We will use the following well-known lemma for the sum of tensor products of pairs of Pauli matrices:
\begin{lemma}[E.g., Lemma 4.10 in~\cite{chen2022exponential}] \label{lem:SumPauli}
    We have
    \begin{align*}
        \sum_{P_a\in\{I,X,Y,Z\}^{\otimes n}}P_a\otimes P_a=2^n\,\mathrm{SWAP}_n\,,    
    \end{align*} 
    where $\mathrm{SWAP}_n$ is the $2n$-qubit SWAP operator.
\end{lemma}

\noindent Using \Cref{lem:SumPauli}, we can obtain the following equality:

\begin{lemma}[E.g., Lemma 5.8 in~\cite{chen2022exponential}]\label{lem:SumPauliExp}
We have
\begin{align*}
    \sup_{\ket{\phi}}\E_a[\bra{\phi}P_a\ket{\phi}^2]=\frac{1}{2^n+1}\,,    
\end{align*}
where the expectation is uniform over $a\in \{I,X,Y,Z\}^{\otimes n} \,\backslash \, I^{\otimes n}$ and $\ket{\phi}$ ranges over all pure states.
\end{lemma}

\paragraph{Quantum channels. }A quantum channel $\mathcal{C}$ from $n$ qubits to $m$ qubits is a linear and completely-positive trace-preserving map $\mathcal{C}:\mathbb{H}^{2^n}\to\mathbb{H}^{2^m}$, where $\mathbb{H}^d$ is the $d$-dimensional Hilbert space. Moreover, it can be written in the following form~\cite{nielsen2010quantum}:
\begin{align*}
\mathcal{C}(\rho)=\sum_iK_i\rho K_i^\dagger,
\end{align*}
where $K_i\in\mathbb{H}^{2^m\times 2^n}$ satisfying $\sum_iK_i^\dagger K_i=\mathbbm{I}$ are called Kraus operators. In this paper, we focus on the case $n=m$. We also denote $I(\cdot)$ as the identity channel in this paper; the number of qubits $n$ involved will be clear from the context. According to the Stinespring dilation theorem~\cite{stinespring1955positive}, any channel $\mathcal{C}:\mathbb{H}^{2^n}\to\mathbb{H}^{2^n}$ can be converted into a unitary evolution acting on an extended space and tracing out the ancilla qubits afterward. The dimension of the extended Hilbert space is bounded above by $2^{2n}$.

A special class of quantum channels that is used in this paper is the control channel (see \Cref{sec:ConMeasProtocol} and \Cref{sec:ChanEstCon}).
\begin{example}[Control channels]\label{exp:ControlChan}
Suppose we have a $n$-qubit quantum system that can be divided into two subsystems each of $n_1$ and $n_2$ qubits with $n=n_1+n_2$. We consider a two-outcome POVM $\{E,I-E\}$ on the first subsystem and two quantum channels $\mathcal{C}_1^{n_2}$ and $\mathcal{C}_2^{n_2}$ on the second subsystem. Given any $\rho=\rho^{n_1}\otimes\rho^{n_2}$ that can be represented as a direct product of some state $\rho^{n_1}$ and $\rho^{n_2}$ in the first and the second subsystem, a control channel $\mathcal{C}:\mathbb{H}^{2^n}\to\mathbb{H}^{2^n}$ on the $n$-qubit quantum system acts by:
\begin{align*}
\mathcal{C}(\rho)=\Tr(E\rho^{n_1})\frac{I}{2^{n_1}}\otimes\mathcal{C}_1^{n_2}(\rho^{n_2})+\Tr((I-E)\rho^{n_1})\frac{I}{2^{n_1}}\otimes\mathcal{C}_2^{n_2}(\rho^{n_2}).
\end{align*} 
\end{example}

We can verify that such a control channel in \Cref{exp:ControlChan} always exists because we can verify that it is linear, completely positive, and trace-preserving. In particular, the linear property follows from the fact that $\mathcal{C}_1^{n_2}$ and $\mathcal{C}_2^{n_2}$ are linear channels. The positive property follows from the observation that $\Tr(E\rho^{n_1})$ and $\Tr((I-E)\rho^{n_1})$ are non-negative for arbitrary $\rho^{n_1}$. Finally, the trace-preserving property comes from the equality that $\Tr(E\rho^{n_1})+\Tr((I-E)\rho^{n_1})=1$ for any POVM $\{E,I-E\}$.

\paragraph{Choi–Jamio{\l}kowski isomorphism.} An alternative way to characterize a given channel $\mathcal{C}:\mathbb{H}^{2^n}\to\mathbb{H}^{2^n}$ is the Choi–Jamio{\l}kowski isomorphism~\cite{choi1975completely,jamiolkowski1972linear}, which we will use in constructing the regularization function of our protocol for learning Pauli channels in~\Cref{sec:OnlinePauli}. Intuitively, the Choi–Jamio{\l}kowski isomorphism maps the set of Kronecker basis states $\{\ket{i}\bra{j}\}$ to a set of canonical basis states $\ket{i,j}$ in an extended space. More formally, given a channel $\mathcal{C}:\mathbb{H}^{2^n}\to\mathbb{H}^{2^n}$ acting on Hilbert space $S$, consider an auxiliary system $S'$ with the same dimension. Define the maximally entangled state in the extended space $S\otimes S'$:
\begin{align*}
\ket{\Phi^+}=\frac{1}{\sqrt{2^n}}\sum_{i=0}^{2^n-1}\ket{i}\otimes\ket{i}.
\end{align*}
The state associated to $\mathcal{C}$ under the Choi-Jamio{\l}kowski isomorphism is given by passing $\ket{\Phi^+}$ through the channel $\mathcal{C}\otimes I_{S'}$, that is,
\begin{align}\label{eq:Choistate}
    \rho_{\text{CJ}}(\mathcal{C})=\mathcal{C}\otimes I_{S'}(\ket{\Phi^+}\bra{\Phi^+})    
\end{align}
Note that this is a linear map.

\subsection{Probability theory}\label{sec:ProbThm}
We will need the concentration and anti-concentration inequality for binomial distribution in \Cref{sec:SingleComp} and \Cref{sec:ThresholdSearch} when we are designing algorithms for learning a Pauli channel with a single component and solving the Pauli channel hypothesis testing problem.

We will use $X\sim B(N,p)$ to denote a sample from the binomial distribution, i.e.
\begin{align*}
\Pr[X=k]=\binom{N}{k}p^k(1-p)^{N-k}
\end{align*}
for $0\leq k\leq N$. In the sequel, the parameter $p$ for this binomial distribution will be given by $p=\frac{1-\epsilon}{2}$ for some $\epsilon>0$. 

Without loss of generality, we assume $N$ is even and we consider the upper tail larger than $k=N/2$ with $k$ being an integer. By using Hoeffding's inequality, we can obtain the concentration inequality on the upper tail of the cumulative distribution function at $X\geq N/2$:
\begin{align}\label{eq:Concentration}
\Pr[X\geq N/2]\leq\exp(-N\epsilon^2/2).
\end{align}

Also, we can also obtain the anti-concentration inequality on the upper tail of the cumulative distribution function at $X\geq N/2$ as follows:
\begin{align}\label{eq:AntiConcentration}
\Pr[X\geq N/2]\geq\frac{1}{2}\left(1-\sqrt{1-\exp\left(-\frac{N\epsilon^2}{1-\epsilon^2}\right)}\right).
\end{align}

An exemplary proof for Eq.~\eqref{eq:AntiConcentration} can be given as follows. According to Slud's inequality~\cite{slud1977distribution}, If $k$ is an integer, $p<1/2$, and $Np\leq k\leq N(1-p)$, then we have
\begin{align*}
\Pr[X\geq k]\geq\Pr[Z\geq\frac{k-Np}{\sqrt{Np(1-p)}}].
\end{align*}
where $Z\sim\mathcal{N}(0,1)$ is a random variable with standard Gaussian distribution. Thus, we have 
\begin{align*}
\Pr[X\geq k]\geq\Pr[Z\geq\frac{k-Np}{\sqrt{Np(1-p)}}]\geq\frac{1}{2}\left(1-\sqrt{1-\exp\left(-\frac{N\epsilon^2}{1-\epsilon^2}\right)}\right).
\end{align*}

\subsection{Online learning, regret, and mistake bounds}\label{sec:OnlineLearning}
In this part, we introduce the tool of online learning we will require for the Pauli channel eigenvalue estimation algorithm for \Cref{thm:ChanEstCon_informal} (and formally for \Cref{thm:ChanEstCon} in \Cref{sec:OnlinePauli}). Online learning has been developed as an important tool in learning quantum tasks. Initialized by Ref.~\cite{aaronson2018online}, online learning has been exploited in learning properties of unknown quantum states (also known as shadow tomography). It is also employed in the following works with improved bounds of shadow tomography~\cite{aaronson2018shadow,aaronson2019gentle,buadescu2021improved,gong2023learning} and different learning settings~\cite{chen2020more,yang2020revisiting,lumbreras2022multi,chen2022adaptive,zimmert2022pushing}.

In this paper, we consider an online learning scheme for estimating the eigenvalues of a given Pauli channel in \Cref{sec:OnlinePauli}. As defined in \Cref{prob:def}, given a Pauli channel 
\begin{align*}
\Lambda(\cdot)=\frac{I\Tr(\cdot)+\sum_{a}\lambda_aP_a\Tr(P_a(\cdot))}{2^n},    
\end{align*}
the goal of the learner is to learn all the eigenvalues $\lambda_a$ within accuracy demand $\epsilon$. In iteration $t$ of the online learning scheme, the learner is given an index $a_t\in\{1,...,4^n-1\}$ and is challenged to output a hypothesis Pauli channel 
\begin{align*}
\hat{\Lambda}_t=\frac{I\Tr(\cdot)+\sum_{a=1}^{4^n-1}\hat{\lambda}_{a,t}P_a\Tr(P_a(\cdot))}{2^n}.    
\end{align*} 
The hypothesis channel is also known as the \emph{prediction} in iteration $t$. The learner then obtains feedback. The feedback considered in our algorithm is a ``good" approximation of $\lambda_{a_t}$ with $b_t\in[0,1]$ with $\abs{b_t-\lambda_{a_t}}\leq O(\epsilon/n)$. The learner then suffers from a $L_1$ loss equal to
\begin{align*}
\ell_t(\hat{\Lambda}_t)\coloneqq\abs{b_t-\hat{\lambda}_{a_t,t}}.
\end{align*}
We define the regret for the first $T$ iterations to be the amount by which the actual loss of the learner exceeds the loss of the best single prediction $\hat{\Lambda}$:
\begin{align}\label{eq:DefRegret}
R_T=\sum_{t=1}^T\ell_t(\hat{\Lambda}_t)-\min_{\hat{\Lambda}}\sum_{t=1}^T\ell_t(\hat{\Lambda}).
\end{align}

After receiving the feedback, the learner computes the loss function $\ell_t(\hat{\Lambda}_t)$. If the loss function is larger than a certain threshold, the learner performs an update on the hypothesis channel. In \Cref{sec:ChanEstCon}, we show that there exists a strategy such that the regret in \eqref{eq:DefRegret} is bounded by $R_T=O(\sqrt{Tn})$, yielding the fact that the number of ``bad" iterations in which $\abs{\lambda_{a_t}-\hat{\lambda}_{a_t,t}}>O(\epsilon/\sqrt{n})$ is bounded by $O(n^2/\epsilon^2)$. Therefore, the number of updates (resp. iterations) in the online learning procedure is also bounded by $O(n^2/\epsilon^2)$.

\subsection{Tree representation and Le Cam's method}
\label{sec:tree}
We use the tool of tree representation and Le cam's method when we prove the lower bounds for any ancilla-free protocols (\Cref{thm:AdaptConLow}) and $k$-ancilla protocols (\Cref{thm:AdaptConKLow}) in \Cref{sec:MeasComp} and \Cref{sec:SampComp}, respectively. We adapt the learning tree formalism of~\cite{bubeck2020entanglement,aharonov1997fault,huang2021information,chen2022exponential} to the setting of Pauli channel eigenvalue estimation. We first consider protocols without quantum memory. Given an unknown Pauli channel, a (possibly adaptive and concatenating) protocol prepares an input state, constructs a process composed of the unknown quantum channel and carefully crafted data-processing channels, and performs a POVM measurement on the output state. Since the adaptive POVM measurement performed at each time step depends on the previous outcomes, it is natural to consider a tree representation. Each node on the tree represents the current state of the classical memory. By convexity, we can assume that the system is initialized to a pure state in each round. Because we consider $T$ measurements, all the leaf nodes are at depth $T$.
\begin{definition}[Tree representation for learning Pauli channels with concatenating measurements]\label{def:TreeMemless}
Given an unknown $n$-qubit quantum channel $\Lambda$, a learning protocol without quantum memory with query access to $\Lambda$ can be expressed as a rooted tree $\cT$ of depth $T$. Each node on the tree encodes the transcript measurement outcomes the protocol has seen so far. The tree satisfies the following:
\begin{itemize}
    \item Each node $u$ is associated with a probability $p^\Lambda(u)$.
    \item For the root $r$ of the tree, $p^\Lambda(r)=1$.
    \item At each non-leaf node $u$, we measure an adaptive POVM $\cM_s^u=\{\omega_s^u 2^n\ket{\psi_s^u}\bra{\psi_s^u}\}_s$ on an adaptive input state $\ket{\phi^u}\bra{\phi^u}$ going through $M^u\geq 1$ concatenating $\Lambda$ to obtain a classical outcome $s$. Each child node $v$ of the node $u$ is connected through the edge $e_{u,s}$.
    \item If $v$ is the child node of $u$ connected through the edge $e_{u,s}$, then
    \begin{align*}
    p^\Lambda(v)=p^\Lambda(u)\omega_s^u2^n\bra{\psi_s^u}\rho_{\text{out}}\ket{\psi_s^u},
    \end{align*}
    where
    \begin{align*}
    \rho_{\text{out}}=\Lambda(\cC^{u}_{M^u-1}(\Lambda(\cdots\Lambda(\cC^{u}_{1}(\Lambda(\ket{\phi^u}\bra{\phi^u})))\cdots))).
    \end{align*}
    Here, $\cC_{1},\ldots,\cC_{M^u-1}$ are data-processing channels.
    \item Every root-to-leaf path is of length $T$. Note that for a leaf node $\ell$, $p^\Lambda(\ell)$ is the probability that the classical memory is in state $\ell$ after the learning procedure.
\end{itemize}
\end{definition}

\noindent We consider a reduction from the Pauli channel estimation task we care about to a two-hypothesis distinguishing problem.
\begin{problem}[Many-versus-one distinguishing problem]\label{prob:Distinguish}
We reduce the estimation problem to a distinguishing problem in order to prove the lower bounds. We consider the following two cases happening with equal probability.
    \begin{itemize}
    \item The unknown channel is the depolarization channel $\Lambdadep=\frac{1}{2^n}I\Tr(\cdot)$.
    \item The unknown channel is sampled uniformly from a set of channels $\{\Lambda_i\}_{i=1}^M$.
\end{itemize}
The goal of the learning protocol is to predict which event has happened.
\end{problem}

\noindent A specific task in the context of Pauli channel eigenvalue estimation is the following problem we refer to as \emph{Pauli Spike Detection: }
\begin{problem}[Pauli spike detection (search)]\label{prob:SpikeDetect}
We consider the distinguishing task between the following cases happening with equal probability.
\begin{itemize}
    \item The unknown channel is the depolarization channel $\Lambdadep=\frac{1}{2^n}I\Tr(\cdot)$.
    \item The unknown channel is sampled uniformly from a set of channels $\{\Lambda_a\}_{a\in\Z_2^{2n}\backslash \{0\}}$ with $\Lambda_a=\frac{1}{2^n}[I\Tr(\cdot)+ P_a\Tr(P_a(\cdot))]$.
\end{itemize}
In the case of the Pauli spike search problem, we are asked to further figure out the value of $a$ in the second case.
\end{problem}

\noindent The central idea for proving complexity lower bounds based on \Cref{prob:Distinguish} and \Cref{prob:SpikeDetect} is the two-point method. In this tree representation, in order to distinguish between the two events, the distribution over the classical memory on the leaves for the two events must be sufficiently distinct. Formally, we have the following:
\begin{lemma}[Le Cam's two-point method, see e.g. Lemma 1 in~\cite{yu1997assouad}]\label{lem:LeCam}
The probability that the learning protocol represented by the tree $\cT$ correctly solves the many-versus-one distinguishing task in Problem~\ref{prob:Distinguish} is upper bounded by
\begin{align*}
\frac12\sum_{\ell\in\mathrm{leaf}(\cT)}\abs{\E_i \,p^{\Lambda_i}(\ell)-p^{\Lambdadep}(\ell)}\,.
\end{align*}
\end{lemma}

\section{Efficient Pauli spike detection using two ancillas and channel concatenation}\label{sec:ConMeasProtocol}
One of the main contributions of our work is to exhibit the power for the simultaneous usage of channel concatenation and (even a constant or logarithmic size) quantum memory. Within this vein, we start with the easier task of Pauli spike detection. In \Cref{thm:AdaptConKLow_informal} (formally in \Cref{thm:AdaptConKLow}), we prove that any algorithms with $k<n$ ancilla qubit require an exponential number of queries to the unknown Pauli channel for the eigenvalue estimation problem, even if arbitrary adaptive control and channel concatenation are allowed. We prove the hardness of this task using the Pauli spike detection problem defined in \Cref{prob:SpikeDetect}. Specifically, we consider distinguishing between the following scenarios:
\begin{itemize}
    \item The unknown channel is the depolarization channel $\Lambdadep=\frac{1}{2^n}I\Tr(\cdot)$.
    \item The unknown channel is sampled uniformly from a set of channels $\{\Lambda_a\}_{a\in\Z_2^{2n}\backslash \{0\}}$ with $\Lambda_a=\frac{1}{2^n}[I\Tr(\cdot)+ P_a\Tr(P_a(\cdot))]$.
\end{itemize}
We strictly prove that any algorithm with bounded ancillas (i.e., $k\leq(1-o(1))n$) has to access the unknown channel for exponential times to distinguish between the two cases (see \Cref{thm:AdaptConKLow_informal} in the introduction and formally \Cref{thm:AdaptConKLow} in \Cref{sec:SampComp}). Ref.~\cite{chen2022quantum} also proved an exponential lower bound on the query (measurement) complexity for arbitrary ``non-concatenating" protocols. However, concatenating protocols seem much more powerful than non-concatenating ones in the sense that they can query the channel multiple times (even up to exponentially many times) with quantum post-processing before a single measurement. We remark that \Cref{thm:AdaptConKLow_informal} does not rule out the possibility of solving the Pauli estimation problem efficiently concerning the number of measurements (i.e., polynomial number of measurements) using concatenation. In the following, we provide a rigorous proof that concatenating protocols are more powerful in solving distinguishing tasks.
\begin{theorem}\label{thm:SpikeDetCon}
There exists a non-adaptive, concatenating strategy with a \textbf{two-qubit (i.e., $k=2$)} quantum memory that can solve the Pauli spike detection in \Cref{prob:SpikeDetect} with a high probability using a \textbf{single} measurement but concatenating $\tilde{O}(4^{n})$ queries to the unknown Pauli channel $\Lambda$. Here, $\tilde{O}(\cdot)$ omits the terms that have polynomial dependence on the system size $n$. 
\end{theorem}

\noindent\Cref{thm:SpikeDetCon} also indicates that the Pauli spike detection task exploited to prove the query lower bound in our work and the lower bounds in the previous related works~\cite{chen2022quantum} are solvable with polynomial measurement complexity, or more specifically, even a single measurement, using concatenated learning protocols with $O(1)$ ancilla qubits. This yields an exponential quantum speedup in measurement complexity compared to either non-concatenating protocols or ancilla-free protocols. In the rest of this section, we give the proof for \Cref{thm:SpikeDetCon}.

\begin{proof}
For simplicity, we encode all the Pauli strings except $I^{\otimes n}$ using $a\in\{1,\ldots,4^n-1\}$. We denote $I^{\otimes n}$ as $I$ in the following analysis for simplicity. We also assign indices to the working qubits from $1$ to $n$, and the two ancilla qubits are denoted as the $(n+1)$-th and the $(n+2)$-th qubit. 

Given an arbitrary $a$, we now consider two quantum states
\begin{align*}
\rho_{\text{dep}}=\frac{I}{2^n}\quad \text{v.s.}\quad\rho_a=\frac{I+\lambda_a P_a}{2^n}.
\end{align*}
If we measure these two states using the POVM along the eigenvectors $\left\{\ket{\varphi_a^j}\bra{\varphi_a^j}\right\}_{j=1}^{2^n}$ corresponding to eigenvalue $\pm 1$ of the Pauli string $P_a$, we will obtain the distributions 
\begin{align}\label{eq:POVMDistr}
\left(\frac{1}{2^n},...,\frac{1}{2^n},\frac{1}{2^n},...,\frac{1}{2^n}\right)\quad\text{v.s.}\quad\left(\frac{1+\lambda_a}{2^n},...,\frac{1+\lambda_a}{2^n},\frac{1-\lambda_a}{2^n},...,\frac{1-\lambda_a}{2^n}\right),
\end{align}
where we have half of entries being $\frac{1+\lambda_a}{2^n}$ and the other half being $\frac{1-\lambda_a}{2^n}$ in the second case. The total variation distance between two distributions is $O(\lambda_a)$. In the case of Pauli spike detection, we fix $\lambda_a=1$, the distributions after the POVM will be
the distributions 
\begin{equation*}
\left(\frac{1}{2^n},...,\frac{1}{2^n},\frac{1}{2^n},...,\frac{1}{2^n}\right)\quad\text{v.s.}\quad\left(\frac{1}{2^{n-1}},...,\frac{1}{2^{n-1}},0,...,0\right),
\end{equation*}
where we have half of entries being $\frac{1}{2^{n-1}}$ and the other half being $0$ in the second case. The total variation distance between two distributions is $\frac12=O(1)$. In the following, we denote $\left\{\ket{\varphi_a^{j,+}}\bra{\varphi_a^{j,+}}\right\}_{j=1}^{2^{n-1}}$ as the set of eigenstates corresponding to positive eigenvalue $+1$ of $P_a$ and $\left\{\ket{\varphi_a^{j,-}}\bra{\varphi_a^{j,-}}\right\}_{j=2^{n-1}+1}^{2^{n}}$ corresponding to negative eigenvalue $-1$ of $P_a$. We consider the two-outcome POVM given by \begin{align}\label{eq:POVMBasis}
\{E_{\text{data},a},I-E_{\text{data},a}\}=\left\{\sum_{j=1}^{2^{n-1}}\ket{\varphi_a^{j,+}}\bra{\varphi_a^{j,+}},\sum_{j=2^{n-1}+1}^{2^n}\ket{\varphi_a^{j,-}}\bra{\varphi_a^{j,-}}\right\},
\end{align}
the first case will give two outcomes with probability $\left(\frac12,\frac12\right)$ while the second case will always give the first outcome. 

To record this information to the ancilla qubit, we consider a control channel exhibited in \Cref{exp:ControlChan} using the POVM $\{E_{\text{data},a},I-E_{\text{data},a}\}$. Suppose we are given a $(n+1)$-qubit state $\rho_n\otimes \rho_1$ with $\rho_n$ an $n$-qubit state and $\rho_1$ an ancilla qubit. The control channel acts as
\begin{align*}
\mathcal{C}(\rho)=\Tr(E_{\text{data},a}\rho_n)\cdot\frac{I}{2^n}\otimes \rho_1+\Tr((I-E_{\text{data},a})\rho_n)\cdot\frac{I}{2^n}\otimes \ket{1}\bra{1}.
\end{align*}
Suppose we are given the pre-knowledge that $\rho_n$ is either $\frac{I}{2^n}$ and $\frac{I+P_a}{2^n}$, we prepare the ancilla qubit as $\rho_1=\ket{0}\bra{0}$. After inputting these $n+1$ qubits into $\mathcal{C}$, we obtain $\frac{I}{2}$ on the ancilla qubit in the first case and obtain $\ket{0}\bra{0}$ in the second case. Intuitively, this process records the information of the $n$ working qubits to the ancilla qubit. Following this intuition, we formally consider the following protocol in \Cref{algo:AncillaConDis}:

\begin{algorithm}[htbp]
\caption{The protocol for Pauli spike detection}
\label{algo:AncillaConDis}
\begin{algorithmic}[1]
\REQUIRE An unknown $n$-qubit Pauli channel $\Lambda$.
\ENSURE Distinguish whether the unknown channel is the depolarization channel $\Lambdadep$ or is sampled uniformly from $\{\Lambda_a\}_{a\in\Z_2^{2n}\backslash \{0\}}$ with $\Lambda_a=\frac{1}{2^n}[I\Tr(\cdot)+ P_a\Tr(P_a(\cdot))]$.
\STATE We prepare an initial state $\frac{I+P_1}{2^n}$ on the $n$ working qubits and $\ket{0}\bra{0}$ on the first and the second ancilla qubit (quantum memory). The overall initial state is $\frac{I+P_1}{2^n}\otimes\ket{0}\bra{0}\otimes\ket{0}\bra{0}$.
\FOR{Iteration $t=1,...,4^n-1$}
\STATE We start with the initial state $$\frac{I+P_t}{2^n}\otimes\ket{0}\bra{0}\otimes\ket{0}\bra{0}.$$ 
\FOR{Epoch $k=1,...,3n$}
\STATE We input the current state into $\Lambda_n\otimes I$, which is the unknown channel on the $n$ working qubits and the identity channel on the two ancilla qubits.
\STATE We then perform the \textbf{$(n+1)$-qubit iteration-$t$ data record channel $\mathcal{E}_{\text{data},t}$} on the first $(n+1)$ qubits based on the two-outcome POVM (defined in \eqref{eq:POVMBasis}). It is a control channel (defined in \Cref{exp:ControlChan}) given by $\{E_{\text{data},t},I-E_{\text{data,t}}\}=\left\{\sum_{j=1}^{2^{n-1}}\ket{\varphi_t^{j,+}}\bra{\varphi_t^{j,+}},\sum_{j=2^{n-1}+1}^{2^n}\ket{\varphi_t^{j,-}}\bra{\varphi_t^{j,-}}\right\}$. Given a $(n+1)$-qubit state $\rho_n\otimes \rho_1$ with $\rho_n$ an $n$-qubit state on the working qubits and $\rho_1$ the state on ancilla qubit, the channel acts as
\begin{align}\label{chan:TDataRec}
\mathcal{E}_{\text{data},t}:\rho\to \frac{I}{2^n}\otimes \rho_1\Tr\left(E_{\text{data},t}\cdot\rho_n\right)+\frac{I}{2^n}\otimes\ket{1}\bra{1}\Tr\left((I-E_{\text{data},t})\cdot\rho_n\right).
\end{align}
\STATE We perform the working qubit initialization channel 
\begin{align}\label{chan:WorkInitial}
\mathcal{E}_{\text{work ini},t}:\rho_n\to\frac{I+P_t}{2^n}\Tr(\rho_n) 
\end{align}
on the first $n$ qubits with $\rho_n$ the reduced density matrix on these qubits and enter the next epoch.
\ENDFOR
\STATE We perform the \textbf{two-qubit ancilla recording channel $\mathcal{E}_\text{anc}$}, which is also a control channel, on the two ancilla qubits. For product state $\rho=\rho_1\otimes\rho_2$ where $\rho_1$ and $\rho_2$ denote the density matrix on the first one and the second ancilla qubit, the channel acts as
\begin{align}\label{chan:AncRec}
\mathcal{E}_\text{anc}:\rho\to \frac{I}{2}\otimes\rho_2\Tr(\ket{1}\bra{1}\rho_1)+\frac{I}{2}\otimes\ket{1}\bra{1}\Tr(\ket{0}\bra{0}\rho_1)
\end{align}
\STATE We perform a channel $\mathcal{E}_{\text{prep},t}$ preparing the input state for the next iteration $t+1$: 
\begin{align*}
\mathcal{E}_{\text{prep},t}: \rho\to\frac{I+P_{t+1}}{2^n}\otimes\ket{0}\bra{0}\Tr(\rho)
\end{align*}
to the $n$ working qubits and the first ancilla qubit.
\ENDFOR
\STATE We measure the second ancilla qubit using the $\{\ket{0}\bra{0},\ket{1}\bra{1}\}$ basis. The outcome $0$ and $1$ corresponds to the first and second case of the distinguishing problem.
\end{algorithmic}
\end{algorithm}

We now analyze the mechanism of this protocol and prove \Cref{thm:SpikeDetCon}. 

\paragraph{Case 1: }In the first case of the distinguishing task, the unknown quantum channel is $\Lambdadep$. We observe that the input state on the $n$ working qubits in each epoch of each iteration $t=1,\ldots,4^n-1$ is $\frac{I+P_t}{2^n}$. Therefore, the output state after the unknown channel $\Lambdadep$ on the working qubits is
\begin{align*}
\Lambdadep\left(\frac{I+P_t}{2^n}\right)=\frac{I}{2^n}.
\end{align*}
Before entering $\mathcal{E}_{\text{data},t}$, the state on the first $(n+1)$ qubits is $\frac{I}{2^n}\otimes \ket{0}\bra{0}$. For simplicity, we denote the quantum process in each epoch of the iteration $t$ on the first $(n+1)$ qubits as $\mathcal{E}_{\text{epo},t}(\cdot)$. Notice that for repeatedly implementing the epochs, we have
\begin{align*}
&\mathcal{E}_{\text{data},t}\left(\frac{I}{2^n}\otimes \ket{0}\bra{0}\right)=\frac{I}{2^n}\otimes \left(\frac{1}{2}\ket{0}\bra{0}+\frac{1}{2}\ket{1}\bra{1}\right),\\
&\mathcal{E}_{\text{epo},t}^k\left(\frac{I+P_t}{2^n}\otimes \ket{0}\bra{0}\right)=\frac{I}{2^n}\otimes \left(\frac{1}{2^k}\ket{0}\bra{0}+\left(1-\frac{1}{2^k}\right)\ket{1}\bra{1}\right),
\end{align*}
for $k=1,...,3n$. After $3n$ repetitions, we have the following state
\begin{align*}
\mathcal{E}_{\text{epo},t}^{3n}\left(\frac{I+P_t}{2^n}\otimes \ket{0}\bra{0}\right)=\frac{I}{2^n}\otimes \left(\frac{1}{8^n}\ket{0}\bra{0}+\left(1-\frac{1}{8^n}\right)\ket{1}\bra{1}\right)
\end{align*}
on the first $n+1$ qubits. Next, we consider the state of the two ancilla qubits after we perform the $\mathcal{E}_\text{anc}$ channel. It reads:
\begin{align*}
\frac{I}{2}\otimes\left(\left(1-\frac{1}{8^n}\right)^{t}\ket{0}\bra{0}+\left(1-\left(1-\frac{1}{8^n}\right)^t\right)\ket{1}\bra{1}\right).
\end{align*}
After we implement the channel $\mathcal{E}_{\text{prep},t}$, we enter the next iteration $(t+1)$ with (we allow $t=0$ at the beginning of the first iteration) 
\begin{align*}
\frac{I+P_{t+1}}{2^n}\otimes\ket{0}\bra{0}\otimes\left(\left(1-\frac{1}{8^n}\right)^{t}\ket{0}\bra{0}+\left(1-\left(1-\frac{1}{8^n}\right)^t\right)\ket{1}\bra{1}\right).
\end{align*}

After we enumerate all $a$'s, the final state on the second ancilla qubit should be
\begin{align}\label{eq:SpikeCase1}
\left(\left(1-\frac{1}{8^n}\right)^{4^n-1}\ket{0}\bra{0}+\left(1-\left(1-\frac{1}{8^n}\right)^{4^n-1}\right)\ket{1}\bra{1}\right),
\end{align}
For large enough $n$, we can ensure that this state has a $1-\exp\left(-\Omega(n)\right)$ coefficient concerning the entry $\ket{0}\bra{0}$ in the density matrix. That is to say, we obtain the measurement outcome $\ket{1}\bra{1}$ with an exponentially small probability when we measure it using the $\{\ket{0}\bra{0},\ket{1}\bra{1}\}$ basis.

\paragraph{Case 2: }In the second case of the distinguishing task, the unknown channel $\Lambda_a$ is sampled uniformly from all $a$'s. Notice that $\Lambda_a$ acts the same with $\Lambdadep$ on the quantum state $\frac{I+P_{a'}}{2^n}$ for $a'\neq a$. Thus in the $t=1,\ldots,a-1$ iterations, the input state is 
\begin{align*}
\frac{I+P_t}{2^n}\otimes\ket{0}\bra{0}\otimes\left(\left(1-\frac{1}{8^n}\right)^{t-1}\ket{0}\bra{0}+\left(1-\left(1-\frac{1}{8^n}\right)^{t-1}\right)\ket{1}\bra{1}\right)
\end{align*}
for each iteration $t$, the first $n$ qubits of the output state after the unknown channel $\Lambda_a$ is
\begin{align*}
\Lambda_a\left(\frac{I+P_t}{2^n}\right)=\frac{I}{2^n}.
\end{align*}
Similar to the first case, our protocol will end in the quantum state
\begin{align*}
\frac{I+P_{t+1}}{2^n}\otimes\ket{0}\bra{0}\otimes\left(\left(1-\frac{1}{8^n}\right)^{t}\ket{0}\bra{0}+\left(1-\left(1-\frac{1}{8^n}\right)^t\right)\ket{1}\bra{1}\right)
\end{align*} 
and enter the next ($(t+1)$-th) iteration. 

However, in the $a$-th iteration, the first $n$ working qubits of the output state after the channel $\Lambda_a$ is
\begin{align*}
\Lambda_a\left(\frac{I+P_a}{2^n}\right)=\frac{I+P_a}{2^n}.
\end{align*}
We then consider the effect of the data processing channel $\mathcal{E}_{\text{data},a}$ and the channel in each epoch $\mathcal{E}_{\text{epo},a}$ under this setting. In particular, we have
\begin{align*}
&\mathcal{E}_{\text{data},a}\left(\frac{I+P_a}{2^n}\otimes \ket{0}\bra{0}\right)=\frac{I}{2^n}\otimes \ket{0}\bra{0},\\
&\mathcal{E}_{\text{epo},a}^k\left(\frac{I+P_a}{2^n}\otimes \ket{0}\bra{0}\right)=\frac{I}{2^n}\otimes \ket{0}\bra{0},
\end{align*}
for $k=1,...,3n$. Thus, we will obtain the following quantum state on the first $(n+1)$ qubits after the $3n$ epochs:
\begin{align*}
\mathcal{E}_{\text{epo},a}^{3n}\left(\frac{I+P_a}{2^n}\otimes \ket{0}\bra{0}\right)=\frac{I}{2^n}\otimes \ket{0}\bra{0}.
\end{align*}
We then apply the channel $\mathcal{E}_\text{anc}$ on the two ancilla qubits. The output state reads:
\begin{align*}
\mathcal{E}_\text{anc}\left(\ket{0}\bra{0}\otimes\left(\left(1-\frac{1}{8^n}\right)^{a-1}\ket{0}\bra{0}+\left(1-\left(1-\frac{1}{8^n}\right)^{a-1}\right)\ket{1}\bra{1}\right)\right)=\frac{I}{2}\otimes\ket{1}\bra{1}.
\end{align*}
After we implement the channel $\mathcal{E}_{\text{prep},a}$, we enter the next iteration $(a+1)$ with 
\begin{align*}
\frac{I+P_{a+1}}{2^n}\otimes\ket{0}\bra{0}\otimes\ket{1}\bra{1}.
\end{align*}

In the $t=a+1,\ldots,4^n-1$-th turn, the input state is then $\frac{I+P_t}{2^n}\otimes\ket{0}\bra{0}\otimes\ket{1}\bra{1}$ at the beginning of each iteration. Therefore, the first $n$ qubits of the output state after the unknown channel $\Lambda_a$ is
\begin{align*}
\Lambda_a\left(\frac{I+P_t}{2^n}\right)=\frac{I}{2^n}
\end{align*}
as $\Tr(P_aP_{a'})=0$ for $a\neq a'$. Again, we obtain
\begin{align*}
\mathcal{E}_{\text{epo},t}^{3n}\left(\frac{I+P_t}{2^n}\otimes \ket{0}\bra{0}\right)=\frac{I}{2^n}\otimes\left(\frac{1}{8^n}\ket{0}\bra{0}+\left(1-\frac{1}{8^n}\right)\ket{1}\bra{1}\right)
\end{align*}
on the first $(n+1)$ qubits. However, the second ancilla qubit will maintain $\ket{1}\bra{1}$ after we implement $\mathcal{E}_\text{anc}$ on the two ancilla qubits. Finally, we will always enter the next iteration with the state
\begin{align}\label{eq:SpikeCase2}
\frac{I+P_{t+1}}{2^n}\otimes\ket{0}\bra{0}\otimes\ket{1}\bra{1}.
\end{align}

After we enumerate all $a$'s, the final state of the second ancilla qubit (as shown in \eqref{eq:SpikeCase1} and \eqref{eq:SpikeCase2}) after the $(4^n-1)$-th iteration will keep in $\ket{1}\bra{1}$.

Therefore, we can distinguish between the two cases with probability $1-\exp(-\Omega(n))$ after we measure the ancilla qubit using the $\{\ket{0}\bra{0},\ket{1}\bra{1}\}$ basis. The number of the unknown Pauli channel $\Lambda$ used is $3n(4^n-1)=\tilde{O}(4^n)$. This finishes the proof for \Cref{thm:SpikeDetCon}.
\end{proof}

\paragraph{Learning the index $a$ in the distinguishing task. }Starting from the protocol above for distinguishing between the two cases of depolarization channel $\Lambdadep$ and Pauli channel with one uniformly picked spike $\Lambda_a=\frac{I+P_a\Tr(P_a\cdot)}{2^n}$, we now consider a slightly harder distinguishing problem, which we refer to as the \textbf{Pauli Spike Search} problem (as defined in \Cref{prob:SpikeDetect}). Given an unknown channel $\Lambda$, we not only want to distinguish between the two cases
\begin{itemize}
    \item The unknown channel is the depolarization channel $\Lambdadep=\frac{1}{2^n}I\Tr(\cdot)$
    \item The unknown channel is of the form $\Lambda_a=\frac{1}{2^n}[I\Tr(\cdot)+ P_a\Tr(P_a(\cdot))]$ for some $a\in\mathbb{Z}^{2n}_2 \backslash \{0\}$
\end{itemize}
but we also want to determine exactly what $a$ is in the latter case.

For this problem, we can obtain the following upper bound. We first use a single measurement to distinguish between the two cases using \Cref{thm:SpikeDetCon}. For the second case, we use $2n$ measurements to decide the $i=1,\ldots,n$ bit of the Pauli string $P_a$ using binary search. In particular, we enumerate all Pauli strings with the $i$-th qubit being $I,X,Y,Z$ using a protocol similar to \Cref{thm:SpikeDetCon} and \Cref{algo:AncillaConDis} except that we enumerate a subset of $a$'s instead of all $a$'s. We formalize the result as:
\begin{theorem}\label{thm:Charsingle}
Given an unknown channel $\Lambda$ in one of the following cases:
\begin{itemize}
    \item The unknown channel is the depolarization channel $\Lambdadep=\frac{1}{2^n}I\Tr(\cdot)$.
    \item The unknown channel is chosen from a set of channels $\{\Lambda_a\}_{a\in\Z_2^{2n}\backslash \{0\}}$ with $\Lambda_a=\frac{1}{2^n}[I\Tr(\cdot)+ P_a\Tr(P_a(\cdot))]$.
\end{itemize}
There exists a non-adaptive, concatenating strategy with a two-qubit quantum memory that can decide which case it is and further obtain $a$ for the second case with high probability using $O(n)$ measurements and $\tilde{O}(4^{n})$ queries to $\Lambda$.
\end{theorem}

\section{Efficient Pauli channel estimation using logarithmic ancillas and channel concatenation}\label{sec:ChanEstCon}
In the previous section, we have shown that the easier task of Pauli spike detection in \Cref{prob:SpikeDetect}, which is shown to require an exponential number of measurements for any protocols with only $(1-o(1))n$ ancilla qubits or only channel concatenations, can be solved using a single measurement if we have only $O(1)$ ancillas and channel concatenations. Through this simple but artificial example, we have already obtained evidence that one can obtain exponential speedups for learning properties of quantum channels using $k<(1-o(1))n$ ancilla qubits due to the usage of concatenation, which is in sharp contrast to the case of learning quantum states. In this section, we go back to the original Pauli channel eigenvalue estimation problem in \Cref{prob:def}. We prove that even for learning Pauli channels within a constant error, the exponential reduction in measurement complexity can be obtained using $k=O(\log n)$ ancilla qubits. For convenience, we restate \Cref{prob:def} here. Given an unknown Pauli channel 
\begin{align}
\Lambda(\cdot)=\frac{I\Tr(\cdot)+\sum_{a=1}^{4^n-1}\lambda_aP_a\Tr(P_a(\cdot))}{2^n},
\end{align}
our goal is to learn $\{\hat{\lambda}_a\}$ in the infinity norm such that $\abs{\lambda_a-\hat{\lambda}_a}\leq\epsilon$ for any $a\in\{1,...,4^n-1\}$. 

\begin{theorem}\label{thm:ChanEstCon}
Given accuracy demand $\epsilon$, there exists a non-adaptive, concatenating strategy with $k=O(\log n/\epsilon^2)$ ancilla qubits (logarithmic quantum memory) that solves the Pauli channel eigenvalue estimation with a high probability using $\tilde{O}(n^2/\epsilon^2)$ measurements. The protocol uses up to $\tilde{O}(4^n/\epsilon^2)$ levels of channel concatenation in each measurement to the unknown Pauli channel $\Lambda$.
\end{theorem}

\noindent In the regime of $\epsilon=\Theta(1)$, our protocol in \Cref{thm:ChanEstCon}, using only a logarithmic number of ancilla qubits, provides an exponential speedup on the measurement complexity compared to any ancilla-free strategies or concatenation-free strategies with $(1-o(1))n$-qubit quantum memory. 

The remainder of this section is organized as follows. We first provide some tools and intermediate results for the proof. In \Cref{sec:SingleComp}, we move a step forward from \Cref{sec:ConMeasProtocol} and propose an algorithm that can learn any Pauli channel with a single nonzero $\lambda_a$. In \Cref{sec:ThresholdSearch}, we provide the formal definition of Pauli channel threshold search and propose a protocol to solve this task based on \Cref{sec:SingleComp}. In \Cref{sec:OnlinePauli}, we provide the performance analysis for online learning Pauli channels. Finally, we wrap up all the intermediate results and provide the proof for \Cref{thm:ChanEstCon} in \Cref{sec:Wrap}.

\subsection{Learning single component (positive) Pauli channels}\label{sec:SingleComp}
In this part, we consider a special case of the Pauli channel learning problem in \Cref{prob:def}, which is also an extension of the spike detection (resp. search) problem in \Cref{prob:SpikeDetect}. Specifically, we consider the following problem:
\begin{problem}[Distinguishing (Learning) single component (positive) Pauli channels]\label{prob:SingleComp}
We consider the distinguishing task between the following cases happening with equal probability
\begin{itemize}
    \item The unknown channel is the depolarization channel $\Lambdadep=\frac{I}{2^n}\Tr(\cdot)$.
    \item The unknown channel is sampled uniformly from the set of channels $\{\Lambda_a\}_{a\in\mathbb{Z}_2^{2n}\backslash\{0\}}$ with $\Lambda_a=\frac{I}{2^n}[\Tr(\cdot)+\lambda_aP_a\Tr(P_a(\cdot))]$ for some $\lambda_a\in(\epsilon,1]$.
\end{itemize}
For the case of learning single component Pauli channels, we are asked to further provide the value of $a$ in the second case.
\end{problem}

\noindent When $\lambda_a=1$, \Cref{prob:SingleComp} reduces to \Cref{prob:SpikeDetect}. A straightforward intuition is to employ \Cref{algo:AncillaConDis}, which can solve \Cref{prob:SpikeDetect} using a single measurement, to solve \Cref{prob:SingleComp}. Unfortunately, this intuition does not work. When $\lambda_a\lesssim 1-\Omega(1/n)$, we can observe that repeating $\mathcal{E}_{\text{epo},a}$ for $3n$ times will also result in an exponentially small coefficient for $\ket{0}\bra{0}$ for the second case. Although we can make some alternation on the value of epoch number $m$, no matter what $m$ we choose we can not separate the two cases when $\lambda_a\leq 1-\Omega(1/n)$. To address this issue, we propose a new algorithm with the following performance guarantee:

\begin{lemma}\label{lem:LearnSingle}
There exists a non-adaptive, concatenating strategy (\Cref{algo:LearnSingle}) with $O(\log n/\epsilon^2)$ ancilla qubits that can distinguish (resp. learn) the single component Pauli channel in \Cref{prob:SingleComp} with a high probability using $O(1)$ (resp. $O(n)$) measurements. The protocol is explicit with $\tilde{O}(4^n/\epsilon^2)$ queries to the unknown Pauli channel $\Lambda$ before each measurement.
\end{lemma}

\noindent The brief idea to address the issue in \Cref{algo:AncillaConDis} in solving \Cref{prob:SingleComp} is to perform a special quantum state purification to boost the purity of the ancilla qubit. Before we provide the proof for \Cref{lem:LearnSingle}, we first introduce our special quantum state purification process.

\subsubsection{The purification process}\label{sec:SpecialPurification}
We now consider the purification task of a noisy single-qubit state in the depolarization noise assumption
\begin{align*}
\rho=(1-\lambda)\frac{I}{2}+\lambda\ket{0}\bra{0}
\end{align*}
at $\lambda>0$ into some state closer to $\ket{0}\bra{0}$ using $N$ copies of $\rho$. 
\begin{remark}\label{rem:DiffPurification}
We remark that the purification task considered here is significantly different from the standard purification task~\cite{cirac1999optimal,keyl2001rate,childs2023streaming}. In these works, the noisy states are $\rho=(1-\lambda)\frac{I}{2}+\lambda\ket{0}_{\hat{n}}\bra{0}_{\hat{n}}$ along some \textbf{unknown} direction $\hat{n}$ on the Block sphere. It is proved that any purification protocol that purifies the state within $\epsilon$ distance from $\ket{0}_{\hat{n}}\bra{0}_{\hat{n}}$ requires $\Omega(1/\epsilon)$~\cite{cirac1999optimal}. However, in our setting, we know the target pure state is $\ket{0}\bra{0}$. Thus we are able to surpass this limitation.
\end{remark}

\noindent Suppose we are given $N$ copies of noisy state $\rho=(1-\lambda)\frac{I}{2}+\lambda\ket{0}\bra{0}=\frac{1+\lambda}{2}\ket{0}\bra{0}+\frac{1-\lambda}{2}\ket{1}\bra{1}$. We can extend $\rho^{\otimes N}$ along the computational basis. It is straightforward to observe that the density matrix under the computational basis for this state is diagonal. In particular, we consider a bit string $x$ with $\abs{x}$ $1$'s, the coefficient for $\ket{x}\bra{x}$ follows the binomial distribution. The state can be written as
\begin{align*}
\rho^{\otimes N}=\sum_{x\in\{0,1\}^N}\binom{N}{\abs{x}}p^{\abs{x}}(1-p)^{N-\abs{x}}\ket{x}\bra{x},
\end{align*}
where $p=(1-\lambda)/2$. Now, we consider the following two-outcome POVM $\{E_\text{pure}^N,I-E_\text{pure}^N\}$ with
\begin{align}\label{eq:PurificationPOVM}
E_\text{pure}^N=\sum_{x\in\{0,1\}^N:\abs{x}\leq N/2}\ket{x}\bra{x},\quad I-E_\text{pure}^N=\sum_{x\in\{0,1\}^N:\abs{x}>N/2}\ket{x}\bra{x}.
\end{align}

When $\lambda>0$,the probability corresponding to the element $E_\text{pure}^N$ is lower-bounded by
\begin{align}\label{eq:PurificationPosLower}
\Tr(E_\text{pure}^N\rho^{\otimes N})\geq1-\exp(-N\lambda^2/2)
\end{align}
according to the concentration inequality in \eqref{eq:Concentration}. It is also upper-bounded by
\begin{align}\label{eq:PurificationPosUpper}
\Tr(E_\text{pure}^N\rho^{\otimes N})\leq\frac{1}{2}\left(1+\sqrt{1-\exp\left(-\frac{N\lambda^2}{1-\lambda^2}\right)}\right)
\end{align}
according to the anti-concentration inequality in \eqref{eq:AntiConcentration}.

Although we define $\lambda>0$, we can also consider the case when $\lambda\leq 0$. At $\lambda=0$, we have
\begin{align}\label{eq:PurificationZero}
\Tr(E_\text{pure}^N\rho^{\otimes N})=\Tr((I-E_\text{pure}^N)\rho^{\otimes N})=\frac12.
\end{align}
And at $\lambda<0$, we always have
\begin{align}\label{eq:PurificationNeg}
\Tr(E_\text{pure}^N\rho^{\otimes N})\leq\frac{1+\lambda}{2}\leq\frac 12.
\end{align}

In the following, we consider a $(N+1)$-qubit special state purification channel acting on $N$ copies of $\rho$ with an ancilla qubit to record the information in \Cref{algo:SpecialPurification}.

\begin{algorithm}[htbp]
\caption{Special quantum state purification ({\sf StatePurification($N$)})}
\label{algo:SpecialPurification}
\begin{algorithmic}[1]
\REQUIRE $N$ copies of noisy quantum state $\rho^{\otimes N}$ where $\rho=(1-\lambda)\frac{I}{2}+\lambda\ket{0}\bra{0}$ with an ancilla qubit $\rho_\text{anc}$. 
\ENSURE Output $\Tr(E_\text{pure}^N\rho^{\otimes N})\ket{0}\bra{0}+\left(1-\Tr(E_\text{pure}^N\rho^{\otimes N})\right)\ket{1}\bra{1}$ when the ancilla qubit is initialized with $\rho_\text{anc}=\ket{0}\bra{0}$, where $E_\text{pure}^N$ is defined in \eqref{eq:PurificationPOVM}.
\STATE We implement a control channel (defined in \Cref{exp:ControlChan}) $\mathcal{E}_\text{pure}(\cdot)$ on the $N$ noisy copies and the ancilla qubit as:
\begin{align*}
\mathcal{E}_\text{pure}(\cdot):\rho^{\otimes N}\otimes\rho_\text{anc}\to\Tr(E_\text{pure}^N\rho^{\otimes N})\cdot\frac{I}{2^N}\otimes\rho_\text{anc}+\Tr((I-E_\text{pure}^N)\rho^{\otimes N})\cdot\frac{I}{2^N}\otimes\ket{1}\bra{1}.
\end{align*}
\end{algorithmic}
\end{algorithm}

\noindent When we initialize the ancilla qubit with $\rho_\text{anc}=\ket{0}\bra{0}$, we can directly verify that this channel outputs $\Tr(E_\text{pure}^N\rho^{\otimes N})\ket{0}\bra{0}+\left(1-\Tr(E_\text{pure}^N\rho^{\otimes N})\right)\ket{1}\bra{1}$ with 
\begin{align*}
\Tr(E_\text{pure}^N\rho^{\otimes N})\geq1-\exp(-N\lambda^2/2).
\end{align*}

\subsubsection{The algorithm}
We now provide the algorithm and the proof for \Cref{lem:LearnSingle}. Similar to \Cref{sec:ConMeasProtocol}, we encode all the Pauli strings except $I^{\otimes n}$ using $a\in\{1,\ldots,4^n-1\}$. We denote $I^{\otimes n}$ as $I$ in the following analysis. We fix the number of ancilla qubits as $k=\lceil2\log n/\epsilon^2\rceil+2$. We denote $N=2\lceil\log n/\epsilon^2\rceil=k-2$. We also assign indices to the working qubits from $1$ to $n$, and the $k$ ancilla qubits are denoted as the $(n+1)$-th to the $(n+N+2)$-th qubit. 

Given an arbitrary $a$, we still consider two quantum states $\rho_{\text{dep}}=\frac{I}{2^n}$ and $\rho_a=\frac{I+\lambda_a P_a}{2^n}$. As illustrated in the proof for \Cref{thm:SpikeDetCon} (i.e.,~\eqref{eq:POVMDistr}), if we measure these two states using the POVM along the eigenvectors $\left\{\ket{\varphi_a^j}\bra{\varphi_a^j}\right\}_{j=1}^{2^n}$ corresponding to eigenvalue $\pm 1$ of the Pauli string $P_a$, we will obtain the distributions 
$\left(\frac{1}{2^n},...,\frac{1}{2^n},\frac{1}{2^n},...,\frac{1}{2^n}\right)$ and $\left(\frac{1+\lambda_a}{2^n},...,\frac{1+\lambda_a}{2^n},\frac{1-\lambda_a}{2^n},...,\frac{1-\lambda_a}{2^n}\right)$ for the two states, where we have half of entries being $\frac{1+\lambda_a}{2^n}$ and the other half being $\frac{1-\lambda_a}{2^n}$ in the second case. The total variation distance between two distributions is $O(\lambda_a)$. We inherit the notations $\left\{\ket{\varphi_a^{j,+}}\bra{\varphi_a^{j,+}}\right\}_{j=1}^{2^{n-1}}$ as the set of POVM elements corresponding to positive eigenvalue $+1$ of $P_a$ and $\left\{\ket{\varphi_a^{j,-}}\bra{\varphi_a^{j,-}}\right\}_{j=2^{n-1}+1}^{2^{n}}$ corresponding to negative eigenvalue $-1$ of $P_a$. We consider the two-outcome POVM given by $\{E_{\text{data},a},I-E_{\text{data},a}\}=\left\{\sum_{j=1}^{2^{n-1}}\ket{\varphi_a^{j,+}}\bra{\varphi_a^{j,+}},\sum_{j=2^{n-1}+1}^{2^n}\ket{\varphi_a^{j,-}}\bra{\varphi_a^{j,-}}\right\}$ as in \eqref{eq:POVMBasis}, the first case will give two outcomes with probability $\left(\frac12,\frac12\right)$ while the second case will always give the probability distribution $\left(\frac{1+\lambda_a}{2},\frac{1-\lambda_a}{2}\right)$. Based on this observation, we consider the following protocol in the following \Cref{algo:LearnSingle}:

\begin{algorithm}[htbp]
\caption{The protocol for distinguishing (learning) single component Pauli channels using $k=O(\log n/\epsilon^2)$ ancilla qubits ({\sf LearnSingle($n$,$\epsilon$)})}
\label{algo:LearnSingle}
\begin{algorithmic}[1]
\REQUIRE An unknown $n$-qubit Pauli channel $\Lambda$ and an accuracy demand $\epsilon$.
\ENSURE Distinguish whether the unknown channel is the depolarization channel $\Lambdadep$ or is sampled uniformly from $\{\Lambda_a\}_{a\in\Z_2^{2n}\backslash \{0\}}$ with $\Lambda_a=\frac{1}{2^n}[I\Tr(\cdot)+ \lambda_a P_a\Tr(P_a(\cdot))]$ for $\lambda_a\in(\epsilon,1]$.
\STATE We initialize the working qubits with $\frac{I+P_1}{2^n}$ and the $k$ ancilla qubits with $\ket{0}\bra{0}^{\otimes k}$. We set $k=\lceil 2\log n/\epsilon^2\rceil+2$ and $N=k-2=\lceil 2\log n/\epsilon^2\rceil$.
\FOR{Iteration $t=1,...,4^n-1$}
\STATE We start with the initial state $\frac{I+P_t}{2^n}\otimes \ket{0}\bra{0}^{\otimes N+1}$ in the first $N+1$ qubits 
\FOR{Epoch $k=1,...,3n$}
\STATE We input the current state into $\Lambda_n\otimes I$, which is the unknown channel on the $n$ working qubits and the identity channel on the $k$ ancilla qubits.
\STATE We apply a sequences of $(n+1)$-qubit quantum channels $\mathcal{E}_{\text{data},t}$ (defined in \eqref{chan:TDataRec}) acting on the working qubits and the $j$-th ancilla qubit for $j=1,...,N$. After each $\mathcal{E}_{\text{data},t}$ at $j=1,...,N-1$, we re-initialize the working qubits using the working qubit initialization channel $\mathcal{E}_{\text{work ini},t}$ defined in \eqref{chan:WorkInitial} and input the current state into $\Lambda_n\otimes I$ (i.e., query the unknown channel $\Lambda$ again).
\STATE Next, we perform {\sf StatePurification($N$)} (\Cref{algo:SpecialPurification}) on the first $N$ ancilla qubits and the $(N+1)$-th ancilla qubit. The $(N+1)$-th ancilla qubit is chosen as the ancilla qubit in {\sf StatePurification($N$)}.
\STATE We re-initialize the first $n+N$ qubits into $\frac{I+P_t}{2^n}\otimes \ket{0}\bra{0}^{\otimes N}$ using the following epoch initialization channel $\mathcal{E}_{\text{epo ini},t}$:
\begin{align}\label{chan:EpoT}
\mathcal{E}_{\text{epo ini},t}:\rho\to\frac{I+P_t}{2^n}\otimes \ket{0}\bra{0}^{\otimes N}\Tr(\rho).
\end{align}
\ENDFOR
\STATE We perform the channel $\mathcal{E}_\text{anc}$ (defined in \eqref{chan:AncRec}) on the last two ancilla qubits.
\STATE We perform a channel $\mathcal{E}_{\text{prep},t}$ preparing the input state for the next iteration $t+1$: 
\begin{align}\label{chan:PrepT}
\mathcal{E}_{\text{prep},t}: \rho\to\frac{I+P_{t+1}}{2^n}\otimes\ket{0}\bra{0}^{\otimes N+1}\Tr(\rho)
\end{align}
on all qubits except the last ancilla qubit.
\ENDFOR
\STATE We measure the last ancilla qubit using the $\{\ket{0}\bra{0},\ket{1}\bra{1}\}$ basis. Output the first case if the probability for detecting $\ket{1}\bra{1}$ $\leq e^{-3}$ and output the second case otherwise.
\end{algorithmic}
\end{algorithm}

We now analyze the mechanism of this protocol and prove \Cref{lem:LearnSingle}. We start with the distinguishing problem of \Cref{prob:SingleComp}. Similar to the proof in the previous section, we consider the performance in the two cases.

\paragraph{Case 1: }In this case, the unknown channel is a depolarization channel $\Lambdadep$. In each epoch for any $t$, the output state on the first $n+N$ qubits after $\Lambdadep\otimes I$ is
\begin{align*}
\Lambdadep\otimes I\left(\frac{I+P_t}{2^n}\otimes\ket{0}\bra{0}^{\otimes N}\right)=\frac{I}{2^n}\otimes\ket{0}\bra{0}^{\otimes N}.
\end{align*}
Thus, the states on each of the first $N$ ancilla qubits after we sequentially perform $\mathcal{E}_{\text{data},t}$ on the working qubits and the $j$-th ancilla qubit for $j=1,...,N$ will be $\frac{I}{2}$, which is exactly the single-qubit maximal mixed state. Therefore, after the special state purification channel {\sf StatePurification($N$)} defined in \Cref{algo:SpecialPurification}, we obtain 
\begin{align*}
\frac{1}{2^m}\ket{0}\bra{0}+\left(1-\frac{1}{2^m}\right)\ket{1}\bra{1}
\end{align*}
after the $m$-th epoch on the $(N+1)$-th qubit (or the $(n+N+1)$-th ancilla) according \eqref{eq:PurificationZero}. Thus the final state on the $(n+N+1)$-th qubit after $m=3n$ epochs is again
\begin{align*}
\frac{1}{8^n}\ket{0}\bra{0}+\left(1-\frac{1}{8^n}\right)\ket{1}\bra{1}.
\end{align*}
By induction, we can observe that the state on the last ancilla qubit (the $(n+k)$-th qubit) after the $t$-th iteration is
\begin{align}\label{eq:SingleCase1}
\left(1-\frac{1}{8^n}\right)^t\ket{0}\bra{0}+\left(1-\left(1-\frac{1}{8^n}\right)^t\right)\ket{1}\bra{1}.
\end{align}
After $4^n-1$ iterations, the coefficient of the support $\ket{1}\bra{1}$ is still exponentially small of $\sim 2^{-n}$.

\paragraph{Case 2: }Now, we consider the case when $\lambda_a>\epsilon$. In each epoch $t<a$, the process remains the same with the first case as 
\begin{align*}
\Lambda_a\otimes I\left(\frac{I+P_t}{2^n}\otimes\ket{0}\bra{0}^{\otimes N}\right)=\frac{I}{2^n}\otimes\ket{0}\bra{0}^{\otimes N},\ \forall t\neq a.
\end{align*}
By induction the same as the first case, we can observe that the state on the last ancilla qubit (the $(n+k)$-th qubit) after the $t$-th iteration is
\begin{align*}
\left(1-\frac{1}{8^n}\right)^t\ket{0}\bra{0}+\left(1-\left(1-\frac{1}{8^n}\right)^t\right)\ket{1}\bra{1}.
\end{align*}

At the $a$-th iteration, however, the output state on the first $n+N$ qubits after $\Lambda_a\otimes I$ is 
\begin{align*}
\Lambda_a\otimes I\left(\frac{I+P_a}{2^n}\otimes\ket{0}\bra{0}^{\otimes N}\right)=\frac{I+\lambda_aP_a}{2^n}\otimes\ket{0}\bra{0}^{\otimes N}.
\end{align*}
Thus, the states on each of the first $N$ ancilla qubits after we sequentially perform $\mathcal{E}_{\text{data},t}$ on the working qubits and the $j$-th ancilla qubit for $j=1,...,N$ will be
\begin{align*}
\frac{1+\lambda_a}{2}\ket{0}\bra{0}+\frac{1-\lambda_a}{2}\ket{1}\bra{1}=\lambda_a\ket{0}\bra{0}+(1-\lambda_a)\frac{I}{2}.
\end{align*}

Now, we consider the effect of step 7 of \Cref{algo:LearnSingle}. In this step, we perform the improved state purification {\sf StatePurification($N$)} on the first $N$ ancilla qubits and the $(N+1)$-th ancilla qubit. According to the lower bound in \Cref{eq:PurificationPosLower} and $N=\lceil2\log n/\epsilon^2\rceil$, we have
\begin{align*}
\Tr(E_\text{pure}^N\rho^{\otimes N})\geq1-\exp(-N\lambda^2/2)\geq1-\exp(-N\epsilon^2/2)\geq 1-\frac{1}{n}.
\end{align*}
At $m=1$, the $(N+1)$-th ancilla qubit is initialized in $\ket{0}\bra{0}$. Thus, the output state on the $(N+1)$-th ancilla qubit has at least $1-\frac{1}{n}$ coefficient for the support on $\ket{0}\bra{0}$. By induction, after $m$ epochs, the output state on the $(N+1)$-th ancilla qubit can be written as
\begin{align*}
\Tr(E_\text{pure}^N\rho^{\otimes N})^m\ket{0}\bra{0}+\left(1-\Tr(E_\text{pure}^N\rho^{\otimes N})^m\right)\ket{1}\bra{1}.
\end{align*}

At $m=3n$, without the loss of generality, we write the state on the $(N+1)$-th ancilla qubit as:
\begin{align*}
e^{-3}\ket{0}\bra{0}+(1-e^{-3})\rho'
\end{align*}
for some single-qubit state $\rho'$. Here, we use the fact that at large $n$, we have $(1-1/n)^{3n}\approx e^{-3}$. In the worst case when $\rho'=\ket{1}\bra{1}$, we will obtain the state
\begin{align*}
e^{-3}\ket{0}\bra{0}+(1-e^{-3})\ket{1}\bra{1}.
\end{align*}
after the $3n$-th epoch on the $(N+1)$-th ancilla qubit (or the $(n+N+1)$-th qubit). In any other case for $\rho'$, the coefficient of the support $\ket{0}\bra{0}$ will be even larger. After we perform $\mathcal{E}_\text{anc}$ on the last two ancilla qubits (i.e., the $(N+1)$- and the $(N+2)$-th ancilla qubits), we obtain
\begin{align*}
\left(1-\frac{1}{8^n}\right)^{a-1}\left(1-e^{-3}\right)\ket{0}\bra{0}+\left(1-\left(1-\frac{1}{8^n}\right)^{a-1}\left(1-e^{-3}\right)\right)\ket{1}\bra{1}
\end{align*}
on the last ancilla qubit.

The remaining iterations $t=a+1,...,4^n-1$ are similar with the iterations $t=1,...,a-1$. The $(N+2)$-th ancilla qubit is 
\begin{align}\label{eq:SingleCase2}
\left(1-\frac{1}{8^n}\right)^{4^n-2}\left(1-e^{-3}\right)\ket{0}\bra{0}+\left(1-\left(1-\frac{1}{8^n}\right)^{4^n-2}\left(1-e^{-3}\right)\right)\ket{1}\bra{1}
\end{align}
after the last iteration.

\paragraph{The proof for \Cref{lem:LearnSingle}: }Based on the analysis of \Cref{algo:LearnSingle} in the two cases, we can observe that when we perform the measurement on the last ancilla qubit in \eqref{eq:SingleCase1} and \eqref{eq:SingleCase2}, the probability of getting result $\ket{1}\bra{1}$ is exponentially small as $\sim2^{-n}$ for the first case and is 
\begin{align*}
\left(1-\left(1-\frac{1}{8^n}\right)^{4^n-2}\left(1-e^{-3}\right)\right)\sim e^{-3}
\end{align*}
for the second case. To distinguish between these two cases, we only need $O(1)$ measurements. This finishes the proof for the measurement complexity of the distinguishing problem in \Cref{lem:LearnSingle}. 

For the learning problem, we employ the same technique for proving \Cref{thm:Charsingle}. We use a $O(n)$ overhead on the measurement complexity to decide the $i=1,\ldots,n$ bit of the Pauli string $P_a$ using binary search. Thus, the measurement complexity for learning a single component Pauli channel is given by $O(n)$. The number of the unknown Pauli channel $\Lambda$ used is $3nN(4^n-1)=\tilde{O}(4^n/\epsilon^2)$, which finishes the proof for \Cref{lem:LearnSingle}.

\subsection{Threshold search for Pauli channels}\label{sec:ThresholdSearch}
In this section, we consider the Pauli channel threshold search problem, which is an analog of the quantum threshold search problem in state learning~\cite{aaronson2018shadow,aaronson2019gentle,buadescu2021improved}. We extend and formulate this problem as below:
\begin{problem}[Pauli channel threshold channel with parameters $(\epsilon,\epsilon')$]\label{prob:ThresholdSearch}
Suppose we are given 
\begin{itemize}
    \item Parameters $0\leq\epsilon'\leq\epsilon\leq1$.
    \item An unknown $n$-qubit Pauli channel $\Lambda(\cdot)=\frac{1}{2^n}[I\Tr(\cdot)+\sum_{a=1}^{4^n-1}\lambda_aP_a\Tr(P_a(\cdot))]$ that we can query.
    \item An $n$-qubit hypothesis Pauli channel $\hat{\Lambda}(\cdot)=\frac{1}{2^n}[I\Tr(\cdot)+\sum_{a=1}^{4^n-1}\hat{\lambda}_aP_a\Tr(P_a(\cdot))]$.
\end{itemize}
The goal is to output either
\begin{itemize}
    \item an $a^*\in\{1,...,4^n-1\}$ such that $\lambda_{a^*}>\hat{\lambda}_{a^*}+\epsilon'$, or
    \item for all $a$'s we have $\lambda_a\leq\hat{\lambda}_a+\epsilon$,
\end{itemize}
with a high success probability.
\end{problem}

In the case of state learning~\cite{buadescu2021improved}, the learner is given unentangled copies of unknown quantum threshold $\rho$, a list of $M$ two-outcome POVMs $E_1,...,E_M$, and a list of threshold $\theta_1,...,\theta_M\in[0,1]$. The goal of threshold search is to output either an $i^*$ such that $\Tr(E_{i^*}\rho)>\theta_{i^*}-\epsilon$, or $\Tr(E_i\rho)\leq\theta_i$ for all $i$. This problem was originally proposed in Ref.~\cite{aaronson2018shadow} named ``gentle search problem" and was shown to be solved using $\tilde{O}(\log^4(M)/\epsilon^2)$ copies of $\rho$. Later, this problem is shown to be solved using $\tilde{O}(\log^2(M)/\epsilon^2)$ samples~\cite{buadescu2021improved}. However, both these algorithms require joint and gentle\footnote{Gentle measurements are the measurements after which the post-measurement states are within bounded distance from the original states, see~\cite{aaronson2019gentle} for the formal definition.} measurements on sample batches of size $\poly\log(M)$. The size of quantum memory in this sense is at least $O(\poly\log(M)\cdot n)$, which is much larger than our expectation of logarithmic qubits in the setting of channel learning. To the best of our knowledge, there is no evidence showing that the quantum threshold search can be implemented sample-efficiently in the setting of state learning. In the process learning scenario, however, due to the channel concatenation, we are able to derive the following theorem showing that there exists a strategy that can solve \Cref{prob:ThresholdSearch} using $k=O(\log n/\epsilon^2)$ ancilla qubits and $O(n)$ measurements at $\epsilon'=\Theta(\epsilon/\sqrt{\log n})$.
\begin{lemma}\label{lem:ThresholdSearch}
There exists an algorithm that can solve \Cref{prob:ThresholdSearch} with parameters $(\epsilon',\epsilon)$ where
\begin{align*}
\epsilon'=\Theta\left(\frac{\epsilon}{\sqrt{\log n}}\right),
\end{align*}
with a $k=O(\log n/\epsilon^2)$ quantum memory and $O(n)$ measurements with a high probability.
\end{lemma}

\noindent At first glance, it is not straightforward to observe the connection between the Pauli channel threshold search in \Cref{prob:ThresholdSearch} and the learning single component Pauli channel problem in \Cref{prob:SingleComp}. To reveal the connection and introduce the algorithm for \Cref{prob:ThresholdSearch} in \Cref{lem:ThresholdSearch}, we first introduce an intermediate problem of Pauli channel hypothesis testing and show that the algorithm {\sf LearnSingle($n$,$\epsilon$)} in \Cref{algo:LearnSingle} can effectively solve it in the following part.

\subsubsection{Pauli channel hypothesis testing}
We now consider the following special case of the Pauli channel threshold search called the Pauli channel hypothesis testing problem as follows:
\begin{problem}[Pauli channel hypothesis testing with parameters $(\epsilon',\epsilon)$]\label{prob:PauliHypoTest}
Suppose we are given 
\begin{itemize}
    \item Parameters $0\leq\epsilon'\leq\epsilon\leq1$.
    \item An unknown $n$-qubit Pauli channel $\Lambda(\cdot)=\frac{1}{2^n}[I\Tr(\cdot)+\sum_{a=1}^{4^n-1}\lambda_aP_a\Tr(P_a(\cdot))]$ that we can query.
\end{itemize}
The goal is to output either
\begin{itemize}
    \item an $a^*\in\{1,...,4^n-1\}$ such that $\lambda_{a^*}>\epsilon'$, or
    \item for all $a$'s we have $\lambda_a\leq\epsilon$,
\end{itemize}
with a high success probability.
\end{problem}

For the above problem, we now prove the following lemma regarding the performance of the algorithm {\sf LearnSingle($n$,$\epsilon$)} in \Cref{algo:LearnSingle}.
\begin{lemma}\label{lem:PauliHypoTest}
{\sf LearnSingle($n$,$\epsilon$)} in \Cref{algo:LearnSingle} can solve \Cref{prob:PauliHypoTest} with parameters $(\epsilon',\epsilon)$ where
\begin{align*}
\epsilon'=\Theta\left(\frac{\epsilon}{\sqrt{\log n}}\right),
\end{align*}
with a $k=O\left(\log n/\epsilon^2\right)$ ancilla qubits and $O(n)$ measurement with a high probability.
\end{lemma}
\begin{proof}
We consider the algorithm {\sf LearnSingle($n$,$\epsilon$)} in \Cref{algo:LearnSingle} when we input the unknown $n$-qubit Pauli channel 
\begin{align*}
\Lambda(\cdot)=\frac{1}{2^n}\left[I\Tr(\cdot)+\sum_{a=1}^{4^n-1}\lambda_aP_a\Tr(P_a(\cdot))\right].
\end{align*}

In the iteration $t=1,...,4^n-1$, we assume that the input state is 
\begin{align*}
\frac{I+P_t}{2^n}\otimes\ket{0}\bra{0}^{\otimes (N+1)}\otimes\left((1-a)\ket{0}\bra{0}+a\ket{1}\bra{1}\right)
\end{align*}
for some value $a\in[0,1]$. In each epoch, the output state after we query $\Lambda\otimes I$ on the first $n+N$ qubits is
\begin{align*}
\Lambda\otimes I\left(\frac{I+P_t}{2^n}\otimes\ket{0}\bra{0}^{\otimes N}\right)=\frac{I+\lambda_t P_t}{2^n}\otimes\ket{0}\bra{0}^{\otimes N}.
\end{align*}
Therefore, the state on the working qubits and the $j$-th ancilla qubit after implementing $\mathcal{E}_{\text{data},t}$ at $j=1,...,N$ is
\begin{align*}
\mathcal{E}_{\text{data},t}\left(\frac{I+\lambda_t P_t}{2^n}\otimes\ket{0}\bra{0}\right)=\frac{I}{2^n}\otimes\left(\frac{1+\lambda_t}{2}\ket{0}\bra{0}+\frac{1-\lambda_t}{2}\ket{1}\bra{1}\right).
\end{align*}

We then consider the purification channel. Notice that different from \Cref{prob:SingleComp}, we allow $\lambda_t<0$ here in \Cref{prob:PauliHypoTest} and \Cref{prob:ThresholdSearch}. However, we can prove that, after $m$ epochs, the output state after we implement {\sf StatePurification($N$)} on the $(N+1)$-th ancilla qubit represented as
\begin{align*}
\Tr(E_\text{pure}^N\rho^{\otimes N})^m\ket{0}\bra{0}+\left(1-\Tr(E_\text{pure}^N\rho^{\otimes N})^m\right)\ket{1}\bra{1}.
\end{align*}
After $3n$ epochs, the state of the $(N+1)$-th ancilla qubit represented as
\begin{align*}
\Tr(E_\text{pure}^N\rho^{\otimes N})^{3n}\ket{0}\bra{0}+\left(1-\Tr(E_\text{pure}^N\rho^{\otimes N})^{3n}\right)\ket{1}\bra{1}.
\end{align*}

For simplicity, we use $f_t$ to denote $\Tr(E_\text{pure}^N\rho^{\otimes N})$. If $\lambda_t>0$, we have $f_t>1/2$, and the upper and the lower bound for $f_t$ are given in \eqref{eq:PurificationPosUpper} and \eqref{eq:PurificationPosLower}. The state on the $(N+1)$-th ancilla qubit is
\begin{align*}
f_t^{3n}\ket{0}\bra{0}+\left(1-f_t^{3n}\right)\ket{1}\bra{1}
\end{align*}
after $m=3n$ epochs.

If $\lambda_t=0$, this state is exactly
\begin{align*}
\frac{1}{8^n}\ket{0}\bra{0}+\left(1-\frac{1}{8^n}\right)\ket{1}\bra{1}
\end{align*}
according to \eqref{eq:PurificationZero} as $f_t=1/2$.

In the case when $\lambda_t<0$, the state of the $(N+1)$-th ancilla qubit is
\begin{align*}
f_t^{3n}\ket{0}\bra{0}+\left(1-f_t^{3n}\right)\ket{1}\bra{1}
\end{align*}
for some $f_t<1/2$ according to \eqref{eq:PurificationNeg}.

After we implement $\mathcal{E}_{\text{anc}}$ on the last two ancilla qubits and the preparation channel $\mathcal{E}_{\text{prep},t}$, we enter the next iteration with
\begin{align*}
\frac{I+P_{t+1}}{2^n}\otimes\ket{0}\bra{0}^{\otimes N+1}\otimes\left((1-a)\left(1-f_t^{3n}\right)\ket{0}\bra{0}+\left(1-(1-a)\left(1-f_t^{3n}\right)\right)\ket{1}\bra{1}\right).
\end{align*}

For $\lambda_t=0$, we enter the next iteration with
\begin{align*}
\frac{I+P_{t+1}}{2^n}\otimes\ket{0}\bra{0}^{\otimes N+1}\otimes\left((1-a)\left(1-\frac{1}{8^n}\right)\ket{0}\bra{0}+\left(1-(1-a)\left(1-\frac{1}{8^n}\right)\right)\ket{1}\bra{1}\right).
\end{align*}

By induction on all $t$ and the above analysis for the iteration $t$, we can observe that the final state after all $4^n-1$ iteration on the last ancilla qubit (i.e., $(n+N+2)$-th qubit) is
\begin{align*}
\rho_\text{final}=\Lambda_0\ket{0}\bra{0}+(1-\Lambda_0)\ket{1}\bra{1},
\end{align*}
where
\begin{align*}
\Lambda_0=\prod_{t}\left(1-f_t^{3n}\right)=\prod_{t:\lambda_t>0}\left(1-f_t^{3n}\right)\prod_{t:\lambda_t=0}\left(1-\frac{1}{8^n}\right)\prod_{t:\lambda_t<0}\left(1-f_t^{3n}\right).
\end{align*}

We consider measuring it in the basis $\{\ket{0}\bra{0},\ket{1}\bra{1}\}$. 

Firstly, suppose that there exists some $\lambda_{t^*}>\epsilon$. Then we have
\begin{align}
f_{t^*}>1-\frac 1n
\end{align}
according to \eqref{eq:PurificationPosLower}. Therefore, we have
\begin{align*}
\Lambda_0=\prod_{t}\left(1-f_t^{3n}\right)< (1-f_{t^*}^{3n})=(1-e^{-3}),
\end{align*}
and the probability for detecting $\ket{1}\bra{1}$ is
\begin{align}
\Tr(\ket{1}\bra{1}\rho_{\text{final}})=1-\Lambda_0>e^{-3}.
\end{align}
Therefore, when the probability of obtaining the result $\ket{1}\bra{1}$ is bounded above by $e^{-3}$, we conclude that $\lambda_t\leq\epsilon$ for all $t=1,...,4^n-1$, which is the second case.

We then consider the case when the probability of obtaining the result $\ket{1}\bra{1}$ is larger than $e^{-3}$. In this case, we conclude that there exists some $t^*$ such that $\lambda_{t^*}>\epsilon'=\Theta(\epsilon/\sqrt{\log n})$. To prove this, we consider the adversarial case when some adversary can construct the unknown Pauli channel $\Lambda$. The optimal strategy for the adversary is to construct identical $\lambda_t>0$ for all $t$. We assume that $\lambda_t\leq\epsilon'$, in order that we detect inversely $\Tr(\ket{1}\bra{1}\rho_\text{final})\leq e^{-3}$, we have
\begin{align*}
\Tr(\ket{1}\bra{1}\rho_{\text{final}})&=1-\Lambda_0\\
&=1-\prod_{t=1}^{4^n-1}\left(1-f_t^{3n}\right).\\
&=1-\left(1-f_t^{3n}\right)^{4^n-1}.
\end{align*}
Using the upper bound for $f_t$ given by \eqref{eq:PurificationPosUpper}, we can always find some $\epsilon'=\Theta(1/\sqrt{N})=\Theta(\epsilon/\sqrt{\log n})$ such that
\begin{align*}
\Tr(\ket{1}\bra{1}\rho_{\text{final}})&=1-\Lambda_0\\
&\leq 1-\left(1-\left(\frac{1}{2}\left(1+\sqrt{1-\exp\left(-\frac{N\epsilon^{'2}}{1-\epsilon^{'2}}\right)}\right)\right)\right)^{4^n-1}\\
&\leq e^{-3}
\end{align*}

Therefore, if we detect the probability of obtaining the result $\ket{1}\bra{1}$ is larger than $e^{-3}$, we output the first case of \Cref{prob:PauliHypoTest}. To identify the exact value of $t^*$, we perform a binary search similar to \Cref{thm:Charsingle} and pay an additional $O(n)$ overhead in the measurement complexity. This finishes the proof for \Cref{lem:PauliHypoTest}.
\end{proof}

We remark that the analysis above for \Cref{algo:LearnSingle} in solving the Pauli channel hypothesis testing problem also indicates that we can use \Cref{algo:LearnSingle} to learn all the eigenvalues of a Pauli channel using a polynomial number of measurements if we assume that it has a polynomial number of nonzero $\lambda_a$'s. This is because we can perform the binary search procedure similar to \Cref{thm:Charsingle} to find all the nonzero eigenvalues. As there are at most a polynomial number of nonzero eigenvalues, the number of branches during any time of the search is polynomial. Yet, this result is strictly weaker than \Cref{thm:ChanEstCon} as we show that we can solve the Pauli channel estimation task for \textbf{any} Pauli channel using polynomially many measurements.

\subsubsection{Proof for \Cref{lem:ThresholdSearch}}
We are now ready to prove \Cref{lem:ThresholdSearch}\footnote{As pointed out by the following up work from Wang et al.~\cite{wang2025weakly}, there is an issue in the proof below for \Cref{lem:ThresholdSearch} as the data processing channel in Eq.~\eqref{eq:ELambdaP} and Eq.~\eqref{eq:ELambdaM} are not completely positive. One solution to fix this issue can be referred to Appendix B of Ref.~\cite{wang2025weakly}}.

Given that {\sf LearnSingle($n$,$\epsilon$)} in \Cref{algo:LearnSingle} can solve \Cref{prob:PauliHypoTest} using $k=O(\log n/\epsilon^2)$ ancilla qubits and $O(n)$ measurements, we now consider how to generalize this result to the Pauli channel threshold search in \Cref{prob:ThresholdSearch}. In particular, we are to develop a method to reduce \Cref{prob:ThresholdSearch} to \Cref{prob:PauliHypoTest} by changing the hypothesis channel 
\begin{align*}
\hat{\Lambda}(\cdot)=\frac{1}{2^n}[I\Tr(\cdot)+\sum_{a=1}^{4^n-1}\hat{\lambda}_aP_a\Tr(P_a(\cdot))]
\end{align*}
to the depolarization channel. To achieve this goal, we consider the following procedure for a sub-routine. Suppose we input the state $\frac{I+P_t}{2^n}$ into the unknown channel $\Lambda$. Our goal is to output either $\lambda_t>\hat{\lambda}_t+\epsilon'$ or $\lambda_t\leq\hat{\lambda}_t+\epsilon$ in \Cref{prob:ThresholdSearch}. Notice that the output state is
\begin{align*}
\Lambda\left(\frac{I+P_t}{2^n}\right)=\frac{I+\lambda_tP_t}{2^n}.
\end{align*}
Suppose $\hat{\lambda_t}\geq0$, we consider concatenating the following channel after the unknown channel:
\begin{align}\label{eq:ELambdaP}
\mathcal{E}_{\hat{\lambda}_t,+}:\rho=\frac{I+\sum_{a=1}^{4^n-1}\sigma_aP_a}{2^n}\to\frac{1}{2^n}\left[I+\sum_{a\neq t}\sigma_aP_a+\left(\frac{\sigma_t}{1+\hat{\lambda}_t}-\frac{\hat{\lambda}_t}{1+\hat{\lambda}_t}\right)P_t\right],
\end{align}
for $\sigma_a\in[-1,1]$ and $\rho$ any $n$-qubit density matrix. This channel always exists as we can verify that it is a linear completely-positive trace-preserving map. We then consider the state
\begin{align*}
\mathcal{E}_{\hat{\lambda}_t,+}\circ\Lambda\left(\frac{I+P_t}{2^n}\right)=\frac{I+\lambda_t'P_t}{2^n}.
\end{align*}
If $\lambda_t=\hat{\lambda}_t$, we have $\lambda_t'=0$. If $\lambda_t>\hat{\lambda}_t+\epsilon'$, we have $\lambda_t'>\epsilon'/(1+\hat{\lambda}_t)\geq\epsilon'/2$. On the contrary, if $\lambda_t\leq\hat{\lambda}_t+\epsilon$, we have $\lambda_t'\leq\epsilon/(1+\hat{\lambda}_t)\leq\epsilon$. Therefore, we have transfer a Pauli threshold search task in \Cref{prob:ThresholdSearch} with parameters $(\epsilon,\epsilon')$ to a Pauli channel hypothesis testing in \Cref{prob:PauliHypoTest} with parameters $(\epsilon,\epsilon'/2)$ at $\lambda_t>0$.

Suppose $\hat{\lambda}_t<0$, we consider alternatively concatenating the following channel after the unknown channel:
\begin{align}\label{eq:ELambdaM}
\mathcal{E}_{\hat{\lambda}_t,-}:\rho=\frac{I+\sum_{a=1}^{4^n-1}\sigma_aP_a}{2^n}\to\frac{1}{2^n}\left[I+\sum_{a\neq t}\sigma_aP_a+\left(\frac{\sigma_t}{1-\hat{\lambda}_t}-\frac{\hat{\lambda}_t}{1-\hat{\lambda}_t}\right)P_t\right],
\end{align}
for $\sigma_a\in[-1,1]$ and $\rho$ any $n$-qubit density matrix. We can also prove that this channel always exists as we can verify that it is a linear completely-positive trace-preserving map. We then consider the state
\begin{align*}
\mathcal{E}_{\hat{\lambda}_t,-}\circ\Lambda\left(\frac{I+P_t}{2^n}\right)=\frac{I+\lambda_t'P_t}{2^n}.
\end{align*}
If $\lambda_t=\hat{\lambda}_t$, we have $\lambda_t'=0$. If $\lambda_t>\hat{\lambda}_t+\epsilon'$, we have $\lambda_t'>\epsilon'/(1-\hat{\lambda}_t)\geq\epsilon'/2$. On the contrary, if $\lambda_t\leq\hat{\lambda}_t+\epsilon$, we have $\lambda_t'\leq\epsilon/(1-\hat{\lambda}_t)\leq\epsilon$. Therefore, we have transfer a Pauli threshold search task in \Cref{prob:ThresholdSearch} with parameters $(\epsilon,\epsilon')$ to a Pauli channel hypothesis testing in \Cref{prob:PauliHypoTest} with parameters $(\epsilon,\epsilon'/2)$ for any $\lambda_t$. 

Based on the above construction, we propose the algorithm for solving the Pauli channel threshold search in \Cref{algo:ThresholdSearch}, which immediately proves \Cref{lem:ThresholdSearch} following the same procedure of \Cref{lem:PauliHypoTest}.

\begin{algorithm}[htbp]
\caption{The protocol for distinguishing (learning) single component Pauli channels using $3$ ancilla qubits ({\sf ThresholdSearch($n$,$\epsilon$,$\{\hat{\lambda}_a\}_{a=1}^{4^n-1}$)})}
\label{algo:ThresholdSearch}
\begin{algorithmic}[1]
\REQUIRE An unknown $n$-qubit Pauli channel $\Lambda$, a hypothesis channel $\hat{\Lambda}$ with parameters $\{\hat{\lambda}_a\}_{a=1}^{4^n-1}$, and accuracy demand $\epsilon$.
\ENSURE Distinguish whether the unknown channel is the depolarization channel $\Lambdadep$ or is sampled uniformly from $\{\Lambda_a\}_{a\in\Z_2^{2n}\backslash \{0\}}$ with $\Lambda_a=\frac{1}{2^n}[I\Tr(\cdot)+ \lambda_a P_a\Tr(P_a(\cdot))]$ for $\lambda_a\in(\epsilon,1]$.
\STATE We initialize the working qubits with $\frac{I+P_1}{2^n}$ and the $k$ ancilla qubits with $\ket{0}\bra{0}^{\otimes k}$. We set $k=\lceil 2\log n/\epsilon^2\rceil+2$ and $N=k-2=\lceil 2\log n/\epsilon^2\rceil$.
\FOR{Iteration $t=1,...,4^n-1$}
\STATE We start with the initial state $\frac{I+P_t}{2^n}\otimes \ket{0}\bra{0}^{\otimes N+1}$ in the first $N+1$ qubits 
\FOR{Epoch $k=1,...,3n$}
\STATE We input the current state into $\mathcal{E}_{\hat{\lambda}_t,\pm}\circ\Lambda_n\otimes I$ according to $\hat{\lambda}_t$ (defined in \eqref{eq:ELambdaP} and \eqref{eq:ELambdaM}), where $I(\cdot)$ acts on all ancilla qubits.
\STATE We apply a sequences of $(n+1)$-qubit quantum channels $\mathcal{E}_{\text{data},t}$ (defined in \eqref{chan:TDataRec}) acting on the working qubits and the $j$-th ancilla qubit for $j=1,...,N$. After each $\mathcal{E}_{\text{data},t}$ at $j=1,...,N-1$, we re-initialize the working qubits using the working qubit initialization channel $\mathcal{E}_{\text{work ini},t}$ defined in \eqref{chan:WorkInitial} and input the current state into $\mathcal{E}_{\hat{\lambda}_t,\pm}\circ\Lambda_n\otimes I$ (i.e., query the unknown channel $\Lambda$ again).
\STATE Next, we perform {\sf StatePurification($N$)} (\Cref{algo:SpecialPurification}) on the first $N$ ancilla qubits and the $(N+1)$-th ancilla qubit. The $(N+1)$-th ancilla qubit is chosen as the ancilla qubit in {\sf StatePurification($N$)}.
\STATE We re-initialize the first $n+N$ qubits into $\frac{I+P_t}{2^n}\otimes \ket{0}\bra{0}^{\otimes N}$ using the epoch initialization channel $\mathcal{E}_{\text{epo ini},t}$ defined in \eqref{chan:EpoT}.
\ENDFOR
\STATE We perform the channel $\mathcal{E}_\text{anc}$ (defined in \eqref{chan:AncRec}) on the last two ancilla qubits.
\STATE We perform a channel $\mathcal{E}_{\text{prep},t}$ (defined in \eqref{chan:PrepT}) on all qubits except the last ancilla qubit preparing the input state for the next iteration $t+1$.
\ENDFOR
\STATE We measure the last ancilla qubit using the $\{\ket{0}\bra{0},\ket{1}\bra{1}\}$ basis. Output the first case if the probability for detecting $\ket{1}\bra{1}$ $\leq e^{-3}$ and output the second case otherwise.
\end{algorithmic}
\end{algorithm}

\subsection{Online learning Pauli channels}\label{sec:OnlinePauli}
We now consider learning Pauli channels in the online learning setting. The basic concepts of online learning, mistake bounds, and regrets are introduced in \Cref{sec:OnlineLearning}. As defined in \Cref{prob:def}, given a Pauli channel $\Lambda(\cdot)=\frac{I\Tr(\cdot)+\sum_{a}\lambda_aP_a\Tr(P_a(\cdot))}{2^n}$, the goal of the learner is to learn all the eigenvalues $\lambda_a$ within accuracy demand $\epsilon$. 

In iteration $t$ of the online learning scheme, the learner is given an index $a_t\in\{1,...,4^n-1\}$ and is required to output a hypothesis Pauli channel (prediction)
\begin{align*}
\hat{\Lambda}_t=\frac{I\Tr(\cdot)+\sum_{a=1}^{4^n-1}\hat{\lambda}_{a,t}P_a\Tr(P_a(\cdot))}{2^n}.    
\end{align*} 
The learner then obtains feedback $b_t\in[0,1]$ with $\abs{b_t-\lambda_{a_t}}\leq\epsilon_2$ for some $0\leq\epsilon_2\leq\epsilon\leq 1$. The learner then suffers from a $L_1$ loss equal to
\begin{align*}
\ell_t(\hat{\lambda}_{a_t,t})\coloneqq\abs{b_t-\hat{\lambda}_{a_t,t}}.
\end{align*}
It is straightforward to observe that this $L_1$ loss function is $L$ Lipschitz. The regret for the first $T$ iterations is computed by:
\begin{align*}
R_T=\sum_{t=1}^T\ell_t(\hat{\lambda}_{a_t,t})-\min_{\hat{\Lambda}}\sum_{t=1}^T\ell_t(\hat{\lambda}_{a_t})
\end{align*}
for some optimal $\hat{\Lambda}=\frac{I\Tr(\cdot)+\sum_{a=1}^{4^n-1}\hat{\lambda}_{a}P_a\Tr(P_a(\cdot))}{2^n}$.

After receiving the feedback, the learner computes the loss function $\ell_t(\hat{\Lambda}_t)$. The learner then performs an update procedure on the hypothesis channel $\hat{\Lambda}_t\to\hat{\Lambda}_{t+1}$ if the loss function $\ell_t(\hat{\Lambda}_t)$ is larger than a certain threshold. As the learner, we use the techniques from online convex optimization to minimize regret. The number of ``bad iterations" $t$ in which the hypothesis channel is tested to be far away from the true channel can also be bounded via minimizing the regret. For state learning tasks, Ref.~\cite{aaronson2018online} proposed several online learning update methods in the setting of state learning including the Regularized Follow-the-Leader algorithm (RFTL)~\cite{tsuda2005matrix,hazan2016introduction}, the Matrix Multiplicative Weights algorithm~\cite{tsuda2005matrix,arora2012multiplicative}, and the sequential fat-shattering dimension of
quantum states~\cite{rakhlin2015online}. Some offline algorithms~~\cite{aaronson2018shadow,aaronson2007learnability} are also proposed for a similar update procedure. Here, we follow the template of the RFTL using the regularization function of (negative) von Neumann entropy for the Choi–Jamio{\l}kowski isomorphism $\rho_\text{CJ}(\Lambda)$ state of the Pauli channel $\Lambda$ in the extended Hilbert space defined in \eqref{eq:Choistate}. That is to say. given a Pauli channel $\Lambda$, the regularized function is
\begin{align}\label{eq:RFTLRFunc}
R(\Lambda)=\sum_{i=1}^{4^n}\lambda_i\left(\rho_\text{CJ}(\Lambda)\right)\log\lambda_i\left(\rho_\text{CJ}(\Lambda)\right),
\end{align}
where $\lambda_i$ denotes the $i$-th eigenvalue of the matrix.

The full algorithm is provided below in \Cref{algo:RFTLPauli}:
\begin{algorithm}[htbp] 
\caption{RFTL for online Pauli channel estimation}
\begin{algorithmic}[1] \label{algo:RFTLPauli}
\STATE Iteration number $T$, some tuned value $\eta<\frac{1}{2}$, unknown Pauli channel $\Lambda$ with feedbacks $b_t$'s in each iteration, and $\mathcal{K}$ the set of all Pauli channels.
\STATE Set the initial prediction as $\hat{\Lambda}_1(\cdot)\coloneqq\frac{I\Tr(\cdot)}{2^n}$.
\FOR{$t=1,...,T$}
\STATE Predict the hypothesis Pauli channel $\hat{\Lambda}_t$ with Choi–Jamio{\l}kowski isomorphism $\rho_\text{CJ}(\hat{\Lambda}_t)$ in \eqref{eq:Choistate}. Consider the $L_1$ loss
function $\ell_t$ defined in \Cref{sec:OnlinePauli} and $\ell'(x)$ be the derivative of $\ell_t$ with respect to $x$. Define
\begin{align}\label{eq:RFTLLoss}
\nabla_t\hat{\Lambda}_t\coloneqq\rho_\text{CJ}\left(\ell'(\hat{\lambda}_{a_t,t})\cdot\frac{P_{a_t}\Tr(P_{a_t}(\cdot))}{2^n}\right).
\end{align}

\STATE Update decision according to the RFTL rule with regularized function in \eqref{eq:RFTLRFunc} by $\hat{\Lambda}_{t+1}\coloneqq$:
\begin{align*}
\arg\min_{\varphi\in\mathcal{K}}\left\{\eta\sum_{s=1}^t\Tr(\nabla_s\varphi)+R(\varphi) \right\}.
\end{align*}
\ENDFOR
\end{algorithmic}
\end{algorithm}

Before analyzing the performance, we first consider the correctness of this RFTL algorithm. We can observe that the mathematical program is convex. The basic idea for this online Pauli channel estimation is to employ the Choi–Jamio{\l}kowski isomorphism to map the Pauli channel into a $2n$-qubit state. According to Ref.~\cite{aaronson2018online}, the minimizer in the last step is always positive semidefinite in the state learning setting, which also guarantees the correctness of our algorithm.

Now, we consider the regret bound of \Cref{algo:RFTLPauli}. We suppose there are in total $T$ iterations where the learner performs an update procedure. We have the following lemma:
\begin{lemma}\label{lem:OnlineRegret}
By setting $\eta=O(\sqrt{n/T})$ for some suitable constant coefficient, we can bound the regret $R_T$ of \Cref{algo:RFTLPauli} by
\begin{align*}
R_T\leq O(\sqrt{Tn}).
\end{align*}
\end{lemma}
\begin{proof}
By theorem 5.2 of Ref.~\cite{hazan2016introduction}, we can bound the regret of \Cref{algo:RFTLPauli} by
\begin{align*}
R_T\leq O(D_RG_R\sqrt{T})
\end{align*}
by choosing some $\eta=O(D_R/(G_R\sqrt{T}))$, where
\begin{align*}
D_R&=\sqrt{\max_{\phi_1,\phi_2\in\mathcal{K}}\abs{R(\phi_1)-R(\phi_2)}},\\
G_R&=\norm{\nabla_t}_{\text{op}}.
\end{align*}
Notice that $R(\cdot)$ is the von Neumann entropy for a $2n$-qubit quantum state, it is bounded by $2n\log 2=O(n)$. As the loss function $\ell_t(\cdot)$ is $L=1$-Lipschitz, we can bound $\norm{\nabla_t}_{\text{op}}\leq O(1)$ according to its definition in \eqref{eq:RFTLLoss}. Therefore, we have 
\begin{align*}
R_T\leq O(\sqrt{Tn})
\end{align*}
when we choose some $\eta=O(\sqrt{n/T})$. This finishes the proof for \Cref{lem:OnlineRegret}.
\end{proof}

Finally, we bound the number of iterations where we invode the update algorithm. We consider running the update algorithm in \Cref{algo:RFTLPauli} when the loss function $\ell_t(\hat{\lambda}_{a_t,t})=\abs{b_t-\hat{\lambda}_{a_t,t}}>\epsilon_1-\epsilon_2$ for some $0\leq\epsilon_2\leq\epsilon_1\leq\epsilon\leq 1$. By the definition of $b_t$, $\abs{b_t-\lambda_{a_t,t}}\leq\epsilon_2$. Given the regret bound in \Cref{lem:OnlineRegret}, we consider the true Pauli channel $\Lambda$ as the hypothesis. As $\abs{b_t-\lambda_{a_t,t}}\leq\epsilon_2$, the sum of loss function over all iterations $t=1,...,T$ in this case is bounded above by $\epsilon_2T$. Then, given an arbitrary online learning procedure with $T$ iterations, the regret is bounded below by $\Omega((\epsilon_1-\epsilon_2)T)$. According to the regret bound, we have
\begin{align*}
(\epsilon_1-\epsilon_2)T\leq\epsilon_2 T+O(\sqrt{Tn}),
\end{align*}
which directly indicates that
\begin{align*}
T\leq O\left(\frac{n}{(\epsilon_1-2\epsilon_2)^2}\right).
\end{align*}
Thus, we formally have the below lemma:
\begin{lemma}\label{lem:OnlineIter}
Suppose the learner is given feedback $b_t$ such that $\abs{b_t-\lambda_{a_t,t}}\leq\epsilon_2$ in each iteration and invoke the RFTL update procedure if $\ell_t(\hat{\lambda}_{a_t,t})=\abs{b_t-\hat{\lambda}_{a_t,t}}>\epsilon_1-\epsilon_2$. We choose the parameters
\begin{align*}
\epsilon_1=\Theta(\epsilon/\sqrt{\log n}),\quad\epsilon_1\lesssim\epsilon',\quad\epsilon_2=\epsilon_1/2.
\end{align*}
Then the number of iterations $T$ where the update the procedure is used is bounded above by
\begin{align*}
T\leq \tilde{O}\left(\frac{n}{\epsilon^2}\right).
\end{align*}
\end{lemma}

\subsection{Online Pauli channel estimation with logarithmic ancillas and channel concatenation}\label{sec:Wrap}
We now wrap up all the lemmas and provide the proof for \Cref{thm:ChanEstCon}. In particular, we consider an online learning procedure using the RFTL update procedure in \Cref{algo:RFTLPauli} and the Pauli channel threshold search in \Cref{algo:ThresholdSearch}. Assume that the online learning procedure lasts for $T$ iterations. We give the detailed procedure as follows:

At the beginning, we initialize our hypothesis Pauli channel with 
\begin{align*}
\hat{\Lambda}_1(\cdot)=\frac{I\Tr(\cdot)}{2^n}.
\end{align*}
In the $t=1,...,T$-th iteration, we assume that the prediction by the learner is
\begin{align*}
\hat{\Lambda}_t=\frac{I\Tr(\cdot)+\sum_{a=1}^{4^n-1}\hat{\lambda}_{a,t}P_a\Tr(P_a(\cdot))}{2^n}.    
\end{align*} 

We first implement a two-side version of the Pauli channel threshold search in \Cref{prob:ThresholdSearch} with parameter $(\epsilon,\epsilon')$ at some $\epsilon'\gtrsim O(\epsilon/\sqrt{\log n})$. In the Pauli channel threshold search, we are to output either:
\begin{itemize}
    \item an $a^*\in\{1,...,4^n-1\}$ such that $\lambda_{a^*}>\hat{\lambda}_{a^*}+\epsilon'$, or
    \item for all $a$'s we have $\lambda_a\leq\hat{\lambda}_a+\epsilon$.
\end{itemize}
Now, we consider the inverse side, our goal is to output either
\begin{itemize}
    \item an $a^*\in\{1,...,4^n-1\}$ such that $\lambda_{a^*}<\hat{\lambda}_{a^*}-\epsilon'$, or
    \item for all $a$'s we have $\lambda_a\geq\hat{\lambda}_a-\epsilon$.
\end{itemize}
By a similar reduction from the Pauli channel threshold search problem in \Cref{prob:ThresholdSearch} to the Pauli channel hypothesis testing in \Cref{prob:PauliHypoTest}, we only need to solve an inverse side version of the hypothesis testing. In particular, we are to output either
\begin{itemize}
    \item an $a^*\in\{1,...,4^n-1\}$ such that $\lambda_{a^*}<-\epsilon'$, or
    \item for all $a$'s we have $\lambda_a\geq-\epsilon$.
\end{itemize}
As proved by \Cref{lem:PauliHypoTest}, we can solve the original version of Pauli channel hypothesis testing using {\sf LearnSingle($n$,$\epsilon$)} in \Cref{algo:LearnSingle}. To solve this inverse side version of Pauli channel hypothesis testing, we only need to exchange the sequence of the two elements in two-outcome POVM given by $\{E_{\text{data},a},I-E_{\text{data},a}\}=\left\{\sum_{j=1}^{2^{n-1}}\ket{\varphi_a^{j,+}}\bra{\varphi_a^{j,+}},\sum_{j=2^{n-1}+1}^{2^n}\ket{\varphi_a^{j,-}}\bra{\varphi_a^{j,-}}\right\}$ in \Cref{lem:LearnSingle}. Therefore, according to \Cref{lem:ThresholdSearch}, we only need 
\begin{align*}
T_1=O(n)
\end{align*}
measurements and $k=O(\log n/\epsilon^2)$ ancilla qubits to output either:
\begin{itemize}
    \item an $a^*\in\{1,...,4^n-1\}$ such that $\abs{\lambda_{a^*}-\hat{\lambda}_{a^*}}>\epsilon'$, or
    \item for all $a$'s we have $\abs{\lambda_a-\hat{\lambda}_a}\leq\epsilon$.
\end{itemize}

If the two-side version threshold search outputs the second case, we can conclude that the hypothesis channel $\hat{\Lambda}_t$ is $\epsilon$-close to the actual unknown Pauli channel $\Lambda$ in the $\ell_\infty$ norm over all Pauli eigenvalues. However, if the two-side version threshold search outputs the first case for some $a^*$, we can fit $a_t=a^*$ and invoke the update procedure. To compute the feedback $b_t$, we input the state $\frac{I+P_{a_t}}{2^n}$ into the unknown channel for $O(1/\epsilon_2^2)=\tilde{O}(n/\epsilon^2)$ to estimate $\lambda_{a_t}$ with accuracy $\epsilon_2$ with a high probability. According to \Cref{lem:OnlineIter}, there can be at most
\begin{align*}
T_2=\tilde{O}\left(\frac{n}{\epsilon^2}\right)
\end{align*}
such iterations. Therefore, the total measurement complexity is bounded by:
\begin{align*}
T_{\text{measurement}}=T_1T_2=\tilde{O}\left(\frac{n^2}{\epsilon^2}\right),
\end{align*}
which finishes the proof for \Cref{thm:ChanEstCon}.

\section{Lower bound on measurement complexity without quantum memory}\label{sec:MeasComp}

In this section, we prove Theorem~\ref{thm:AdaptConLow_informal}, which gives an optimal lower bound on the number of measurements needed by any learning protocol without quantum memory: 

\begin{theorem}\label{thm:AdaptConLow}
The following holds for any $0<\epsilon\leq 1$. For any (possibly adaptive and concatenating) protocol without quantum memory that solves Pauli eigenvalue estimation (Problem~\ref{prob:def}) and even easier Pauli spike detection (Problem~\ref{prob:SpikeDetect}) to error $\epsilon$ with at least large constant probability, the number of measurements the protocol makes must be at least $\Omega(2^n/\epsilon^2)$.
\end{theorem}

\noindent As discussed in the introduction, this lower bound is tight in both the $n$ and $\epsilon$ dependence, because of the algorithmic upper bound in~\cite{flammia2020efficient}. Furthermore, our lower bound applies for the full range of possible values of $\epsilon \in (0,1]$.

For the proof of Theorem~\ref{thm:AdaptConLow}, we follow the formalism of the many-versus-one distinguishing problem introduced in Section~\ref{sec:tree} and consider distinguishing between the following two scenarios:
\begin{itemize}
    \item The unknown channel is the depolarization channel $\Lambdadep=\frac{1}{2^n}I\Tr(\cdot)$.
    \item The unknown channel is sampled uniformly from a set of channels $\{\Lambda_a\}_{a=1}^{4^n-1}$ with $\Lambda_a=\frac{1}{2^n}[I\Tr(\cdot)+\epsilon P_a\Tr(P_a(\cdot))]$.
\end{itemize}
Here, $P_1,\ldots,P_a,\ldots,P_{4^n-1}$ are Pauli traceless observables. Recall that in the Pauli channel eigenvalue estimation problem, our goal is to learn each eigenvalue $\lambda_a$ within $\epsilon$ in \Cref{prob:def}. Therefore, such a protocol can solve the above distinguishing problem. 

We now consider the tree representation $\cT$ of the learning algorithm without memory in Definition~\ref{def:TreeMemless}. In the first scenario where the unknown channel is $\Lambdadep$, we compute the probability distribution on each leaf $\ell$ as:
\begin{align*}
p^{\Lambdadep}(\ell)=\prod_{t=1}^T\omega_{s_t}^{u_t}.
\end{align*}
In the second scenario where the unknown channel is $\Lambda_a$, we define some notations as
\begin{align*}
\xi_a^{u_t}(m)\coloneqq\begin{cases}
\bra{\phi^{u_t}}P_a\ket{\phi^{u_t}},\quad &m=1\,,\\
2^{-n}\Tr(P_a\cC_{m-1}^{u_t}(I+\epsilon P_a\xi_a^{u_t}(m-1))),\quad &2\leq m\leq M^{u_t}\,.
\end{cases}
\end{align*}
By definition, $\abs{\xi_i^{u_t}(1)}\leq 1$. If we assume inductively that $\abs{\xi_i^{u_t}(m-1)}\leq 1$, we can deduce that
\begin{align*}
\abs{\xi_a^{u_t}(m)}&=\Bigl|\Tr\Bigl(P_a\cC_{m-1}^{u_t}\Bigl(\frac{I+\epsilon P_a\xi_a^{u_t}(m-1)}{2^{n}}\Bigr)\Bigr)\Bigr|\\
&\leq\norm{P_a}_\infty \cdot\Bigl\|\cC_{m-1}^{u_t}\Bigl(\frac{I+\epsilon P_a\xi_a^{u_t}(m-1)}{2^{n}}\Bigr)\Bigr\|_{\Tr}\\
&=\Bigl|\Tr\Bigl(\frac{I+\epsilon P_a\xi_a^{u_t}(m-1)}{2^{n}}\Bigr)\Bigr|\\
&=1\,.
\end{align*}
Here, the first line holds by the recursive expression, the second line holds by the tracial matrix H\"{o}lder inequality, and the third line uses the fact that $\cC_{m-1}^{u_t}$ is a positive trace-preserving map and that $I+\epsilon P_a\xi_i^{u_t}(m-1)$ is positive semidefinite. We then have:
\begin{align*}
\E_a[p^{\Lambda_a}]=\E_a\prod_{t=1}^T\omega_{s_t}^{u_t}\bigl(1+\epsilon\xi_{a}^{u_t}(M^{u_t})\bra{\psi_{s_t}^{u_t}}P_a\ket{\psi_{s_t}^{u_t}}\bigr)\,.
\end{align*}

In the following, we adapt a technique from \cite{chen2022complexity}. Furthermore, in \Cref{app:Shadowlow} we also use this technique to provide an alternative proof of the exponential lower bound of \cite{chen2022exponential} for shadow tomography. 

The final goal is to bound the variation distance $\text{TV}(p_0,\E_a[p_a])$, where we slightly abuse notation and denote $p_0,p_a$ as the probability distributions over leaves under the two scenarios in the many-versus-one distinguishing problem. In each iteration $t$, we denote the concatenating procedure equipped with $M^{u_t}$ unknown channel and data processing channels as $A_t$. Given a path $\z=((u_1,A_1,s_1),\ldots,(u_T,A_T,s_T))$, the likelihood ratio $L_a(\z)$ between traversing the path when the underlying channel is $\Lambda_a$ versus when it is $\Lambdadep$ is given by
\begin{align*}
L_a(\z)&=\prod_{t=1}^T\frac{\omega_{s_t}^{u_t}\bigl(1+\epsilon\xi_a^{u_t}(M^{u_t})\bra{\psi_{s_t}^{u_t}}P_a\ket{\psi_{s_t}^{u_t}}\bigr)}{\omega_{s_t}^{u_t}}\\
&=\prod_{t=1}^T[1+\epsilon\xi_a^{u_t}(M^{u_t})\bra{\psi_{s_t}^{u_t}}P_a\ket{\psi_{s_t}^{u_t}}]\\
&=\prod_{t=1}^T[1+Y_a(u_t,A_t,s_t)],
\end{align*}
where we have defined $Y_a(u_t,A_t,s_t) \triangleq \epsilon\xi_a^{u_t}(M^{u_t})\bra{\psi_{s_t}^{u_t}}P_a\ket{\psi_{s_t}^{u_t}}$. 

We have the following lemma on the moments of this random variable:
\begin{lemma}\label{lem:PropertyYChan}
For any edge $(u,A,s)$ in the tree, we have that
\begin{itemize}
    \item $\E_{s}[Y_a(u,A,s)]=0$;
    \item $\E_{s,a}[Y_a(u,A,s)^2]\leq\frac{\epsilon^2}{2^n+1}$.
\end{itemize}
\end{lemma}
\begin{proof}
The expectation is computed as
\begin{align*}
\E_{s}[Y_a(u,A,s)]&=\sum_s \epsilon\xi_a^{u_t}(M^{u_t})\Tr(\omega_s^u \ket{\psi_2^u}\bra{\psi_s^u}P_a)]\\
&=\epsilon\xi_a^{u_t}(M^{u_t})\sum_{s}\left[\Tr\left(\frac{I}{2^n}P_a\right)\right]\\
&=0\,,
\end{align*}
where the third line follows from the POVM property, i.e., $\sum_s \omega_s^u2^n\ket{\psi_2^u}\bra{\psi_s^u})=I$. Note that this is simply a consequence of the standard fact that likelihood ratios always integrate to 1.

The second moment is bounded by
\begin{align*}
\E_{s,a}[Y_a(u,A,s)^2]&\leq\sum_s \epsilon^2\xi_a^{u_t}(M^{u_t})^2\E_{s,a}[\bra{\psi_{s_t}^{u_t}}P_a\ket{\psi_{s_t}^{u_t}}^2]\\
&\leq\epsilon^2\E_{s,a}[\bra{\psi_{s_t}^{u_t}}P_a\ket{\psi_{s_t}^{u_t}}^2]\\
&=\epsilon^2\sum_s\omega_s^u\E_a[\bra{\psi_{s_t}^{u_t}}P_a\ket{\psi_{s_t}^{u_t}}^2]\\
&\leq\epsilon^2\sum_i\omega_s^u\cdot\frac{1}{2^n+1}\\
&=\frac{\epsilon^2}{2^n+1}\,.
\end{align*}
Here, the second line follows from the fact that $\abs{\xi_a^{u_t}(M^{u_t})}\leq 1$, and the fourth line follows from Lemma~\ref{lem:SumPauliExp}.
\end{proof}
With respect to the distribution over paths $\z$, the distribution on $L_a(\z)$ actually indicates the contribution of this path to the total variation distance between $p_0$ and $\E_{a}[p_a]$. This is because
\begin{align*}
\text{TV}(p_0,\E_{a}[p_a])=\E_{\z\sim p_0}[\max(0,1-\E_a[L_a(\z)])]\leq\E_{\z,a}[\max(0,1-L_a(\z))].
\end{align*}

Next, we introduce the concept of good and balanced paths. Given an edge $(u,A,s)$ and a constant value $\alpha$ to be fixed later, we say the edge is $a$-good if
\begin{align*}
Y_a(u,A,s)\geq-\alpha\,.
\end{align*}
We say that a path $\z$ is $a$-good if all of its constituent edges are $a$-good.

We say a path $\z$ is $(a,\nu)$-balanced if its constituent edges $((u_1,A_1,s_1),\ldots,(u_T,A_T,s_T))$ satisfy
\begin{align*}
\sum_{t=1}^T \E_{s_t}[Y_a(u_t,A_t,s_t)^2]\leq\nu\,.
\end{align*}

Intuitively, the goodness of a path indicates that we never experience any significant multiplicative decreases in the likelihood ratio as we go down a path. The balancedness of a path indicates that the variance of the multiplicative changes in the likelihood ratio are never too large. Now, we can bound the proportion of bad  or imbalanced paths among all possible choices of the index $a$ defining the underlying channel:
\begin{align*}
\Pr_{\z\sim p_0}[\text{proportion of bad paths}]&=\E_a\Pr_{\z}[\z\text{ is not }a\text{-good}]\\
&\leq T\E_a\Pr[(u_t,A_t,s_t)\text{ is not }a\text{-good}]\\
&=\frac{T}{\alpha^2(2^n+1)}.
\end{align*}
Here, the last line follows from Chebyshev's inequality. Thus we obtain the following lemma via Markov's inequality:
\begin{lemma}\label{lem:chanprobgood}
$\Pr_{\z\sim p_0}[{\rm proportion \ of \ good \ paths}\geq 1-\frac{100T}{\alpha^2(2^n+1)}]\leq 0.99\,.$
\end{lemma}

\noindent Similarly, we obtain the following lemma:

\begin{lemma}\label{lem:chanprobbalance}
$\Pr_{\z\sim p_0}[\mathrm{proportion \ of \ balanced \ paths}\geq 1-\frac{100T\epsilon^2}{\nu(2^n+1)}]\leq 0.99\,.$
\end{lemma}
\begin{proof}
\begin{align*}
\Pr_{\z\sim p_0}[\text{proportion of balanced path}&=\E_a[1-\Pr_{\z}[\z\text{ not }a\text{-balanced}]]\\
&\geq\E_a\Bigl[1-\E_{\z}\Bigl[\frac{\sum_{t=1}^T \E_{s_t}[Y_a(u_t,A_t,s_t)^2]}{\nu}\Bigr]\Bigr]\\
&=1-\frac{T\epsilon^2}{\nu(2^n+1)}\,.\qedhere
\end{align*}
\end{proof}

\paragraph{Martingale concentration.} Now, we consider the paths that are both good and balanced. Together we can obtain the following Bernstein-type concentration.
\begin{lemma}\label{lem:chanBernsteinbound}
For any $a\in\{1,\ldots,4^n-1\}$, consider the following sequence of random variables:
\begin{align*}
\bigl\{X_t\coloneqq\log(1+Y_a(u_t,A_t,s_t))\cdot\mathbbm{1}[(u_t,A_t,s_t)\text{ is }a\text{-good}\;]\bigr\}_{t=1}^T
\end{align*}
where the randomness is with respect to $p_0$. For any $\eta,\nu,\epsilon>0$, we have
\begin{align*}
&\Pr_{\z\sim p_0}\left[\sum_{t=1}^TX_t\leq-\left(1+\frac{1}{\alpha}\right)\sum_{t=1}^T \E_{s_t}[Y_a(u_t,A_t,s_t)^2]-\eta\text{ and }\sum_{t=1}^T \E_{s_t}[Y_a(u_t,A_t,s_t)^2]\leq\nu\right ]\\
&\qquad\leq\exp\left(\frac{-\eta^2}{4\nu+2\alpha\eta/3}\right).
\end{align*}
\end{lemma}
\begin{proof}
Notice that for any edge $(u_t,A_t,s_t)$, we have
\begin{align*}
\MoveEqLeft \E_{s_t}[\log(1+Y_a(u_t,A_t,s_t))\cdot\mathbbm{1}[(u_t,A_t,s_t)\text{ is }a\text{-good}]]\\
&\geq\E_{s_t}[(Y_i(u_t,A_t,s_t)-Y_a(s_t,A_t,u_t)^2)\cdot\mathbbm{1}[(u_t,A_t,s_t)\text{ is }a\text{-good}]]\\
&\geq\E_{s_t}[(Y_a(u_t,A_t,s_t)-Y_a(s_t,A_t,u_t)^2)]-\E_{s_t}[Y_a(u_t,A_t,s_t)\cdot\mathbbm{1}[(u_t,A_t,s_t)\text{ is not }a\text{-good}]]\\
&\geq-\E_{s_t}[Y_a(u_t,A_t,s_t)^2]-\E_{s_t}[Y_a(u_t,A_t,s_t)^2]^{1/2}\cdot\Pr[(u_t,A_t,s_t)\text{ is not }a\text{-good}]^{1/2}\\
&=-\Bigl(1+\frac{1}{\alpha}\Bigr)\E_{s_t}[Y_a(u_t,A_t,s_t)^2]\,,
\end{align*}
where in the last two lines we used Cauchy-Schwartz and Chebyshev's inequality. In addition, we have
\begin{align*}
\MoveEqLeft\E_{s_t}[\log(1+Y_a(u_t,A_t,s_t))^2\cdot\mathbbm{1}[(u_t,A_t,s_t)\text{ is }a\text{-good}]]\\
&\leq\E_{s_t}[2Y_i(u_t,A_t,s_t)^2\cdot\mathbbm{1}[(u_t,A_t,s_t)\text{ is }a\text{-good}]]\\
&\leq\E_{s_t}[2Y_i(u_t,s_t)^2]\,.\qedhere
\end{align*}
\end{proof}

Now, we are ready to complete our proof of \Cref{thm:AdaptConLow}. Suppose to the contrary that the protocol given by our learning tree can distinguish between the two scenarios with constant advantage using only
\begin{align*}
T\leq O\left(\frac{2^n}{\epsilon^2}\right)
\end{align*}
measurements.
Without loss of generality, we assume that $\frac{T\epsilon^2}{2^n+1}\leq c$ for some constant $c>0$. We choose the following parameters
\begin{align*}
c=1/(101000),\quad\alpha=10^{-2},\quad\nu=1/(10100),\quad\eta=0.05.
\end{align*}

According to \Cref{lem:chanprobgood}, with probability at least $0.99$, the proportion of good paths with respect to randomly chosen paths $\z\sim p_0$ is at least $0.99$. According to \Cref{lem:chanprobbalance}, with probability at least $0.99$, the proportion of balanced path with respect to $\z\sim p_0$ is at least $0.99$. According to \Cref{lem:chanBernsteinbound}, for all $a$ we have $\sum_{t=1}^TX_t>-0.1$ with probability at least $0.99$. By the definition of $X_t$, we have $L_a(\z)\geq 0.9$ with probability at least $0.99$ with respect to $\z\sim p_0$.

Combining these three results, a randomly chosen path $\z$ is good and balanced with probability at least $0.99^2$. In addition, this path satisfies that $L_a(\z)\geq 0.9$ with probability at least $0.99$ for any $a$. Notice that the variation distance along any remaining paths is bounded by $1$. Thus the total variation distance is bounded by
\begin{align*}
0.99^3\times 0.99^3\times 0.1+(1-0.99^3\times0.99^3)\leq 1/6.
\end{align*}

Thus this algorithm cannot learn the problem with a probability of at least $5/6$. Therefore, any protocols that solve the many-versus-one distinguishing problem with probability at least $5/6$ requires $O(2^n/\epsilon^2)$ measurements. This finishes the proof for \Cref{thm:AdaptConLow}.

\section{Lower bound on query complexity with bounded quantum memory}\label{sec:SampComp}

In this section, we prove Theorem~\ref{thm:AdaptConKLow_informal}, which gives an exponential lower bound on the number of channel queries needed by any learning protocol with $k\le n$ ancilla qubits of quantum memory:
\begin{theorem}\label{thm:AdaptConKLow}
For any (possibly adaptive and concatenating) protocol with $k$ ancilla qubits of the quantum memory that solves Pauli eigenvalue estimation (Problem~\ref{prob:def}) and even easier Pauli spike detection (Problem~\ref{prob:SpikeDetect}) to a constant error with at least large constant probability must make at least $\Omega(2^{(n-k)/3})$ queries to the unknown channel.
\end{theorem}
\noindent Note that unlike in the previous section, the lower bound here is only for the query complexity rather than the measurement complexity. 

We first set some notation. Recall that any adaptive and concatenating strategy with a $k$-qubit memory can be described as follows. Given the unknown Pauli channel $\Lambda$, we perform $T$ rounds of quantum measurement. In the $t$-th round, we input a state $\rho^t$ (possibly adaptively chosen), interleave channel queries with a sequence of $M^t$ adaptive data processing channels $\cC_1^t,\ldots,\cC_{M_t}^t\in\C^{2^{n+k}\times 2^{n+k}}$, and perform an adaptive POVM $\{E^t_j\}_j$ on all $n+k$ qubits. The $j$-th measurement can be obtained with probability:
\begin{align*}
p^t(j|t-1,\ldots,1)=\Tr[E_j^t\Lambda\otimes I(\cC_{M_t}^t(\Lambda\otimes I\cdots\Lambda\otimes I(\cC_1^t(\Lambda\otimes I(\rho^t)))\cdots))],
\end{align*}
where $\Lambda\otimes I$ is the unknown channel acting on the $n$ qubits. Without loss of generality, we assume $\rho^t=\ket{\phi^t}\bra{\phi^t}$ and $\{E^t_j\}_j=\bigl\{\omega_j^t 2^{n+k}\ket{\psi^t_j}\bra{\psi_j^t}\bigr\}_j$.

For the proof of Theorem~\ref{thm:AdaptConKLow}, we follow the formalism of the many-versus-one distinguishing problem introduced in Section~\ref{sec:tree} and consider the Pauli spike detection defined in \Cref{prob:SpikeDetect}, which we rewrote here as distinguishing between the following two scenarios:
\begin{itemize}
    \item The unknown channel is the depolarization channel $\Lambdadep=\frac{1}{2^n}I\Tr(\cdot)$.
    \item The unknown channel is sampled uniformly from a set of channels $\{\Lambda_a\}_{a\in\Z_2^{2n}\backslash \{0\}}$ with $\Lambda_a=\frac{1}{2^n}[I\Tr(\cdot)+ P_a\Tr(P_a(\cdot))]$.
\end{itemize}
It is straightforward to observe that the learning strategy can accomplish this task with high probability. Let the probability distribution of the measurement outcomes under these two cases be $p_1$ and $p_2$, respectively. The total variation distance TV$(p_1,p_2)$ must be bounded below by $O(1)$ for the distinguishing task to succeed with high probability.

Now, we begin to prove \Cref{thm:AdaptConKLow}. Without loss of generality, we consider an algorithm of $T$ adaptive, concatenating measurements. In the $t$-th turn, we assume that the algorithm concatenates $M_t$ channels before the measurement. Therefore, the total number of $\Lambda$ used is
\begin{align*}
N=\sum_{t=1}^TM_t.
\end{align*}

In the $t$-th turn, we denote
\begin{align*}
p_0^t(j)=\Tr[E_j^t\Lambdadep\otimes I(\cC_{M_t}^t(\Lambdadep\otimes I\cdots\Lambdadep\otimes I(\cC_1^t(\Lambdadep\otimes I(\ket{\phi^t}\bra{\phi^t})))\cdots))],
\end{align*}
which is the probability of obtaining the result $j$ in the first case. We further define some intermediate probability
\begin{align*}
p_m^t(j)=\Tr[E_j^t\underbrace{\Lambda_a\otimes I(\cC_{M_t}^t(\Lambda_a\otimes I}_{m_t\text{ channels}}\cdots\underbrace{\Lambdadep\otimes I(\cC_1^t(\Lambdadep\otimes I}_{M_t-m_t+1\text{ channels}}(\ket{\phi^t}\bra{\phi^t})))\cdots))],
\end{align*}
where the last $m$ channels are replaced by $\Lambda_a$. We define the procedure in this turn the $m_t$-intermediate procedure. Therefore, we have $p_{M_t+1}^t(j)$ the probability of obtaining result $j$ in the $t$-th turn:
\begin{align*}
p_{M_t+1}^t(j)=\Tr[E_j^t\Lambda_a\otimes I(\cC_{M_t}^t(\Lambda_a\otimes I\cdots\Lambda_a\otimes I(\cC_1^t(\Lambda_a\otimes I(\ket{\phi^t}\bra{\phi^t})))\cdots))].
\end{align*}

Now, we consider the whole algorithm, the probability of obtaining a result $\bm{j}=j_1,\ldots,j_T$ with $m_t$-intermediate procedures is represented by
\begin{align*}
p_{\bm{m}}(\bm{j})=\E_a \prod_{t=1}^Tp_{m_t}^t(j_t),
\end{align*}
where $\bm{m}=m_1,\ldots,m_T$. In the first case, the probability of obtaining the result $\bm{j}$ is
\begin{align*}
p_1(\bm{j})=p_{\bm{0}}(\bm{j}),
\end{align*}
(actually independent from $a$) while in the second case, the probability of obtaining result $\bm{j}$ is
\begin{align*}
p_2(\bm{j})=p_{\bm{M}}(\bm{j}),
\end{align*}
where $\bm{M}=M_1+1,\ldots,M_T+1$. According to the triangular inequality, the total variation distance between $p_1$ and $p_2$ is bounded by
\begin{align}\label{eq:TriangleDecomp}
\text{TV}(p_1,p_2)\leq\sum_{t=1}^T\sum_{m_t=0}^{M_t}\text{TV}(p_{(M_1+1)\ldots(M_{t-1}+1)m_t0\ldots 0},p_{(M_1+1)\ldots(M_{t-1}+1)(m_t+1)0\ldots 0}).
\end{align}
By the above inequality, we decompose the total variation distance between $p_1$ and $p_2$ into $O(N+T)=O(N)$ terms, each containing the total variation distance between two algorithms with \textit{exactly} one different channel $\Lambdadep$ and $\Lambda_a$. In the following, we consider each term
\begin{align*}
\text{TV}(p_{M_1\ldots M_{t-1}m_t0\ldots 0},p_{M_1\ldots M_{t-1}(m_t+1)0\ldots 0}).
\end{align*}
In particular, we prove the following two lemmas.
\begin{lemma}
For any $t$, we have $\text{TV}(p_{(M_1+1)\ldots(M_{t-1}+1)00\ldots 0},p_{(M_1+1)\ldots(M_{t-1}+1)10\ldots 0})\leq 2\left(\frac{2^k}{2^n-1}\right)^{1/3}$.
\end{lemma}
\begin{proof}
For simplicity, we denote $p_{(M_1+1)\ldots(M_{t-1}+1)00\ldots 0}=p'$ and $p_{(M_1+1)\ldots(M_{t-1}+1)10\ldots 0}=p''$. We also denote
\begin{align*}
\rho_{\phi}=\cC_{M_t}^t(\Lambda_a\otimes I\cdots\Lambda_a\otimes I(\cC_1^t(\Lambda_a\otimes I(\ket{\phi^t}\bra{\phi^t})))\cdots)
\end{align*}
It is straightforward to observe that
\begin{align*}
\text{TV}(p',p'')&=\sum_{\bm{j}:p'(\bm{j})>p''(\bm{j})}(p'(\bm{j})-p''(\bm{j}))\\
&=\E_a\sum_{\bm{j}:p'(\bm{j})>p''(\bm{j})}\left(\prod_{t'=1}^{t-1}p^{t'}_{M_{t'}+1}(j_{t'})p_0^{t}(j_{t})\prod_{t''=t+1}^{T}p^{t''}_{0}(j_{t''})-\prod_{t'=1}^{t-1}p^{t'}_{M_{t'}+1}(j_{t'})p_1^{t}(j_{t})\prod_{t''=t+1}^{T}p^{t''}_{0}(j_{t''})\right)\\
&=\E_a\sum_{\bm{j}:p'(\bm{j})>p''(\bm{j})}p'(\bm{j})\left(1-\frac{p_1^{t}(j_{t})}{p_0^{t}(j_{t})}\right).
\end{align*}
We further observe that
\begin{align*}
1-\frac{p_1^{t}(j_{t})}{p_0^{t}(j_{t})}=1-\frac{\Tr(E_{j_t}^t\Lambda_a\otimes I(\rho_\phi))}{\Tr(E_{j_t}^t\Lambdadep\otimes I(\rho_\phi))}.
\end{align*}
Due to the convexity of quantum states and channels, we can regard $\rho_\phi$ as a pure state. Thus, we still denote $\rho_\phi=\ket{\phi}\bra{\phi}$. The above term can be written as:
\begin{align*}
1-\frac{p_1^{t}(j_{t})}{p_0^{t}(j_{t})}=-\frac{\sum_{b=0}^{4^k}\bra{\psi_{j_t}^t}P_a\otimes P_b\ket{\psi_{j_t}^t}\bra{\phi}P_a\otimes P_b\ket{\phi}}{\sum_{b=0}^{4^k}\bra{\psi_{j_t}^t}I\otimes P_b\ket{\psi_{j_t}^t}\bra{\phi}I\otimes P_b\ket{\phi}}.
\end{align*}
We denote $A,B_j\in\C^{2^n\times 2^k}$ such that
\begin{align*}
\ket{\phi}=\sum_{p=0}^{2^n-1}\sum_{q=0}^{2^k-1}\bra{p}A\ket{q}\ket{p}\ket{q},\qquad\ket{\psi_j^t}=\sum_{p=0}^{2^n-1}\sum_{q=0}^{2^k-1}\bra{p}B_j\ket{q}\ket{p}\ket{q}.
\end{align*}
According to the normalization condition, we have $\Tr(B_j^\dagger B_j)=\Tr(A^\dagger A)$. We further denote $C_j=B_jA^\dagger$ a $2^n\times 2^n$ matrix of rank at most $2^k$. Thus, we have
\begin{align*}
1-\frac{p_1^{t}(j_{t})}{p_0^{t}(j_{t})}&=-\frac{\sum_{b=0}^{4^k}\Tr(B_j^\dagger P_aB_j P_b^T)\Tr(A^\dagger P_aA P_b^T)}{\sum_{b=0}^{4^k}\Tr(B_j^\dagger B_j P_b^T)\Tr(A^\dagger A P_b^T)}\\
&=-\frac{\Tr(B_j^\dagger P_aB_jA^\dagger P_a A)}{\Tr(B_j^\dagger B_jA^\dagger A)}\\
&=-\frac{\Tr(P_aC_j^\dagger P_aC_j)}{\Tr(C_j^\dagger C_j)},
\end{align*}
where the second line follows from the fact that $\mathbb{I}(\cdot)=2^{-k}\sum_{b=0}^{4^k}P_b\Tr(P_b(\cdot))$. Notice that $\Tr(P_aC_j^\dagger P_aC_j)^2$\\$\leq 2^k \Tr(C_j C_j^\dagger P_a C_j C_j^\dagger P_a)$ according to the Cauchy-Schwarz inequality, the second moment can be calculated as:
\begin{align*}
\sum_{a=1}^{4^n-1}\left(\frac{\Tr(P_aC_j^\dagger P_aC_j)}{\Tr(C_j^\dagger C_j)}\right)^2&\leq\sum_{a=1}^{4^n-1}\frac{2^k\Tr(C_j^\dagger C_j P_a C_jC_J^\dagger P_a)}{\Tr(C_j^\dagger C_j)}\\
&=2^k\frac{2^n\Tr(C_j^\dagger C_j)-\Tr(C^\dagger_j C_jC_j^\dagger C_j)}{\Tr(C_j^\dagger C_j)}\\
&\leq 2^k2^n\,.\qedhere
\end{align*}
Therefore, there can be at most $\left(\frac{2^k}{2^n-1}\right)^{1/3}(4^n-1)$ choices of $a$ such that
\begin{align*}
\abs{\frac{\Tr(P_aC_j^\dagger P_aC_j)}{\Tr(C_j^\dagger C_j)}}\geq\left(\frac{2^k}{2^n-1}\right)^{1/3}.
\end{align*}
We have
\begin{align*}
\text{TV}(p',p'')&=\E_a\sum_{\bm{j}:p'(\bm{j})>p''(\bm{j})}p'(\bm{j})\Bigl(1-\frac{p_1^{t}(j_{t})}{p_0^{t}(j_{t})}\Bigr)\\
&=\frac{1}{4^n-1}\sum_{\bm{j}:p'(\bm{j})>p''(\bm{j})}p'(\bm{j}) \cdot \sum_{a=1}^{4^n-1}\Bigl(1-\frac{p_1^{t}(j_{t})}{p_0^{t}(j_{t})}\Bigr)\\
&\leq\sum_{\bm{j}:p'(\bm{j})>p''(\bm{j})}p'(\bm{j})\Bigl[\Bigl(\frac{2^k}{2^n-1}\Bigr)^{1/3}+\Bigl(1-\Bigl(\frac{2^k}{2^n-1}\Bigr)^{1/3}\Bigr)\Bigl(\frac{2^k}{2^n-1}\Bigr)^{1/3}\Bigr]\\
&\leq 2\Bigl(\frac{2^k}{2^n-1}\Bigr)^{1/3}\,.
\end{align*}
\end{proof}

\noindent We also have the following lemma
\begin{lemma}
For any $t$ and $0<m_t\leq M_t$, we have 
\begin{align*}
    \text{TV}(p_{(M_1+1)\ldots(M_{t-1}+1)m_t0\ldots 0},p_{(M_1+1)\ldots(M_{t-1}+1)(m_t+1)0\ldots 0})\leq 2\Bigl(\frac{2^k}{2^n-1}\Bigr)^{1/3}\,.
\end{align*}
\end{lemma}
\begin{proof}
For simplicity, we also denote $p_{(M_1+1)\ldots(M_{t-1}+1)m_t0\ldots 0}=p'$ and $p_{(M_1+1)\ldots(M_{t-1}+1)(m_t+1)0\ldots 0}=p''$. We can observe that
\begin{align*}
\text{TV}(p',p'')&=\sum_{\bm{j}:p'(\bm{j})>p''(\bm{j})}(p'(\bm{j})-p''(\bm{j}))\\
&=\E_a\sum_{\bm{j}:p'(\bm{j})>p''(\bm{j})}\left(\prod_{t'=1}^{t-1}p^{t'}_{M_{t'}+1}(j_{t'})p_{m_t}^{t}(j_{t})\prod_{t''=t+1}^{T}p^{t''}_{0}(j_{t''})-\prod_{t'=1}^{t-1}p^{t'}_{M_{t'}+1}(j_{t'})p_{m_t+1}^{t}(j_{t})\prod_{t''=t+1}^{T}p^{t''}_{0}(j_{t''})\right)\\
&=\E_a\sum_{\bm{j}:p'(\bm{j})>p''(\bm{j})}p'(\bm{j})\left(1-\frac{p_{m_t+1}^{t}(j_{t})}{p_{m_t}^{t}(j_{t})}\right).
\end{align*}
In particular, we have
\begin{align*}
1-\frac{p_{m_t+1}^{t}(j_{t})}{p_{m_t}^{t}(j_{t})}=1-\frac{\Tr[E_j^t\underbrace{\Lambda_a\otimes I(\cC_{M_t}^t(\Lambda_a\otimes I}_{m_t+1\text{ channels}}\cdots\underbrace{\Lambdadep\otimes I(\cC_1^t(\Lambdadep\otimes I}_{M_t-m_t\text{ channels}}(\ket{\phi^t}\bra{\phi^t})))\cdots))]}{\Tr[E_j^t\underbrace{\Lambda_a\otimes I(\cC_{M_t}^t(\Lambda_a\otimes I}_{m_t\text{ channels}}\cdots\underbrace{\Lambdadep\otimes I(\cC_1^t(\Lambdadep\otimes I}_{M_t-m_t+1\text{ channels}}(\ket{\phi^t}\bra{\phi^t})))\cdots))]}.
\end{align*}
Without loss of generality, we denote
\begin{align*}
\rho_\phi=\underbrace{\Lambdadep\otimes I(\cC_1^t(\Lambdadep\otimes I}_{M_t-m_t\text{ channels}}(\ket{\phi^t}\bra{\phi^t}))),\quad\cC=\underbrace{\Lambda_a\otimes I(\cC_{M_t}^t\cdots(\Lambda_a\otimes I}_{m_t\text{ channels}}(\cdot))\cdots).
\end{align*}
Therefore, we have
\begin{align*}
1-\frac{p_{m_t+1}^{t}(j_{t})}{p_{m_t}^{t}(j_{t})}=1-\frac{\Tr(E_{j_t}^t\cC(\Lambda_a\otimes I(\rho_\phi)))}{\Tr(E_{j_t}^t\cC(\Lambdadep\otimes I(\rho_\phi)))}.
\end{align*}

When summing over all the $\bm{j}$'s to compute the total variation distance, we are also summing up all $j_t$. 
We compute the total variation distance between $p'$ and $p''$ as:
\begin{align*}
\text{TV}(p',p'')&=\E_a\sum_{\bm{j}:p'(\bm{j})>p''(\bm{j})}p'(\bm{j})\Bigl(1-\frac{p_{m_t+1}^{t}(j_{t})}{p_{m_t}^{t}(j_{t})}\Bigr)\\
&=\E_a\sum_{\bm{j}:p'(\bm{j})>p''(\bm{j})}p'(\bm{j})\Bigl(1-\frac{\Tr(E_{j_t}^t\cC(\Lambda_a\otimes I(\rho_\phi)))}{\Tr(E_{j_t}^t\cC(\Lambdadep\otimes I(\rho_\phi)))}\Bigr)\\
\end{align*}
Without loss of generality, we assume that $\cC(\rho)=\sum_{k=1}^K M_k^\dagger\rho M_k$ for any quantum state $\rho$, where $\rho$ is an arbitrary quantum state. Suppose the unitary in the expanded Hilbert space for $\cC$ according to the Stinespring dilation theorem~\cite{stinespring1955positive} is $U$, we consider the channel corresponding to $U^\dagger$ in the expanded Hilbert space. The quantum channel corresponding to this unitary is exactly $\cC'=\sum_{k=1}^K M_k\rho M_k^\dagger$. We notice that
\begin{align*}
\frac{\Tr(E_{j_t}^t\cC(\Lambda_a\otimes I(\rho_\phi)))}{\Tr(E_{j_t}^t\cC(\Lambdadep\otimes I(\rho_\phi)))}&=\frac{\Tr\left(E_{j_t}^t\sum_{k=1}^KM_k^\dagger(\Lambda_a\otimes I(\rho_\phi))M_k\right)}{\Tr\left(E_{j_t}^t\sum_{k=1}^KM_k^\dagger(\Lambdadep\otimes I(\rho_\phi))M_k\right)}\\
&=\frac{\Tr\left(\sum_{k=1}^KM_kE_{j_t}^tM_k^\dagger(\Lambda_a\otimes I(\rho_\phi))\right)}{\Tr\left(\sum_{k=1}^KM_kE_{j_t}^tM_k^\dagger(\Lambdadep\otimes I(\rho_\phi))\right)}\\
&=\frac{\Tr\left(\cC'\left(\ket{\psi_{j_t}^t}\bra{\psi_{j_t}^t}\right)\Lambda_a\otimes I(\rho_\phi)\right)}{\Tr\left(\cC'\left(\ket{\psi_{j_t}^t}\bra{\psi_{j_t}^t}\right)\Lambdadep\otimes I(\rho_\phi)\right)},
\end{align*}
where the last line follows from the definition of the POVM. We denote $\cC'\left(\ket{\psi_{j_t}^t}\bra{\psi_{j_t}^t}\right)$ as $\rho_{\psi}$. Due to the convexity of quantum states and channels, we can regard $\rho_{\psi}$ and $\rho_\phi$ as a pure state. Thus, we still denote $\rho_{\psi}=\ket{\psi_{j_t}^t}\bra{\psi_{j_t}^t}$ and $\rho_\phi=\ket{\phi}\bra{\phi}$. Then this problem reduces exactly to the same case as the previous lemma. We still use the definition of $A$, $B_j$'s, and $C$. According to the proof for the previous lemma, we have
\begin{align*}
1-\frac{\Tr(E_{j_t}^t\Lambda_a\otimes I(\rho_\phi))}{\Tr(E_{j_t}^t\Lambdadep\otimes I(\rho_\phi))}=-\frac{\Tr(P_aC_j^\dagger P_aC_j)}{\Tr(C_j^\dagger C_j)}.
\end{align*}
Thus, there can be at most $\left(\frac{2^k}{2^n-1}\right)^{1/3}(4^n-1)$ choices of $a$ such that
\begin{align*}
\Bigl|\frac{\Tr(P_aC_j^\dagger P_aC_j)}{\Tr(C_j^\dagger C_j)}\Bigr|\geq\Bigl(\frac{2^k}{2^n-1}\Bigr)^{1/3}.
\end{align*}
Therefore, we have
\begin{align*}
\text{TV}(p',p'')&=\E_a\sum_{\bm{j}:p'(\bm{j})>p''(\bm{j})}p'(\bm{j})\Bigl(1-\frac{p_1^{t}(j_{t})}{p_0^{t}(j_{t})}\Bigr)\\
&=\frac{1}{4^n-1}\sum_{\bm{j}:p'(\bm{j})>p''(\bm{j})}p'(\bm{j})\cdot\sum_{a=1}^{4^n-1}\Bigl(1-\frac{p_1^{t}(j_{t})}{p_0^{t}(j_{t})}\Bigr)\\
&\leq\sum_{\bm{j}:p'(\bm{j})>p''(\bm{j})}p'(\bm{j})\Bigl(\Bigl(\frac{2^k}{2^n-1}\Bigr)^{1/3}+\Bigl(1-\Bigl(\frac{2^k}{2^n-1}\Bigr)^{1/3}\Bigr)\Bigl(\frac{2^k}{2^n-1}\Bigr)^{1/3}\Bigr)\\
&\leq 2\Bigl(\frac{2^k}{2^n-1}\Bigr)^{1/3}.\qedhere
\end{align*}
\end{proof}

Combining the above two lemmas and the triangle inequality in \eqref{eq:TriangleDecomp}, we have
\begin{align*}
\text{TV}(p_1,p_2)\leq O(N)\cdot\Bigl(\frac{2^k}{2^n-1}\Bigr)^{1/3}.
\end{align*}
Notice that the maximal success probability of distinguishing two probability distributions is given by $\frac12(1+\text{TV}(p_1,p_2))$, we prove the lower bound on the number of $\Lambda$ required for learning Pauli channel using adaptive and concatenating strategy, and $k$-qubit quantum memory is
\begin{align*}
N=\Omega\bigl(2^{(n-k)/3}\bigr).
\end{align*}
This finishes the proof for \Cref{thm:AdaptConKLow}.

\section*{Acknowledgements}
We thank Yuan-Zhuo Wang, Yi-Ran Xiao, Ming-Yang Li, Shengjun Wu, and Zeng-Bing Chen for finding and fixing an issue in our proof for \Cref{lem:ThresholdSearch}, and solving one of our open questions by reducing the number of ancilla qubits for learning eigenvalues of Pauli channels with polynomial measurement complexity to a constant~\cite{wang2025weakly}. We thank Matthias Caro, Senrui Chen, Jordan Cotler, Dong-Ling Deng, Hsin-Yuan Huang, Tongyang Li, Qi Ye, and Sisi Zhou for helpful discussions. This work was supported in part by NSF Award 2430375.

\bibliographystyle{myhamsplain}
\bibliography{QChanEstBoundMem}
\clearpage
\appendix

\section{An alternative proof of a lower bound for shadow tomography}\label{app:Shadowlow}
We consider the task of learning the expectation value of observables $O_1,\ldots,O_M$ for an unknown quantum state $\rho$, i.e., approximating $\Tr(O_i\rho)$ for all $i=1,\ldots,M$. Throughout, we assume that $\norm{O_i}_{\sf op} \le 1$. 

We still consider a tree representation of the (adaptive) learning procedure as in the main text. 
\begin{definition}[Tree representation for (adaptive) state learning algorithms]Fix an unknown $n$-qubit quantum state $\rho$, a learning algorithm without quantum memory can be expressed as a rooted tree $\cT$ of depth $T$. Each node on the tree encodes the measurement outcomes the algorithm has seen so far. The tree satisfies:
\begin{itemize}
    \item Each node $u$ is associated with a probability $p^\rho(u)$.
    \item For the root $r$ of the tree, $p^\rho(r)=1$.
    \item At each non-leaf node $u$, we measure an adaptive POVM $\cM_s^u=\{\omega_s^u 2^n\ket{\psi_s^u}\bra{\psi_s^u}\}_s$ on an adaptive input state $\ket{\phi^u}\bra{\phi^u}$ on one $\rho$ to obtain a classical outcome $s$. Each child node $v$ of the node $u$ is connected through the edge $e_{u,s}$.
    \item If $v$ is the child node of $u$ connected through the edge $e_{u,s}$, then
    \begin{align*}
    p^\rho(v)=p^\rho(u)\omega_s^u2^n\bra{\psi_s^u}\rho\ket{\psi_s^u}.
    \end{align*}
    \item Every root-to-leaf path is of length $T$. Note that for a leaf node $\ell$, $p^\rho(\ell)$ is the probability that the classical memory is in state $\ell$ after the learning procedure.
\end{itemize}
\end{definition}

\noindent Again, due to the convexity of linear combination and the fact that we never consider post-measurement quantum states, we can assume all the POVMs are rank-$1$ elements. We can also assume each input state to be a pure state. Because we consider $T$ measurements, all the leaf nodes are at depth $T$. 

For a shadow tomography problem fixing $M$ observable $O_1,\ldots,O_M$, we consider a reduction from the prediction task we care about to a two-hypothesis distinguishing problem. Consider the following $M+1$ states
\begin{itemize}
    \item The unknown state is the maximal mixed state $\rho_0=\frac{1}{2^n}I$.
    \item The unknown channel is sampled uniformly from a set of states $\{\rho_i=\frac{I+\epsilon O_i}{2^n}\}_{i=1}^M$.
\end{itemize}

The goal of the learning algorithm is to predict which event has happened. In this tree representation, in order to distinguish between the two events, the distribution over the classical memory state in the two events must be sufficiently distinct. 

In particular, we can prove the following strengthening of Theorem 5.5 from~\cite{chen2022exponential} for the above distinguishing task.

\begin{theorem}\label{thm:shadow}
The following holds for any $0<\epsilon\leq1$. Given $M$ \textit{traceless} observables $O_1,\ldots,O_M$, any learning algorithm without quantum memory requires
\begin{align*}
T\geq\frac{1}{\epsilon^2\delta(O_1,\ldots,O_M)}
\end{align*}
copies of $\rho$, where $\delta(O_1,\ldots,O_M)=\frac{1}{M}\max_{\ket{\phi}}\sum_{i=1}^M\bra{\phi}O_i\ket{\phi}^2$, to predict expectation values of $\Tr(O_i\rho)$ within at most $\pm\epsilon$ error for all $i\in\{1,\ldots,M\}$ with a high probability.
\end{theorem}

\noindent Although this theorem has a similar form to Theorem 5.5 in~\cite{chen2022exponential}, we emphasize that it is stronger in the following sense. In~\cite{chen2022exponential}, the observables $O_1,\ldots,O_M$ had to be closed under negation, and the theorem only applied to $\epsilon$ which is upper bounded by some absolute constant. In contrast, our theorem applies to an arbitrary collection of observables $O_1,\ldots,O_M$ and any $\epsilon$ in the range $[0,1)$.

To prove Theorem~\ref{thm:shadow}, we adapt the sub-martingale technique from \cite{chen2022complexity}. Our goal is to bound $\text{TV}(p_0,\E_{i}[p_i])$, where $p_0$ and $p_a$ are the probability distribution on all leaves corresponding to the cases $\rho_0$ and $\rho_i$. Given a path $\z=((u_1,s_1),\ldots,(u_T,s_T))$, the likelihood ratio $L_i(\z)$ between traversing that path when the underlying state is $\rho_i$ versus when it is $\rho_0$ is given by
\begin{align*}
L_i(\z)&=\prod_{t=1}^T\frac{\omega_{s_t}^{u_t} 2^n\bra{\psi_{s_t}^{u_t}}\rho_i\ket{\psi_{s_t}^{u_t}}}{\omega_{s_t}^{u_t} 2^n\bra{\psi_{s_t}^{u_t}}\rho_0\ket{\psi_{s_t}^{u_t}}}\\
&=\prod_{t=1}^T\left[1+\epsilon\bra{\psi_{s_t}^{u_t}}O_i\ket{\psi_{s_t}^{u_t}}\right]\\
&=\prod_{t=1}^T(1+Y_i(u_t,s_t)),
\end{align*}
where $Y_i$ is defined as
\begin{align*}
Y_i(u_t,s_t)=\epsilon\bra{\psi_{s_t}^{u_t}}O_i\ket{\psi_{s_t}^{u_t}}.
\end{align*}
With respect to the distribution over paths $\z$, the distribution on $L_i(\z)$ actually indicates the total variation distance. This is because
\begin{align*}
\text{TV}(p_0,\E_{i}[p_i])=\E_{\z\sim p_0}[\max(0,1-\E_i[L_i(\z)])]\leq\E_{\z,i}[\max(0,1-L_i(\z))].
\end{align*}

Regarding the variable $Y_i(u_t,s_t)$, we have the following lemma
\begin{lemma}\label{lem:PropertyY}
For any edge $(u,s)$ in the tree, we have that
\begin{itemize}
    \item $\E_{s}[Y_i(u,s)]=0$;
    \item $\E_{s,i}[Y_i(u,s)^2]\leq\epsilon^2\delta(O_1,\ldots,O_M)$.
\end{itemize}
\end{lemma}
\begin{proof}
The first-order average is computed as
\begin{align*}
\E_{s,i}[Y_i(u,s)]&=\sum_s\Tr(\omega_s^u \ket{\psi_2^u}\bra{\psi_s^u}O_i)]\\
&=\Tr\left(\frac{I}{2^n}O_i\right)\\
&=0,
\end{align*}
where the third line follows from the POVM property, i.e., $\sum_s \omega_s^u2^n\ket{\psi_2^u}\bra{\psi_s^u})=I$.

The second-order average is bounded by
\begin{align*}
\E_{s,i}[Y_i(u,s)^2]&\leq\epsilon^2\E_{s,i}[\bra{\psi_{s_t}^{u_t}}O_i\ket{\psi_{s_t}^{u_t}}^2]\\
&=\epsilon^2\sum_s\omega_s^u\E_i[\bra{\psi_{s_t}^{u_t}}O_i\ket{\psi_{s_t}^{u_t}}^2]\\
&\leq\epsilon^2\sum_i\omega_s^u\delta(O_1,\ldots,O_M)\\
&=\epsilon^2\delta(O_1,\ldots,O_M).
\end{align*}
\end{proof}

Recall the concept of good and balanced paths in the main text. Given an edge $(u,s)$ and a constant value $\alpha$ to be fixed later, we it is $i$-good if
\begin{align*}
Y_i(u,s)\geq-\alpha.
\end{align*}
We say that a path $\z$ is $i$-good if all of its constituent edges are $i$-good.

We say a path $\z$ is $(i,\nu)$-balanced if its constituent edges $((u_1,s_1),\ldots,(u_T,s_T))$ satisfy
\begin{align*}
\sum_{t=1}^T \E_{s_t}[Y_i(u_t,s_t)^2]\leq\nu
\end{align*}

Similar to the main text, we bound the probability of a bad path or an imbalanced path:
\begin{align*}
\Pr_{\z\sim p_0}[\text{proportion of bad path}]&=\E_i\Pr_{\z}[\z\text{ is not }i\text{ good}]\\
&\leq T\E_i\Pr[(u_t,s_t)\text{ is not }i\text{ good}]\\
&=\frac{T\epsilon^2\delta(O_1,\ldots,O_M)}{\alpha^2}.
\end{align*}
Here, the last line follows from Chebyshev's inequality. Thus we obtain the following lemma via Markov's inequality
\begin{lemma}\label{lem:probgood}
$\Pr_{\z\sim p_0}[\text{proportion of good path}\geq 1-\frac{100T\epsilon^2\delta(O_1,\ldots,O_M)}{\alpha^2}]\leq 0.99$
\end{lemma}

\noindent Similarly, we obtain the following lemma
\begin{lemma}\label{lem:probbalance}
$\Pr_{\z\sim p_0}[\text{proportion of balanced path}\geq 1-\frac{100T\epsilon^2\delta(O_1,\ldots,O_M)}{\nu}]\leq 0.99$
\end{lemma}
\begin{proof}
\begin{align*}
\Pr_{\z\sim p_0}[\text{proportion of balanced path}&=\E_i[1-\Pr_{\z}[\z\text{ not }i-\text{balanced}]]\\
&\geq\E_i\left[1-\E_{\z}\left[\frac{\sum_{t=1}^T \E_{s_t}[Y_i(u_t,s_t)^2]}{\nu}\right]\right]\\
&=1-\frac{T\epsilon^2\delta(O_1,\ldots,O_M)}{\nu}\,.\qedhere
\end{align*}
\end{proof}

\noindent\paragraph{Martingale concentration.} Now, we consider the paths that are both good and balanced. Together we can obtain the following Bernstein-type concentration.
\begin{lemma}\label{lem:Bernsteinbound}
For any $i\in\{1,\ldots,M\}$, consider the following sequence of random variables:
\begin{align*}
\left\{X_t\coloneqq\log(1+Y_i(u_t,s_t))\cdot\mathbbm{1}[(u_t,s_t)\text{ is }i-\text{good}]\right\}_{t=1}^T
\end{align*}
where the randomness is with respect to $p_0$. For any $\eta,\nu,\epsilon>0$, we have
\begin{align*}
\Pr_{\z\sim p_0}\left[\sum_{t=1}^TX_t\leq-(1+\frac{1}{\alpha})\sum_{t=1}^T \E_{s_t}[Y_i(u_t,s_t)^2]-\eta\text{ and }\sum_{t=1}^T \E_{s_t}[Y_i(u_t,s_t)^2]\leq\nu\right ]\leq\exp\left(\frac{-\eta^2}{4\nu+2\alpha\eta/3}\right).
\end{align*}
\end{lemma}
\begin{proof}
Notice that for any edge $(u_t,s_t)$, we have
\begin{align*}
\MoveEqLeft\E_{s_t}[\log(1+Y_i(u_t,s_t))\cdot\mathbbm{1}[(u_t,s_t)\text{ is }i-\text{good}]]\\
&\geq\E_{s_t}[(Y_i(u_t,s_t)-Y_i(s_t,u_t)^2)\cdot\mathbbm{1}[(u_t,s_t)\text{ is }i-\text{good}]]\\
&\geq\E_{s_t}[(Y_i(u_t,s_t)-Y_i(s_t,u_t)^2)]-\E_{s_t}[Y_i(u_t,s_t)\cdot\mathbbm{1}[(u_t,s_t)\text{ is not }i-\text{good}]]\\
&\geq-\E_{s_t}[Y_i(u_t,s_t)^2]-\E_{s_t}[Y_i(u_t,s_t)^2]^{1/2}\cdot\Pr[(u_t,s_t)\text{ is not }i-\text{good}]^{1/2}\\
&=-\left(1+\frac{1}{\alpha}\right)\E_{s_t}[Y_i(u_t,s_t)^2]\,,
\end{align*}
where in the last two lines we use the Cauchy-Schwartz inequality and the Chebyshev inequality. In addition, we have
\begin{align*}
\MoveEqLeft\E_{s_t}[\log(1+Y_i(u_t,s_t))^2\cdot\mathbbm{1}[(u_t,s_t)\text{ is }i-\text{good}]]\\
&\leq\E_{s_t}[2Y_i(u_t,s_t)^2\cdot\mathbbm{1}[(u_t,s_t)\text{ is }i-\text{good}]]\\
&\leq\E_{s_t}[2Y_i(u_t,s_t)^2]\,.\qedhere
\end{align*}
\end{proof}

As we are to prove a lower bound, we inversely assume that
\begin{align*}
T\leq O\left(\frac{1}{\epsilon^2\delta(O_1,\ldots,O_M)}\right).
\end{align*}
Without loss of generality, we assume that $T\epsilon^2\delta(O_1,\ldots,O_M)\leq c$ for some constant $c>0$. We choose the following parameters
\begin{align*}
c=1/(101000),\quad\alpha=10^{-2},\quad\nu=1/(10100),\quad\eta=0.05.
\end{align*}

According to \Cref{lem:probgood}, with probability at least $0.99$, the proportion of good paths with respect to randomly chosen paths $\z\sim p_0$ is at least $0.99$. According to \Cref{lem:probbalance}, with probability at least $0.99$, the proportion of balanced path with respect to $\z\sim p_0$ is at least $0.99$. According to \Cref{lem:Bernsteinbound}, for all $a$ we have $\sum_{t=1}^TX_t>-0.1$ with probability at least $0.99$. By the definition of $X_t$, we have $L_a(\z)\geq 0.9$ with probability at least $0.99$ with respect to $\z\sim p_0$.

Following the same derivation with in the main text, the total variation distance is bounded by
\begin{align*}
0.99^3\times 0.99^3\times 0.1+(1-0.99^3\times0.99^3)\leq 1/6.
\end{align*}

Thus this algorithm cannot learn the problem with a probability of at least $5/6$. Similar to the argument in the main text, we obtain that any algorithm that succeeds with probability at least $5/6$ requires $T\geq \Omega\left(\frac{1}{\epsilon^2\delta(O_1,\ldots,O_M)}\right)$ samples.

\clearpage

\end{document}